\documentclass[12pt]{article}
\usepackage[english]{babel}
\usepackage{pgf, amsmath, amsfonts, amsthm,amsbsy, mathrsfs, latexsym, multicol}
\usepackage{srcltx, hyperref, stackrel,graphics,multicol}
 \usepackage[rightcaption]{sidecap}
\usepackage[margin=2.5cm]{geometry}
 \usepackage{soul}

\newcommand{\epsi}{\varepsilon}
\newcommand{\E}{{\mathrm{e}}}
\newcommand{\I}{\mathrm{i}}
 \newcommand{\R}{ \mathbb{R} }
\newcommand{\C}{ \mathbb{C} }
\newcommand{\N}{ \mathbb{N} }
\newcommand{\Z}{ \mathbb{Z} }
\newcommand{\D}{\mathrm{d}}
\newcommand{\Or}{{\mathcal{O}}}

\newcommand{\tr}{{\mathrm{tr}}}
\newcommand{\Hi}{{\mathfrak{H}}}
\newcommand{\bm}{\begin{pmatrix}}
\newcommand{\Em}{\end{pmatrix}}
  
\newcommand{\ad}{\mathscr{L}}
\newcommand{\add}{\mathrm{ad}}

\newcommand{\dege}{\kappa}
\newcommand{\loc}{{\rm Loc}}
\newcommand{\epslaml}{{\epsi,\Lambda}}
\newcommand{\epslam}{{\Lambda}}
\newcommand{\lplus}{\stackrel{\Lambda}{+}}
\newcommand{\lminus}{\stackrel{\Lambda}{-}}
 \newtheorem{theorem}{Theorem}[section]
\newtheorem*{theorem*}{Theorem}
\newtheorem{lemma}{Lemma}[section]
\newtheorem*{lemma*}{Lemma}

\newtheorem{proposition}{Proposition}[section]
\newtheorem{corollary}{Corollary}[section]
\newtheorem*{remarks}{Remarks} 
\newtheorem*{remark}{Remark}

\title{Adiabatic currents   for interacting fermions\\ on a lattice}
\author{Domenico Monaco%
\thanks{%
Fachbereich Mathematik, Eberhard-Karls-Universit\"at\newline
\textcolor{white}{a} \hspace{.7em} Auf der Morgenstelle 10, 72076 T\"ubingen, Germany\newline
\textcolor{white}{a} \hspace{.7em} \textsl{and}\newline
\textcolor{white}{a} \hspace{.7em} Dipartimento di Matematica e Fisica, Universit\`a degli Studi di Roma Tre\newline
\textcolor{white}{a} \hspace{.7em} Largo San Leonardo Murialdo 1, 00146 Roma, Italy\newline
\textcolor{white}{a} \hspace{.7em} E-mail: dmonaco@mat.uniroma3.it%
}\ %
\ and Stefan Teufel%
\thanks{%
Fachbereich Mathematik, Eberhard-Karls-Universit\"at\newline
	\textcolor{white}{a} \hspace{.7em} Auf der Morgenstelle 10, 72076 T\"ubingen, Germany\newline
	\textcolor{white}{a} \hspace{.7em} E-mail: stefan.teufel@uni-tuebingen.de%
}%
}

\begin{document}
\maketitle

\begin{abstract} 
We prove an adiabatic theorem for  general densities  of   observables that are sums of local terms in finite systems of interacting fermions, without periodicity assumptions on the Hamiltonian and with error estimates that are uniform  in the size of the system. Our result provides an adiabatic expansion to all orders, in particular, also for initial data that lie in eigenspaces of degenerate eigenvalues.
 Our proof is based on ideas from~\cite{BDF}, where Bachmann et al.\ proved  an adiabatic theorem for interacting spin systems.
 
 As one important application of this adiabatic theorem, we provide
 the first rigorous derivation of the adiabatic response  formula for the current density induced by an adiabatic change  of the Hamiltonian of a system of interacting fermions in a ground state, with error estimates uniform in the system size. We also discuss the application to quantum Hall systems.

\medskip


\medskip

\end{abstract}

\newpage

\setcounter{tocdepth}{1}
\tableofcontents

\section{Introduction}

In a number of seminal works, Laughlin \cite{L}, Niu, Thouless {and Wu \cite{NT,NTW}}, and Avron and Seiler \cite{AS} explained the integer and fractional quantization of the Hall conductance resp.\ conductivity in interacting many-body fermion systems starting from the following idea.  According to the adiabatic theorem   of quantum mechanics, such a system remains close to its ground state even when its Hamiltonian slowly changes in time, as long as the ground state remains gapped.
The current density induced by such an adiabatic change is then computed based on the  adiabatic response of the system to this change.
For a system of interacting fermions on a finite cube $\Lambda = (-M/2, M/2]^d\cap \Z^d$ within the lattice $\Z^d$,  the resulting adiabatic response formula for this adiabatic current density can be expressed as follows. Let $(\varphi_0(t),\varphi_1(t),\ldots)$ be an orthonormal basis of eigenvectors of the time-dependent Hamiltonian $H(t)$ with eigenvalues $(E_0(t),E_1(t),\ldots)$, and assume that the system is initially in its non-degenerate ground state~$\varphi_0$. Then the averaged current density induced by a slow change of the Hamiltonian at time $t$  is
\begin{equation}\label{NiuThouless}
\langle J \rangle \approx -\frac{2}{|\Lambda|} {\rm Im}\left( \sum_{n>0} \frac{\langle\varphi_n,\dot \varphi_0\rangle \,\langle \varphi_0, J \varphi_n\rangle}{E_n-E_0}
\right)\,,
\end{equation}
where $J(t)$ is the current operator associated with $H(t)$, $|\Lambda|$ is the number of lattice sites, and $\approx$ refers to asymptotic closeness in the adiabatic limit.

Starting from   formula \eqref{NiuThouless}, e.g.\ Niu and Thouless \cite{NT} argue for quantization of the transported charge under cyclic changes of the Hamiltonian in the thermodynamic limit $|\Lambda|\to\infty$ and, by a similar argument, for integer quantization of Hall conductivity also for interacting fermion systems in the thermodynamic limit. See also Avron and Seiler \cite{AS} for   closely related arguments, Hastings and Michalakis \cite{HM} for  a rigorous proof showing quantization of conductance in  finite interacting spin systems up to almost-exponentially small terms in the system size  (and also Hatsugai et al.\ \cite{KWKH} who provide numerical evidence for this fact in an interacting Hofstadter model).

However, the standard argument (see e.g.\ \cite{ASY} for a rigorous account) leading to the  formula \eqref{NiuThouless} for the current density (i.e.\ the starting points in \cite{NT,AS,HM} and many others)  does not provide error bounds uniform in the system size $|\Lambda|$. This is because in  the standard adiabatic theorem one has no control on the dependence of the error on the system size, and the adiabatic approximation might  deteriorate in the thermodynamic limit.

More precisely, let $H^\Lambda(t)$ be a smooth time-dependent family of bounded self-adjoint Hamiltonians generating the time-evolution  
\[
\I\epsi \frac{\D}{\D t} \rho(t)  = [ H^\Lambda (t), \rho(t)]  \,, \quad \rho(0)=\rho_0\,.
\]
Assume that $E^\Lambda_*(t)\in\sigma(H^\Lambda(t))$ is an eigenvalue depending smoothly on $t$ that remains  
isolated from the rest of the spectrum for all times and 
denote by $P^\Lambda_*(t)$ the corresponding family of spectral projections.
Then a direct consequence of 
 the version of the adiabatic theorem going back to Kato \cite{Ka}
 is that for  any initial state $\rho_0$ in the range of $P^\Lambda_*(0)$, i.e.\ $P^\Lambda_*(0) \rho_0P^\Lambda_*(0)= \rho_0$, and any $T<\infty$ there exists a constant $C^\Lambda_T<\infty$ such that for any bounded $B^\Lambda$
 \begin{equation}\label{introadi}
\sup_{t,s\in [-T,T]} \left| \tr (\rho (t) B^\Lambda ) - \tr (\rho_\parallel(t) B^\Lambda  )\right| \leq  \epsi \,C^\Lambda_T \,\|B^\Lambda\|\,,
\end{equation}
 where $\rho_\parallel(t)$ is the solution to the parallel transport equation 
 \[
\I \frac{\D}{\D t} \rho_\parallel(t) =   [K^\Lambda_\parallel (t), \rho_\parallel(t)] 
 \,, \quad \rho_\parallel(0)=\rho_0\,. 
 \]
Here $K^\Lambda_\parallel(t):= \I [\dot P^\Lambda_*(t),P^\Lambda_*(t)]$ is the generator of parallel transport.
The constant $C^\Lambda_T$ in \eqref{introadi} depends, among other quantities, linearly on the norm $\|\dot H^\Lambda(t)\|$ of $\dot H^\Lambda(t)$. This, however, is unsatisfactory when dealing with extended systems, where the energy $H^\Lambda(t)$ itself as well as its time-derivative $\dot H^\Lambda(t)$ are typically extensive quantities 
with norms proportional to the size of the system. Then $C^\Lambda_T\sim |\Lambda|$ and the estimate \eqref{introadi} becomes worthless whenever one is interested in large $|\Lambda|$  at fixed $\epsi$. 

As a special case of a much more general adiabatic theorem we will show that for   lattice fermions with a Hamiltonian $H^\Lambda(t)$ that is a sum of local terms the estimate \eqref{introadi} basically holds with a constant $C^\Lambda_T\equiv C_T$ independent of the volume $|\Lambda|$ whenever the observable $B^\Lambda$ is also a sum of local terms. The ``basically'' refers to the fact that, if $B^\Lambda$ is not a local observable,   then $\|B^\Lambda\|$  gets replaced by another quantity that grows, however,  at the same rate as $\|B^\Lambda\|$ with the system size $|\Lambda|$, namely proportional to the volume of the support of~$B^\Lambda$.

The result just sketched can be obtained as a  corollary of a recent   result of   Bachmann, De Roeck, and Fraas    \cite{BDFletter,BDF}. Their result is, to our knowledge,  the first instance of an adiabatic theorem for an interacting system with error bounds uniform in the system size. They use  a very subtle combination of Lieb--Robinson bounds and the so-called quasi-adiabatic evolution  in order to maintain locality   in all steps of the adiabatic approximation.   Also our proofs   rely on the machinery developed in \cite{BDF}.

However, mostly with the application to adiabatic currents in mind, we   improve and generalize the result of \cite{BDF} in at least two ways. 
First, we show that the order of the error in \eqref{introadi} can be improved to $\epsi^2$ by  modifying the generator of parallel transport $K^\Lambda = K^\Lambda_\parallel + \epsi K^\Lambda_1$ by an explicit term of order $\epsi$.
It is well known (e.g.\ \cite{NT, ST}) and at the heart of our derivation of \eqref{NiuThouless} that this first order correction to the parallel transport is responsible for the leading order contribution to adiabatic currents.
Second, we show that if $\dot H^\Lambda(t)$ and $B^\Lambda$ are both  supported around lower-dimensional  planes, then $C_T\|B^\Lambda\|$ in the right hand side of \eqref{introadi} can be replaced by a constant times $M^{\tilde d}$, where we recall that $M$ is the side-length of the cube $\Lambda$ and $\tilde d$ is the dimension of the intersection of  the supports of $\dot H^\Lambda$ and $B^\Lambda$. This is relevant, e.g., when computing the conductance in a two-dimensional quantum-Hall system. There $\dot H^\Lambda$ is supported near a line and the observable $B^\Lambda$ is the current across a line perpendicular to the first one. The intersection of the supports of $\dot H^\Lambda$ and $B^\Lambda$ is a fixed area independent of $\Lambda$, hence $\tilde d = 0$, and the right hand side of \eqref{introadi} is of the form $\epsi C_T$ with a constant $C_T$ independent of $\Lambda$. For a more detailed presentation and discussion of our general adiabatic theorem and its relation to \cite{BDF} we refer to the remarks after Theorem~\ref{AdiThm} in Section~\ref{sec:AdiThm}.

As mentioned before, our  results   relate to quantum Hall systems, in particular to quantization of conductivity and conductance.  
 Assuming that the results of Hastings and Michalakis \cite{HM,H} or Bachmann et al.\ \cite{BBDF}
carry over as expected  from spin systems to interacting fermions, then   our derivation of adiabatic response formulas for adiabatic currents  completes a rigorous chain of arguments that starts from microscopic first principles and proves quantization of Hall conductance in the thermodynamic limit for certain perturbations of gapped  free fermion Hamiltonians,  cf.\ \eqref{infinitecond}.
Here it should be noted that  Fr\"ohlich \cite{Fr} (and references therein) developed a different approach through gauge-theoretic arguments to the quantization of conductance in interacting Hall insulators.

\medskip

We end the introduction with a few remarks on related literature.
The idea of a topological quantum pump in a non-interacting fermion system was pioneered by Thouless \cite{ThPump}, and has been recently experimentally realized with ultracold atoms \cite{NTTIOWTT,LSZAB}. Similar ideas inspired the simulation of a topological adiabatic pump in a quasicrystal through optical waveguides \cite{KLRVZ}, where the role of the adiabatic time is played by the length of the waveguide. The tunability of these quantum simulation systems could allow to test experimentally the validity of our predictions when interactions are turned on.
The formula for the induced current is not only relevant for quantum pumps and the quantum Hall effect, but also for computing the change of polarization in the piezoelectric effect. For non-interacting systems, the resulting formula in the thermodynamic limit is called the King-Smith and Vanderbilt formula \cite{KSV} and it was rigorously derived for continuous periodic systems in \cite{PST} and for random systems on a lattice in \cite{ST}. 
 
 A closely related   problem is the justification of linear response formulas in systems where the driving actually closes the gap. For example, the addition of a   uniform electric field, i.e.\ a linearly growing scalar potential, is expected to close the gap of any initially gapped  Hamiltonian. 
 While for interacting systems this problem was tackled only recently in \cite{T2}, heavily using the machinery developed in the present paper and in \cite{BDF}, 
for non-interacting systems there are numerous rigorous results (e.g.\ \cite{BES,BGKS}). For example, in~\cite{BGKS} the authors take a step towards the justification of linear response formulas for magnetic Schr\"odinger operators with random potentials, where, instead of a spectral gap, only a mobility gap is assumed for the initial Hamiltonian. A more general ``analytic-algebraic'' approach, based partly on ideas from~\cite{BES} and \cite{BGKS}, has been formalized by De Nittis and Lein in the recent monograph \cite{DNL}.
 A different approach for dealing with perturbations that close the spectral gap, but that leave a microlocal gap structure, is  based on space-adiabatic theory, see e.g.\ \cite{PST2,PST3,T,MMPT}. This approach does not apply, however, in the presence of a  mobility gap only.
 
Finally we mention a recent series of papers (see \cite{BD,BDH} and references therein) by Bru, de Siqueira Pedra, and Hertling on the derivation of a microscopic Ohm's law for interacting fermion systems at finite temperature. While their setup is quite similar to ours, they answer a different kind of question. They consider periodic systems with homogeneous randomness   initially in a thermal state at positive temperature and establish, among other things,  that the microscopic current density induced by compactly supported electro-magnetic fields has a  leading term proportional to the strength of the field with higher order terms being quadratic in the field strength uniformly in the system size. Results on the validity of linear response were obtained by Jak\v{s}i\'{c}, Ogata and Pillet, see \cite{JOP} and references therein, using a similar formalism, adapted to the context of open quantum systems.

\medskip

Our paper is structured as follows. In Section~\ref{sec:framework} we introduce  the mathematical framework for fermionic many-body Hamiltonians on a lattice.  This is mostly standard and serves to fix notation, with one exception: We introduce new spaces of  local Hamiltonians  that are localized in certain directions. This will be useful for handling observables like the charge current through a line or surface. 
 In Section~\ref{sec:AdiThm} we formulate the assumptions and the statement of our adiabatic theorem, Theorem~\ref{AdiThm}, and indicate the main steps of the proof. The application to adiabatic currents and the rigorous derivation of the adiabatic response formulas are presented in Section~\ref{sec:AdiCurrent}. 
 Section~\ref{PropProof} and Section~\ref{sec:proofAdiabatic} contain the proof of the adiabatic theorem. Finally we end with several appendices proving different technical details. 
 
\medskip

\noindent{\bf Acknowledgement.} We are grateful to Giuseppe De Nittis, Max Lein, Giovanna Marcelli,  Gianluca Panati, Felix Rexze, and Cl\'{e}ment Tauber for intensive discussions concerning closely related questions. We also profited from continual exchange with Sven Bachmann, Wojciech de Roeck, and Martin Fraas.
Finally we thank Marcello Porta for valuable hints to the literature. This work was supported by the German Research Foundation within the Research Training Group 1838 on ``Spectral theory and dynamics of quantum systems''. Financial support from the ERC Consolidator Grant 2016 ``UniCoSM -- Universality in Condensed Matter and Statistical Mechanics'' is also gratefully acknowledged.

\section{The mathematical framework} \label{sec:framework}

Let $\Gamma = \Z^d$ be the infinite lattice and  $\Lambda =\Lambda(M) := \{-\frac{M}{2}+1,\ldots, \frac{M}{2}\}^d\subset \Gamma$  the centered box of size $M$, with $M\in\N$ even. 
The map
$
 \Gamma\times \Gamma\to\Gamma$, $(x,y) \mapsto   x+y$,
makes $\Gamma$  an abelian group.
In order to  have a meaningful framework for considering currents  also in  finite systems, 
we  think of $\Lambda$ as a $d$-dimensional torus, i.e.\ as representing the quotient $\Gamma/( M\cdot\Gamma)$ of $\Gamma$ by the normal subgroup $ M\cdot\Gamma$.  
This turns also $\Lambda$ into an abelian group and we will use the notation
\[
  \Lambda\times \Lambda\to\Lambda\,,\quad  (x,y) \mapsto    x\lplus y  
\]
for the sum of elements in $\Lambda$ modulo translations in $ M\cdot\Gamma$.

The one-particle Hilbert space is $\mathfrak{h}_\Lambda = \ell^2(\Lambda, \C^\ell)$, where $\C^\ell$ describes   spin and the internal structure of the unit cell (that is, sublattice or pseudospin degrees of freedom). The $N$-particle Hilbert space is then $\Hi_{\Lambda,N} := \bigwedge_{j=1}^N \mathfrak{h}_\Lambda$, and the fermionic Fock space is denoted by $\mathfrak{F}_\Lambda = \bigoplus_{N=0}^{\ell M^d} \Hi_{\Lambda,N}$, where $\Hi_{\Lambda,0} := \C$. 
Note that all Hilbert spaces in the following are finite-dimensional and thus all operators are actually matrices.
Let $a_{i,x}$ and $a_{i,x}^*$, $i=1,\ldots,\ell$, $x\in\Gamma$, be the standard fermionic annihilation and creation operators satisfying the canonical anti-commutation relations
\[
\{ a_{i,x}, a_{j,y}^* \} = \delta_{i,j} \delta_{x,y} {\bf 1}_{\mathfrak{F}_\Lambda}\quad\mbox{and}\quad \{ a_{i,x}, a_{j,y}  \} = 0 =  \{ a_{i,x}^*, a_{j,y}^*  \}\,,
\]
where $\{a,b\}:= ab+ba$ is the anti-commutator.
While it turns out useful   in the following to write all  operators on Fock space~$\mathfrak{F}_\Lambda$, we will consider only Hamiltonians that preserve the number of particles.

For a   subset $X\subset \Lambda$ we denote by $\mathcal{A}_X\subset \mathcal{L}(\mathfrak{F}_\Lambda)$ the algebra of operators generated by the set
$\{   a_{i,x}, a_{i,x}^*\,|\, x\in X\,, i=1,\ldots, \ell\}$. Those elements of $\mathcal{A}_X$ commuting with the number operator 
\[
\mathfrak{N}_X := \sum_{x\in X} a_x^*a_x := \sum_{x\in X} \sum_{j=1}^\ell a_{j,x}^*a_{j,x}
\]
form a subalgebra $\mathcal{A}_X^\mathfrak{N}$ of $\mathcal{A}_X$ contained in the subalgebra $\mathcal{A}_X^+$ of even elements\footnote{An operator in $\mathcal{A}_X$ is called \emph{even} (resp.~\emph{odd}) if it commutes (resp.~anti-commutes) with the fermion parity operator $(-1)^{\mathfrak{N}_X}$. The subalgebra of even operators is denoted by $\mathcal{A}_X^+$.}, i.e.\ $\mathcal{A}_X^\mathfrak{N}\subset  \mathcal{A}_X^+\subset  \mathcal{A}_X $. Note that we will use the vector notation for $a_x$ as introduced above without further notice in the following.

We now come to the definition of interactions and Hamiltonians.
Let $\mathcal{F}(\Gamma) := \{ X\subset \Gamma\,|\, |X|<\infty\}$ be the set of all finite subsets of $\Gamma$. Analogously we define  also $\mathcal{F}(\Lambda) := \{ X\subset \Lambda\}$. 
An  interaction $\Phi= \{\Phi^\epslaml\}_{\epsi\in (0,1],\Lambda=\Lambda(M), \: M \in \N}$ is a family of maps  
\[
\Phi^\epslaml: \mathcal{F} (\Lambda) \to \bigcup_{X \in \mathcal{F}(\Lambda)} \mathcal{A}_X^\mathfrak{N} \,,\quad X\mapsto\Phi^\epslaml(X)\in \mathcal{A}_X^\mathfrak{N}
\]
taking values in the self-adjoint operators. Here $\epsi\in(0,1]$ is the adiabatic parameter and $\epsi$-dependent interactions and Hamiltonians will naturally appear in our analysis, typically (but not necessarily) by considering interactions which depend on the adiabatic time $\tau = \epsi t$.
The   Hamiltonian $A = \{ A^\epslaml\}_{\epsi,\Lambda}$ associated with the   interaction  $\Phi$ is the family of self-adjoint operators
\begin{equation}\label{Hamiltonian}
A^\epslaml \equiv A^\epslaml(\Phi) :=  \sum_{X \subset \Lambda} \Phi^\epslaml(X)  \in \mathcal{A}_\Lambda^\mathfrak{N}\,.
\end{equation}
One can turn the vector space of interactions into a normed space as follows (cf.\ e.g.\ \cite{NSY}). Introduce first
\[
d^\Lambda: \Lambda\times \Lambda \to \N_0\,,\quad d^\Lambda(x,y) := d( 0, y\lminus x)\,,
\]
where  $d:\Gamma\times\Gamma \to \N_0$ denotes the $\ell^1$-distance on $\Gamma$. Thus  $d^\Lambda$ is exactly the $\ell^1$-distance on the ``torus'' $\Lambda$.
Moreover, define
\[
F (r) := \frac{1}{(1+r)^{d+1}}   \qquad\mbox{and}\qquad F_\zeta(r) := \frac{\zeta(r)}{(1+r)^{d+1}}\,,
\]
where 
\begin{equation} \label{eqn:S_def}
\begin{aligned}
\zeta\in \mathcal{S} & := \{ \zeta:[0,\infty)\to (0,\infty)\,|\, \mbox{$\zeta$ is bounded, non-increasing, satisfies }\\
& \quad\quad
\zeta(r+s) \geq \zeta(r)\zeta(s)\;\mbox{ for all } r,s\in[0,\infty)\mbox{ and } \\
& \quad\quad
\sup\{ r^n\zeta(r)\,|\, r\in [0,\infty)\} <\infty \mbox{ for all } n\in\N\}\,.
\end{aligned}
\end{equation}
For $\zeta \in \mathcal{S}$, the corresponding norm on the vector space of  interactions is then given by
\[
\|\Phi\|_{\zeta,n} :=  \sup_{\epsi\in (0,1]} \sup_{\Lambda} \sup_{x,y\in\Lambda} \sum_{\substack{X \subset \Lambda:\\ \{x,y\}\subset X}} |X|^n \frac{\|\Phi^\epslaml(X)\|}{F_\zeta(d^\Lambda(x,y))} 
\]
for $n\in \N_0$. The prime example for a function $\zeta\in\mathcal{S}$ is  $\zeta(r) = \E^{-ar}$ for some $a>0$: for this specific choice of $\zeta$ we write $F_a$ and $\|\Phi\|_{a,n}$ for the corresponding norm. 

It will be important to consider also interactions that are localized in certain directions around certain locations. To this end we introduce the space of localization planes 
\[
\loc :=  \{0,1\}^d \times {\textstyle \prod}_{M=2}^\infty \Lambda(M)\,.
\]
The idea is that a point 
 $L=: ( \ell, l^{\Lambda(2)},  l^{\Lambda(4)}, \ldots)\in \loc$ defines for each $\Lambda$ a $(d-|\ell|)$-dimensional hyperplane through the point $l^\Lambda\in\Lambda$ which is parallel to the one given by $\{x_j=0 \text{ if } \ell_j=1\}$.
 Here $|\ell|:=|\{ \ell_j =1\}|$ is the number   of   constrained directions. 
  \begin{SCfigure}[1.15]
   \includegraphics[width=0.31\textwidth]{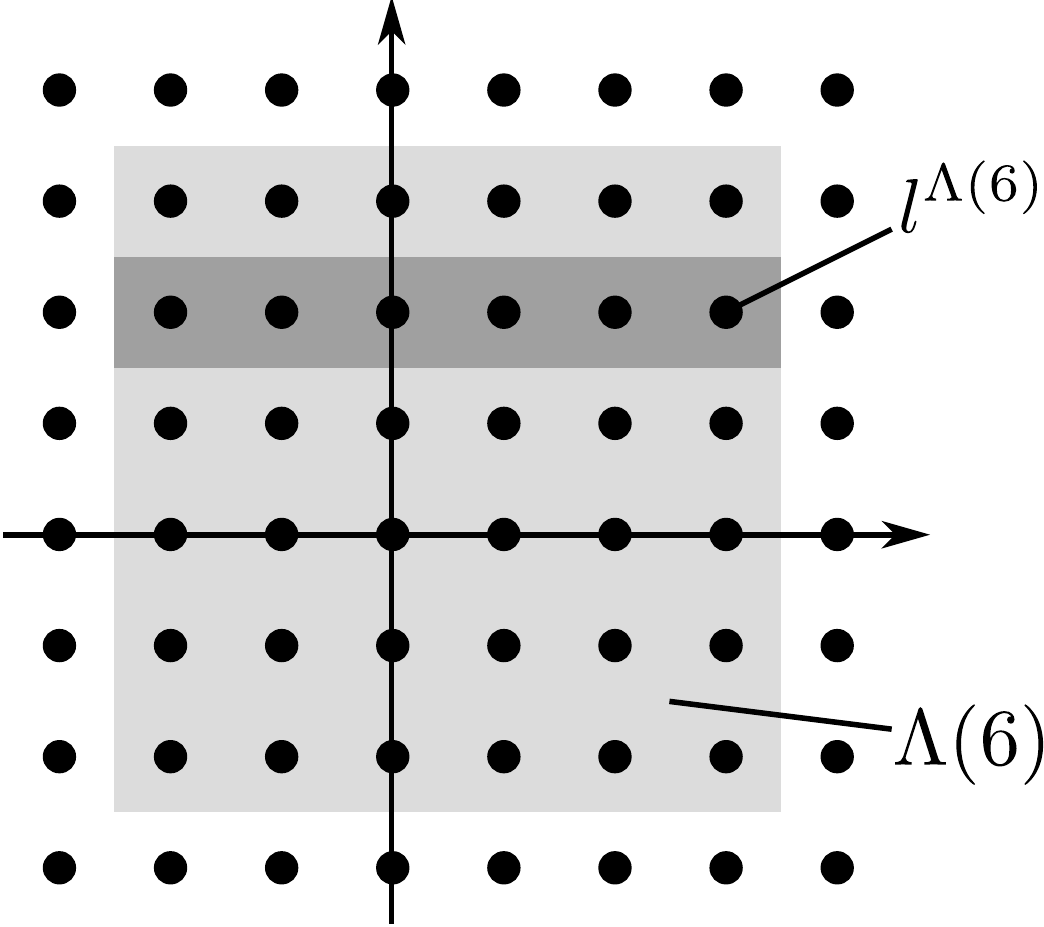}
    \caption{\small The light shaded region is the cube $\Lambda(6)$. The darker shaded region is the hyperplane defined by the localization vector $L$ with $\ell =(0,1)$ corresponding to localisation   in the $2$-direction around the point $l^{\Lambda(6)}$. An interaction with finite $\|\cdot\|_{\zeta,n,L}$-norm is then localized near this hyperplane. } 
    \label{fig:1}
\end{SCfigure}
 The 
 distance of a point $x\in\Lambda$ to this hyperplane is
 \begin{equation} \label{distxL}
 {\rm dist}(x,L) := \sum_{j=1}^d   | (x\lminus l^\Lambda)_j| \,\ell_j
 \end{equation}
 and we define a
 new ``metric'' on $\Lambda$ by 
 \[
 d^\Lambda_L:\Lambda\times\Lambda\to \N_0\,,\quad (x,y)\mapsto d^\Lambda_L(x,y) := d^\Lambda(x,y) +  {\rm dist}(x,L) + {\rm dist}(y,L) \,.
 \] 
 Note that $d^\Lambda_L$ is no longer a metric on $\Lambda$ but obviously still satisfies the triangle inequality.
 The corresponding norms are denoted by 
 \[
\|\Phi\|_{\zeta,n,L} :=  \sup_{\epsi\in(0,1]} \sup_{\Lambda} \sup_{x,y\in\Lambda} \sum_{\substack{X \subset \Lambda:\\ \{x,y\}\subset X}} |X|^n \frac{\|\Phi^\epslaml(X)\|}{F_\zeta(d^\Lambda_L(x,y))}  \,.
\]
These norms will basically always be used for the following type of estimate,
\[
\sum_{\substack{X \subset \Lambda:\\ \{x,y\}\subset X}} |X|^n  \|\Phi^\epslaml(X)\|\;\leq\; \|\Phi\|_{\zeta,n,L}\,F_\zeta(d^\Lambda_L(x,y))\,.
\] 
That means, in particular,  that  $\|\Phi^\epslaml(X)\|$ is small whenever the diameter of $X$ is large or if the distance of $X$ to $L$ is large.  This situation is illustrated by Figure~\ref{fig:1}.

A Hamiltonian $A$ with interaction $\Phi_A$ such that $\|\Phi_A\|_{\zeta,0,L}<\infty$ for some $\zeta\in\mathcal{S}$ is called \emph{local and $L$-localized}. One crucial property of local $L$-localized Hamiltonians is that the norm of the finite-size operator $A^\epslaml$
grows at most as the volume $M^{d-|\ell|}$ of its support, 
\begin{equation}\label{normgrowth}
\|A^\epslaml\|\leq C_\zeta\, \|\Phi_A\|_{\zeta,0,L} \, M^{d-|\ell|}\, \,,
\end{equation}
  cf.\ Lemma~\ref{BoundLemma} in Appendix~\ref{AppendixTech}.  
 
Let $\mathcal{B}_{\zeta,n,L}$ be the Banach space of   interactions with finite $\|\cdot\|_{\zeta,n,L}$-norm, and put
\[
 \mathcal{B}_{\mathcal{S},n,L} :=  \bigcup_{\zeta\in\mathcal{S}}  \mathcal{B}_{\zeta,n,L}\,, \qquad   \mathcal{B}_{\mathcal{E},n,L} := \bigcup_{a>0}  \mathcal{B}_{a,n,L} \,,
\]
and
\[
 \mathcal{B}_{\mathcal{S},\infty,L} := \bigcap_{n\in\N_0}  \mathcal{B}_{\mathcal{S},n,L} \,,\qquad  \mathcal{B}_{\mathcal{E},\infty,L} := \bigcap_{n\in\N_0}  \mathcal{B}_{\mathcal{E},n,L}\,.
\]
Note that $\Phi \in \mathcal{B}_{\mathcal{S},\infty,L} $ merely means that there exists a sequence $\zeta_n\in\mathcal{S}$ such that
$\Phi\in\mathcal{B}_{\zeta_n,n,L}$ for all $n\in\N_0$.
The corresponding spaces of   Hamiltonians are denoted by $\mathcal{L}_{\zeta,n,L}$, $ \mathcal{L}_{\mathcal{E},n,L}$, $\mathcal{L}_{\mathcal{E},\infty,L}$, $ \mathcal{L}_{\mathcal{S},n,L}$, and $\mathcal{L}_{\mathcal{S},\infty,L}$ respectively: that is, a   Hamiltonian $A$ belongs to $\mathcal{L}_{\zeta,n,L}$ if it can be written in the form \eqref{Hamiltonian} with an   interaction in~$\mathcal{B}_{\zeta,n,L}$, and similarly for the other spaces.  Lemma~\ref{Slemma} in Appendix~\ref{Sapp} shows that the spaces $ \mathcal{B}_{\mathcal{S},n,L} $ and thus also $ \mathcal{B}_{\mathcal{S},\infty,L} $ are indeed vector spaces.
One of the crucial features of these spaces, that will be used repeatedly in the following, is that these are in general not \emph{algebras} of operators (that is, the product of two local $L$-localized operators  need neither be local nor $L$-localized), but nonetheless are closed under taking \emph{commutators}: for example, $A\in \mathcal{L}_{\mathcal{S},\infty,L}$ and $B\in \mathcal{L}_{\mathcal{S},\infty}$ implies $\add_{A}(B) := [A,B] \in \mathcal{L}_{\mathcal{S},\infty,L}$,
 compare Lemmas~\ref{lemma:comm1} and \ref{manyadlemma} in Appendix~\ref{AppendixTech}. Note that when we don't write the index $L$, this means that $L=0 := (\vec{0},0,0,\ldots) $ and the interaction (respectively the Hamiltonian) is local but not localized in any direction.

Finally, we say that an interaction $\Phi_A$, resp.\ the corresponding Hamiltonian~$A$, is {\em uniformly finite range} if 
\[
\sup_{\epsi,\Lambda} \max \{ \operatorname{diam}(X)\,|\, X \subset \Lambda\,, \Phi^\epslaml_A (X)\not=0\}  <\infty
\]
and 
\[
\sup_{\epsi,\Lambda} \max \{ \|\Phi^\epslaml_A(X)\|\,|\, X \subset \Lambda\}  <\infty\,.
\]
Note that these conditions imply that $\Phi_A\in\mathcal{B}_{\mathcal{E},\infty}$ and thus $A\in\mathcal{L}_{\mathcal{E},\infty}$.

\section{Adiabatic theorems} \label{sec:AdiThm}

Let $\Phi_H(t)$, $t\in \R$, be a time-dependent   interaction  giving rise to a time-dependent   Hamiltonian  $H(t)$, which will be the physical Hamiltonian of the system in the following. 
A typical  example of a physically relevant family of Hamiltonians to which our results apply is the family of operators
\begin{equation}\label{ExHamiltonian}
\begin{aligned}
H^\Lambda_{TVW}(t) &=  \sum_{ (x,y)\in \Lambda^2}    a^*_x \,T(t,x\lminus y) \,a_y\; +\; \sum_{x\in\Lambda}  a^*_xV(t,x)a_x \\
&\quad +\;  \sum_{\{x,y\}\subset \Lambda }  a^*_xa_x \,W(t,d^\Lambda(x,y))\,a^*_ya_y \;-\; \mu\, \mathfrak{N}_\Lambda\,.
\end{aligned}
\end{equation}
Here the kinetic term $T(t):\Gamma \to \mathcal{L}(\C^\ell)$ is a compactly supported function with $T(t,-x) = T(t,x)^*$, the potential term $V(t):\Gamma \to \mathcal{L}(\C^\ell)$ is a bounded function taking values in the self-adjoint matrices, and the two-body interaction $W(t):\N_0\to 
\mathcal{L}(\C^\ell)$ is compactly supported and also takes values in the self-adjoint matrices. The real number $\mu\in\R$ is   the chemical potential.
Under these conditions on $T$, $V$, and $W$, the Hamiltonian $H_{TVW}(t)$ is uniformly finite range, as the interactions $\Phi_{TVW}^\Lambda(t,X)$ associated to $H_{TVW}^\Lambda(t)$ via \eqref{Hamiltonian} vanish whenever $X \subset \Lambda$ has cardinality larger than $2$.

We will now state the standing assumptions on the Hamiltonian needed to formulate our adiabatic theorems. To this end, we introduce the following norms for time-dependent interactions. For $\zeta \in \mathcal{S}$, $n \in \N_0$, $L \in \loc$, and  $T \ge 0$ let
\[
\|\Phi\|_{\zeta,n,L,T} := \sup_{|t|\leq T} \|\Phi(t)\|_{\zeta,n,L} \,.
\]

\noindent {\bf (A1)$_{ m,L_{ H }}$ Smoothness  of  the Hamiltonian and   localization of   the driving.}\\ {\em
Let $\Phi_H(t)$, $t\in\R$, be a time-dependent interaction with $\|\Phi_H\|_{a,n, T} <\infty$ for some $a>0$ and all $T\in [0,\infty)$ and $n\in\N_0$. 
Let $m\in\N$ and assume that each map $[0,\infty) \to \mathcal{A}_X^\mathfrak{N}$, $t\mapsto \Phi^\epslaml_H(t,X)$ is $(d+m)$-times differentiable. Let $\{ (\Phi^\epslaml_H)^{(k)}(t)\}_\Lambda$ be the time-dependent  interaction defined by their $k$-th derivatives,  for $1\leq k \leq d+m$. Assume that  for some localization vector $L _H\in\loc$ and all $T>0$ and $n\in\N_0$
\[
  \sup_{1\leq k \leq d+m} \| (\Phi_H)^{(k)}\|_{a,n,L_H,T} <\infty \,.
\]}

According to this assumption,  the driving, that is the region of space where the Hamiltonian varies in time, can---but need not---be localized around some lower dimensional plane.
Note also that the Hamiltonian $H^\Lambda_{TVW}(t)$  in~\eqref{ExHamiltonian} satisfies Assumption (A1)$_{m,0}$ $\equiv$ (A1)$_{m}$  whenever $T$, $V$, and $W$ are, in addition to the conditions formulated above,  $(d+m)$-times differentiable with respect to $t$. \\

\noindent {\bf (A2) Uniform gap in the spectrum.} {\em 
We assume that there exists $M_0\in\N$ such that for all $M\geq M_0$ and corresponding $\Lambda=\Lambda(M)$  the operator $H^\epslaml(t)$   has a gapped part $\sigma^\epslaml_*(t)\subset\sigma(H^\epslaml (t))$ of its spectrum with the following properties: There exist continuous functions $f^\epslaml_\pm:\R\to \R$ and  constants $g >0$ and $\delta<\infty$  such that  for all $\epsi\in (0,1]$, $\Lambda$, and $t\in\R$
\[
f^\epslaml_\pm(t)\in \rho(H^\epslaml(t))\,,\qquad [f^\epslaml_-(t), f^\epslaml_+(t)] \cap \sigma(H^\epslaml (t)) = \sigma^\epslaml_*(t)\,,
\]
\[
{\rm diam}(\sigma_*^\epslaml(t)) \leq \delta   \,,\quad\mbox{and}\quad
{\rm dist}\left( \sigma_*^\epslaml(t) , \sigma(H^\epslaml(t))\setminus \sigma_*^\epslaml(t)\right) \geq g \,.
\]
  We denote by $P_*^\epslaml(t)$ the spectral projection of $H ^\epslaml(t)$ corresponding  
to the spectrum $\sigma_*^\epslaml(t)$.}\\

To prove the existence of a uniform gap is   a nontrivial problem in general. For Hamiltonians of the form $H^\Lambda_{TVW}$, however, existence of a gap for appropriate choice of the chemical potential $\mu$ can be deduced as follows:  For $W=0$, the problem can be reduced to the spectral analysis of the underlying  one-body Hamiltonian
\[
h^\Lambda: \ell^2(\Lambda,\C^\ell)\to \ell^2(\Lambda,\C^\ell)\,,\quad (h^\Lambda\psi)(x):= \sum_{y\in\Lambda} T( x\lminus y)\psi(y)  + V(x) \psi(x)\,.
\]
 If the latter has a spectral gap and the chemical potential  $\mu$  lies in this gap, then the corresponding non-interacting many-body Hamiltonian $H^\Lambda_{TV}$ has a gapped ground state. see also Figure~\ref{fig2}.
   Moreover,  it has been shown recently     in \cite{H,DS} (and announced in \cite{NSY}) that   for  sufficiently small interactions $W$ the gap of the interacting many-body operator $H^\Lambda_{TVW}$ remains open. 
    \begin{SCfigure}[1.78] \label{fig2}
   \includegraphics[width=0.35\textwidth]{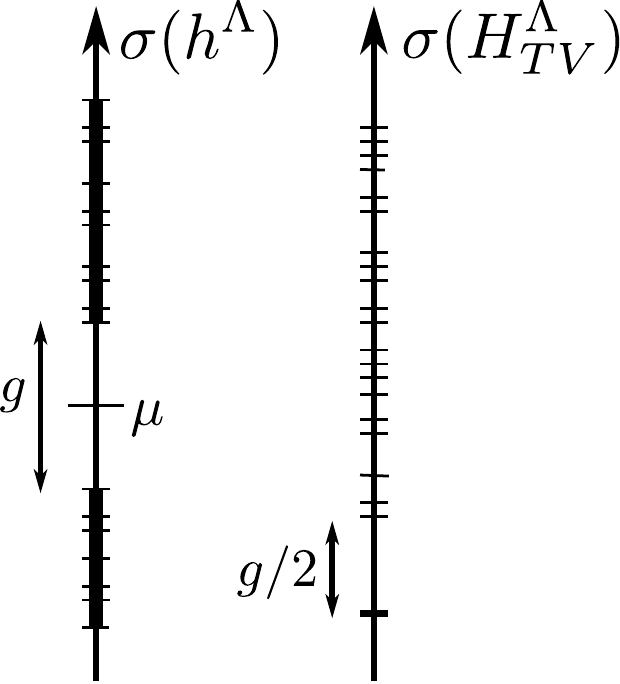}
    \caption{\small The spectrum of the one-body Hamiltonian $h^\Lambda$ restricted to $\Lambda$ is shown on the left side and assumed to have a gap of size $g$. E.g.\ for periodic one-body operators $h^\Lambda$ the eigenvalues (horizontal lines) lie in intervals, so-called spectral bands  (thick vertical lines) that are independent of $\Lambda$.  If $\mu$ lies in the middle of the gap of $h^\Lambda$, then the spectrum of the corresponding non-interacting many-body operator $H^\Lambda_{TV}$ is shown on the right and has a gap of size $g/2$ above its lowest eigenvalue, the so-called   ground state (thick horizontal line). The ground state eigenfunction of $H^\Lambda_{TV}$ is supported in the sector of Fock space with  $N$ particles, where $N$ is the number of eigenvalues of $h^\Lambda$ below $\mu$. The restriction $H^\Lambda_{TV}|_{\mathfrak{N}=N}$ of $H^\Lambda_{TV}$ to this sector even has a gap of size $g$ above its ground state.
  } 
\end{SCfigure}
\medskip

For each $M\geq M_0$ we   consider the   time evolution generated by $H^{\epsi,\Lambda}(t)$ with adiabatic scaling, i.e.\ the unitary propagator $U^{\epsi,\Lambda}(t,s)$ satisfying
\begin{equation} \label{Hdynamics}
\I\epsi \frac{\D}{\D t} U^{\epsi,\Lambda}(t,s) = H^{\epsi,\Lambda}(t) \,U^{\epsi,\Lambda}(t,s)\,, \quad U^{\epsi,\Lambda}(s,s) = {\bf 1}\;\mbox{ for }\; 
 t,s\in \R\,.
\end{equation}
We will be interested in the adiabatic limit, that is the asymptotic behavior of the solution $U^{\epsi,\Lambda}(t,s)$ for $\epsi\ll 1$. 
 
We start by first formulating  a simple leading order adiabatic theorem in the spirit of Kato \cite{Ka}. 
It will follow as a special case of a much more general superadiabatic theorem stated and proved afterwards.
First recall that
the generator of the parallel transport within the time-dependent eigenspaces is given by 
\begin{equation}\label{Kato}
  K^{\epsi,\Lambda}_\parallel (t) :=   \I\, [ \dot P_*^{\epsi,\Lambda}(t), P_*^{\epsi,\Lambda}(t)]   \,,
\end{equation}
i.e.\ the parallel transport map  $U^{\epsi,\Lambda}_\parallel (t,s)$   is the solution to 
\begin{equation}  
\I  \frac{\D}{\D t} U^{\epsi,\Lambda}_\parallel (t,s) =  K^{\epsi,\Lambda}_\parallel (t)  \,U^{\epsi,\Lambda}_\parallel(t,s)\,, \quad U^{\epsi,\Lambda}_\parallel (s,s) = {\bf 1}\;\mbox{ for }\; 
 t,s\in \R\,.
\end{equation}
Note that if the Hamiltonian $H^{\epsi,\Lambda}=H^{\Lambda}$ does not depend on $\epsi$ (which is the typical situation in adiabatic theory), then also $P_*^{\epsi,\Lambda}=P_*^{ \Lambda}$, $K^{\epsi,\Lambda}_\parallel=K^{ \Lambda}_\parallel$, and $U^{\epsi,\Lambda}_\parallel=U^{ \Lambda}_\parallel$ are independent of $\epsi$.
It is well known, and easy to check by differentiating the following equality, that the parallel transport map indeed intertwines the eigenspaces at different times,
\begin{equation}\label{adievo}
U^{\epsi,\Lambda}_\parallel (t,s) \,P^{\epsi,\Lambda}_*(s) = P^{\epsi,\Lambda}_*(t)\,U^{\epsi,\Lambda}_\parallel (t,s)\,.
\end{equation}

From now on we will drop the superscript $\epsi$ from $\epsi$-dependent operators and other quantities in order to not overburden the notation. We keep the superscript $\Lambda$ in order to distinguish a local Hamiltonian $B=\{B^\epslaml\}$ from    its elements $B^\epslam=B^\epslaml \in \mathcal{A}^\mathfrak{N}_\Lambda$ and analogously for interactions.
 \medskip

 Kato's adiabatic theorem implies, under the additional assumption that $\sigma^\epslam_*(t) = \{E^\Lambda(t)\}$ is an eigenvalue,   that on the range of $P_*^\epslam(t)$ the Heisenberg time evolution of arbitrary observables $B^\Lambda$ can be approximated by parallel transport in the following sense: For each $T>0$ there exists a constant $C_T^\Lambda$ such that for all $B^\Lambda  \in \mathcal{A}^\mathfrak{N}_\Lambda$ is holds that
\begin{equation}\label{KatoThm}
\sup_{s,t\in [-T,T]} \left \| P^\epslam_*(s)\Big(U^\epslam(s,t) \,B\,   U^\epslam(t,s)- U^\epslam_\parallel (s,t) \,B\,  U^\epslam_\parallel ( t,s)\Big) P^\epslam_*(s)\right\| \leq   \epsi\,C_T^\Lambda\, \|B^\Lambda \|\,.
\end{equation}
However, the constant $C_T^\Lambda$ grows, in general with the system size $\Lambda$. Then   the standard adiabatic theorem is of no use, as it stands, if one is interested in large $\Lambda$  or in the thermodynamic limit $|\Lambda|\to \infty$.

As part (a) of the following theorem shows, when restricting to observables that are given by local Hamiltonians, then
\eqref{KatoThm} remains ``almost'' valid with a constant $C_T$ uniform in $\Lambda$. ``Almost'', because  for $B\in \mathcal{L}_{\zeta,2,L}$ the norm $\|B^\epslam\|$ in \eqref{KatoThm}  must be replaced by  the quantity  $C_\zeta\, M^{d-|\ell|}\, \|\Phi_B\|_{\zeta,2,L} $, which, as  explained  after \eqref{normgrowth}, is expected to reflect the growth of $\|B^\epslam\|$ with $\Lambda$ correctly.
On the other hand, if the  time-dependence of the Hamiltonian is also spatially localised, then we can improve on \eqref{KatoThm} by showing that the error grows not like the size of the ``support'' of $B^\epslam$, but only with the size of the intersection of the supports of $B^\epslam$ and~$\dot H^\epslam$, cf.\ Figure~\ref{fig3}. Moreover, part (b) of the following theorem shows that the standard first order corrections to the adiabatic approximation yield an order $\epsi^2$  approximation also in the present setting.

\begin{theorem}\label{AdiThm0} {\bf (Adiabatic Theorem: Leading orders for eigenvalues)}\\   Let the Hamiltonian $H$ satisfy Assumptions  (A1)$_{1,L_H}$ and (A2) for some   $L_H  \in\loc$. Assume      that $\sigma^\epslam_*(t) = \{E^\epslam_*(t)\}$ is an eigenvlaue and that $\frac{\D^n}{\D t^n} H^\Lambda(0)=0$ for all $n=1,\ldots,d+1$.

 Then for any    $T>0$, $\zeta\in\mathcal{S}$, and localization vector $L\in\loc$ with $\ell\cdot \ell_H =0$ there exists a constant $C$, independent of $\epsi$ and $\Lambda$,   such that for any    $B\in \mathcal{L}_{\zeta,2,L}$    the following holds:   
 \begin{itemize}
 \item[(a)] {\rm Adiabatic approximation:}
 \begin{eqnarray}\label{ATstatement0}\lefteqn{\hspace{-40pt}
 \sup_{t\in [-T,T]}\;   \Big\| P^\epslam_*(0) \Big( U^\epslam(0,t) \, B^\epslam \, U^\epslam(t,0) - U^\epslam_\parallel(0,t)\, B^\epslam \,U^\epslam_\parallel(t,0)\Big)P^\epslam_*(0)
 \Big\|}\nonumber\\ &  &\hspace{180pt}\leq \;\epsi \,C\,M^{d-|\ell|-|\ell_H|} \,\|\Phi_B\|_{\zeta,2,L}  \,.
\end{eqnarray}
\item[(b)] {\rm First order superadiabatic approximation:}
\begin{equation}\label{ATstatement11} 
 \sup_{t\in [-T,T]}\;   \Big\| P^\epslam_*(0) \Big( U^\epslam(0,t) \, B^\epslam \, U^\epslam(t,0) -  B^\epslam_{\parallel\,(1)}(t)\Big)P^\epslam_*(0)
 \Big\| \leq \;\epsi^2 \,C\,M^{d-|\ell|-|\ell_H|} \,\|\Phi_B\|_{\zeta,2,L}  \,,
\end{equation}
where
\begin{eqnarray}\label{B1formel}
B^\epslam_{\parallel\,(1)}(t) & :=&   U^\epslam_{\parallel\,(1)} (0,t)\, B^\epslam \,U^\epslam_{\parallel\,(1)}(t,0)\\
&&  +\;  \I\,\epsi\,  U^\epslam_{\parallel}(0,t ) \left( B^\epslam R^\epslam_*(t)\dot P^\epslam_*(t)   - \dot P^\epslam_*(t)  R^\epslam_*(t) B^\epslam\right)  U^\epslam_{\parallel}(t,0 )\nonumber
\end{eqnarray}
and
  $U^\epslam_{\parallel\,(1)}(t,s)$ is the solution of the modified parallel transport equation 
\begin{equation}  
\I  \frac{\D}{\D t} U^\epslam_{\parallel\,(1)} (t,s) =  \left( K^\epslam_\parallel (t) + \epsi  \, \dot P_*^\epslam(t) R_*^\epslam(t) \dot P_*^\epslam(t)   \right) U^\epslam_{\parallel\,(1)}(t,s)\,, \quad U^\epslam_{\parallel\,(1)}(s,s) = {\bf 1} \,,
\end{equation}
with $R_*^\epslam(t) := (H^\epslam(t) -E^\epslam_*(t))^{-1} ({\bf 1} - P^\epslam_*(t))$ denoting the reduced resolvent.
\end{itemize}
\end{theorem}
  \begin{SCfigure}[1.4] \label{fig3}
   \includegraphics[width=0.31\textwidth]{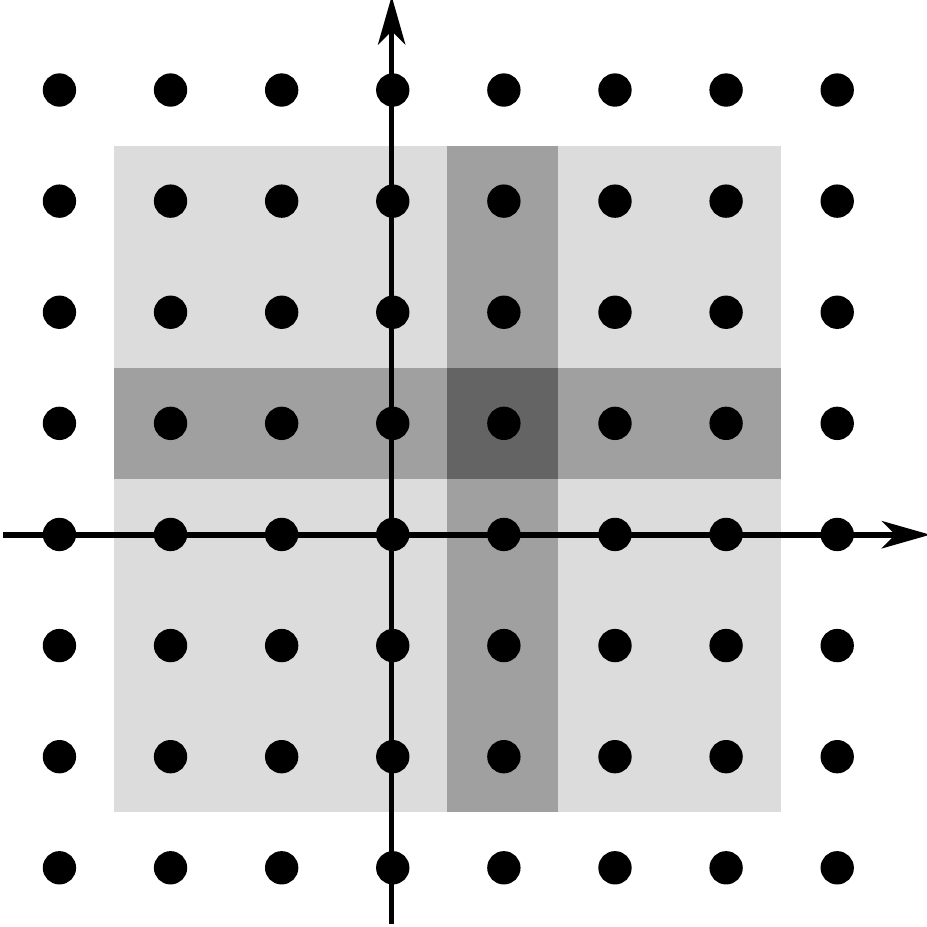}
    \caption{\small In the figure we have $d=2$, $|\ell|=1$,   $|\ell_H|=1$, and $\ell\cdot\ell_H=0$. That means  the support of $B^\Lambda$ is  localised around a line, say the horizontal shaded one,  and the support of $\dot H^\Lambda$ is   localised around   the vertical one. The intersection of the two regions, the darker shaded area, does not grow with $\Lambda$ and thus  the factor $M^{d-|\ell|-|\ell_H|} = M^0=  1$ in the error estimates in Theorem~\ref{AdiThm0} does not grow with $\Lambda$ either.
    } 
\end{SCfigure}
In Section~\ref{sec:AdiCurrent} we will use part (b) of Theorem~\ref{AdiThm0} to prove  response formulas for adiabatic currents uniformly in the system size.
The proof of Theorem~\ref{AdiThm0} is based on  a   general superadiabatic theorem that we formulate now and that generalises the above result in several ways. First of all, we replace the parallel transport $U^\epslam_\parallel(t)$ by a superadiabatic time-evolution $U^\epslam_{\rm sa}(t)$ and thereby obtain an error estimate of order~$\epsi^m$ on the whole space.
We also allow for more general gapped parts of the spectrum instead of only  eigenvalues.   Finally, we prove explicit asymptotic expansions when restricting to the range of $P_*^\epslam(t)$, which, as a special case, yield  the statements of Theorem~\ref{AdiThm0}.  

The superadiabatic evolution is constructed from two ingredients. The first one is a modified parallel transport $U^\epslam_{\rm a}$, called the  adiabatic evolution in the following, satisfying 
\begin{equation}\label{adiabaticevolution}
\I\epsi \frac{\D}{\D t} U^\epslam_{\rm a}(t,s) = H^\epslam_{\rm a}(t) \,U^\epslam_{\rm a}(t,s)\,, \quad U^\epslam_{\rm a}(s,s) = {\bf 1}\;\mbox{ for }\; 
 t,s\in \R\,,
\end{equation}
with generator 
\begin{equation}\label{HaDef}
H^\epslam_{\rm a}(t) := H^\epslam(t)   + \epsi K^\epslam(t)\,.
\end{equation}
 Note that we now need   to incorporate $H^\epslam(t) $ in the   effective evolution since we no longer restrict ourselves to spectral subspaces corresponding to a single eigenvalue (where $ H^\epslam(t)P^\epslam_*(t)$ acts trivially by multiplication with the eigenvalue).  
 Moreover, the generator $K^\epslam_\parallel (t)$ of parallel transport is  replaced  by $K^\epslam(t)\equiv K^\epslaml(t)$ in such a way   that   $[ K^\epslam_\parallel (t) - K^\epslam(t), P^\epslam_*(t)]=0$ and thus
  $U^\epslam_{\rm a}$  still intertwines the spectral subspaces exactly,  
 \begin{equation}\label{adievo2}
U^\epslam_{\rm a} (t,s) \,P^\epslam_*(s) = P^\epslam_*(t)\,U^\epslam_{\rm a} (t,s)\,.
\end{equation}
The second ingredient is
the superadiabatic near-identity transformation  $V^\epslam(t)$ that unitarily maps the spectral projection  $P_*^\epslam(t)$ to the so-called superadiabatic projection 
\begin{equation}\label{superadipro}
P^\epslam_{\rm sa}(t)  := V^\epslam(t)\, P_*^\epslam(t)\,V^\epslam(t)^*\,.
\end{equation}
It is well known (see e.g.\ \cite{Nenciu} and references therein) that the  full evolution $U^\epslam $ generated by $H^\epslam(t)$ intertwines the instantaneous spectral subspaces ran$P^\epslam_*(t)$ only up to errors of order $\epsi$, i.e.\ that the leakage out of these subspaces is of order $\epsi$ and the error term in \eqref{KatoThm} can not be improved in general. To obtain improved error estimates one has to track solutions within the superadiabatic subspaces ran$P^\epslam_{\rm sa}(t)$ that are intertwined by $U^\epslam (t)$ up to much smaller errors.

With these two ingredients  the superadiabatic evolution is defined by
\begin{equation}\label{superadi}
U^\epslam_{\rm sa}(t,s) := V^\epslam(t)\, U^\epslam_{\rm a}(t,s)\, V^\epslam(s)^* 
\end{equation}
and intertwines, by construction,  the superadiabatic subspaces,
 \begin{equation}\label{superadievo2}
U^\epslam_{\rm sa} (t,s) \,P^\epslam_{\rm sa}(s) = P^\epslam_{\rm sa}(t)\,U^\epslam_{\rm sa} (t,s)\,.
\end{equation}
In Proposition~\ref{SuperLemma} we will construct  $K^\epslam(t)$ and $V^\epslam(t)$ in such a way that
 \begin{equation}\label{supadieqvorn}
 \I\epsi    \frac{\D}{\D t}
 U^\epslam_{\rm sa} (t,s)   =\Big( H^\epslam(t) + \epsi^{m+d+1}  \,R^\epslam(t)\Big)U^\epslam_{\rm sa} (t,s) \,,
   \end{equation}
where the remainder term $R^\epslam(t)$ is a local Hamiltonian. Note that  \eqref{supadieqvorn} implies immediately that
\[
\|   U^\epslam  (t,s)  -  U^\epslam_{\rm sa} (t,s)   \| = \Or(\epsi^{m+d } \|R^\epslam\|)= \Or(\epsi^{m+d } \,|\Lambda|)\,,
\] 
since, as a local Hamiltonian,  $\|R^\epslam\|$ grows at most like $|\Lambda|$.
Again, the following theorem shows that when considering the Heisenberg time evolution of  local Hamiltonians $B$, then the factor $|\Lambda|$ coming from $\|R^\epslam\|$ is absent and the only growth of the approximation error with $|\Lambda|$ comes from $\|B^\epslam\|$.

\begin{theorem}\label{AdiThm}{\bf  (Superadiabatic Theorem) }\\
 Let $H$ satisfy Assumptions (A1)$_{m,L_H}$ and (A2) for some   $m\in\N$ and $L_H  \in\loc$.  
 There exist local Hamiltonians $   S  (t),   K  (t) \in \mathcal{L}_{\mathcal{S},\infty,L_H}$, $t\in \R$,
 such that
\begin{itemize}
\item  $t\mapsto   S^\epslam(t) , K^\epslam(t)  \in C^{1} (\R,  \mathcal{A}^\mathfrak{N}  _\Lambda)$ for all $\epsi,\Lambda$, and
\item the adiabatic evolution $U^\epslam_{\rm  a} (t,s)$ defined in \eqref{adiabaticevolution}  satisfies \eqref{adievo2} and the superadiabatic evolution $U^\epslam_{\rm sa} (t,s)$ defined in \eqref{superadi} 
 with  
 \begin{equation}\label{VDef}
 V^\epslam(t) := \E^{\I\epsi S^\epslam(t)}
 \end{equation}
  has the following properties:
\end{itemize}
 
For any    $T>0$, $\zeta\in\mathcal{S}$, and $L\in\loc$ with $\ell\cdot \ell_H =0$ there exists a constant~$C$, independent of $\epsi$ and $\Lambda$,   such that for any     $B\in \mathcal{L}_{\zeta,2,L}$           it holds that 
   \begin{equation}\label{ATstatement}
   \sup_{t,s\in [-T,T]}\;   \Big\|   U^\epslam(s,t) \, B^\epslam \, U^\epslam(t,s) - U^\epslam_{\rm sa}(s,t) \,B^\epslam\, U^\epslam_{\rm sa}(t,s)
 \Big\| \leq \epsi^m \,C\,M^{d-|\ell|-|\ell_H|} \,\|\Phi_B\|_{\zeta,2,L}  \,.
\end{equation}
If at some time $t'\in \R$ it holds that $\frac{\D^n}{\D t^n} H^\Lambda(t')=0$ for all $n=1,\ldots, m+d$, then 
$S^\epslam(t') = K^\epslam(t') =0$ and thus, in particular, $V^\epslam(t') = \mathrm{id}$ and $P_{\rm sa}^\epslam(t') = P_*^\epslam(t')$.
\end{theorem}

\begin{remarks}\rm
\begin{enumerate}
\item As an immediate consequence of \eqref{ATstatement} we find for any initial state $\rho_0^\Lambda$  (at initial time $t=s$) the full evolution is well approximated by the superadiabatic one when testing against observables $B^\Lambda$ given by local Hamiltonians in the appropriate trace ``per unit volume'':
\[
 \sup_M  \;\sup_{|t|\leq T}\;  \frac{1}{M^{d -|\ell|-|\ell_H|}} \left|  \tr  \left(  \left(U^\epslam(t,s) \rho_0^\Lambda U^\epslam(s,t) - U^\epslam_{\rm sa}(t,s) \rho_0^\Lambda  U^\epslam_{\rm sa}(s,t)^*
\right)    B^\epslam\right)\right| \leq  \epsi^m \,C \,\|\Phi_B\|_{\zeta,2,L}  \,.
\]
In particular, if initially $\rho_0^\Lambda = P_{\rm sa}^\epslam(s) \rho_0^\Lambda P_{\rm sa}^\epslam(s)$ lives in the range of $P_{\rm sa}^\epslam(s)$, then $\rho ^\Lambda(t) := U^\epslam(t,s) \rho_0^\Lambda U^\epslam(s,t)$ is well approximated by the state $\rho_{\rm sa}^\Lambda(t):= U^\epslam_{\rm sa}(t,s) \rho_0^\Lambda  U^\epslam_{\rm sa}(s,t)$ that lives in the range of $P_{\rm sa}^\epslam(t)$, i.e.\ satisfies $\rho_{\rm sa}^\Lambda(t) = P_{\rm sa}^\epslam(t) \rho_{\rm sa}^\Lambda(t) P_{\rm sa}^\epslam(t)$. In this sense the superadiabatic subspaces are almost invariant for the time-evolution generated by $H^\Lambda(t)$.
\item
The superadiabatic projection $ P^\epslam_{\rm sa}(t)$ agrees at any order with the superadiabatic projection constructed by the standard adiabatic expansion, cf.\ e.g.\ \cite{Nenciu}. The crucial novelty here,  as in \cite{BDF}, is that the error terms are uniform in the system size $|\Lambda|$ even for interacting systems.
 \item One might wonder why it is of any interest to replace the full time evolution $U^\epslam(t)$ by the superadiabatic time evolution $U^\epslam_{\rm sa}(t)$, although the latter seems even more complicated. Because the superadiabatic time evolution intertwines the superadiabatic subspaces exactly, cf.\ \eqref{superadievo2},   one can restrict the  superadiabatic time evolution to these subspaces. Moreover, if the range of $P^\Lambda_*(t)$ has finite dimension uniformly in $\Lambda$, then also $ P^\epslam_{\rm sa}(t)$ has finite-dimensional range and the action of
   $H^\Lambda_{\rm sa} (t) $ on it might be   more accesible than that of $H^\Lambda(t)$ on the full Hilbert space.
\item \label{compremark}As was mentioned above, the proof of Theorem~\ref{AdiThm}  given below  is based  on a key proposition, Proposition~\ref{SuperLemma}, that is proved in Section~\ref{PropProof}. Both proofs, the one of Theorem~\ref{AdiThm} and the one of Proposition~\ref{SuperLemma},
rely heavily on ideas from and technical lemmas proved in  \cite{BDF}. However, many small and several substantial changes in the arguments are necessary to arrive at Theorems~\ref{AdiThm0} and \ref{AdiThm}. 
Let us briefly comment on these changes. 

The step from spin systems to fermions on a lattice is   straightforward, in particular, since Lieb--Robinson bounds are readily available also for fermions, cf. e.g.\ \cite{NSY, BD}. 
 The change from bounded subsets $\Lambda\subset \Gamma$ to the ``torus'' $\Lambda$ enters only in the proof of the Lieb--Robinson bound. For this reason, in Appendix~\ref{AppendixLR}  we  state the Lieb--Robinson bound for systems on the torus and briefly discuss the small necessary modifications in the proof. 

 One novelty of our result compared to \cite{BDF} is the explicit  treatment of arbitrary densities, including the trace per unit volume\footnote{The possibility to treat also densities  for translation invariant Hamiltonians, observables, and  states was already indicated in \cite{BDF}. Note that we make no such assumption on translation invariance at all.}, and localized driving. This change poses new technical problems and we need to adapt and   extend several technical lemmas from \cite{BDF} to  local $L$-localized Hamiltonians in Appendix~\ref{AppendixTech}. 
 
The second  novelty is the superadiabatic tracking of the solution within the spectral subspace\footnote{Note that a statement about the leading order approximation to the adiabatic evolution within degenerate eigenspaces has been added in a later version of \cite{BDF}  after the first version of our paper appeared.} and the formula \eqref{supadieqvorn}\footnote{In the second equation in Section~2.9 of \cite{BDF}  a seemingly similar claim is formulated, namely that
 \begin{equation}\label{supadieqwrong}
 \I\epsi    \frac{\D}{\D t}
 V^\epslam(t)   =\Big( H^\Lambda(t) + \epsi^{m+d+1}  \,R^\epslam(t)\Big)V^\epslam  (t )   \end{equation}
 with $V^\Lambda(t)$ as in \eqref{VDef}.
  Note that \eqref{supadieqwrong} is clearly wrong and just a lapse in the presentation of \cite{BDF}.}. Among other things, in our approach  this requires a slight modification of 
  a certain map $\mathcal{I}_H  $ introduced in \cite{HW,BMNS} in order to invert the Liouvillian $\ad_H(\cdot):= - \I \,[H,\,\cdot\, ]$.
  In Appendix~\ref{AppendixI} 
  we provide this analysis and also prove an  explicit formula for the action of $\mathcal{I}_H$  in the case of $\sigma_*^\Lambda(t) = \{E^\Lambda(t)\}$,  expressed in terms of the reduced resolvent of the Hamiltonian, which is used to formulate Corollary~\ref{ExpCor} and the formulas for adiabatic currents in the next section without  $\mathcal{I}_H$ appearing in the statements.  
  Finally, Appendix~\ref{Sapp} collects a few useful properties of functions in $\mathcal{S}$, defined in \eqref{eqn:S_def}, that are used throughout the paper but not completely obvious.  
  \end{enumerate}
 \end{remarks}
 
 The proof of Theorem~\ref{AdiThm} is based on the following proposition, in which the ingredients for the superadiabatic evolution are constructed. We need to state it here, because otherwise we couldn't properly formulate  Theorem~\ref{expThm} below concerning explicit expansions. In the following, we abbreviate $P_*^\epslam(t)^\perp := {\bf 1}-P_*^\epslam(t)$.

\begin{proposition}\label{SuperLemma}
 Let the Hamiltonian $H$ satisfy Assumptions (A1)$_{m,L_H}$ and (A2) for some fixed $m$.  There are 
 self-adjoint operators 
 \begin{equation}\label{KSDef}
   K^\epslam(t) = \sum_{\mu = 1}^{m+d} \epsi^{\mu-1} K^\epslam_\mu(t) \qquad\mbox { and } \qquad  S^\epslam(t) = \sum_{\mu = 1}^{m+d} \epsi^{\mu-1}  A^\epslam_\mu(t) \,,
 \end{equation}
   with $K_\mu,  A_\mu\in \mathcal{L}_{\mathcal{S},\infty,L_H}$ for $1\leq \mu\leq  m+d$, such that 
\begin{itemize}
\item  the adiabatic evolution $U^\epslam_{\rm  a} (t,s)$ defined in \eqref{adiabaticevolution}  satisfies \eqref{adievo2}, and 
\item the superadiabatic evolution $U^\epslam_{\rm sa} (t,s)$ defined in \eqref{superadi}    is close to the   time evolution $U^\epslam(t,s)$ defined in \eqref{Hdynamics}
    in the following sense:
\end{itemize}
  There exists a time dependent local Hamiltonian    $R^\epslam(t)\in  \mathcal{L}_{\mathcal{S},\infty,L_H}$, $t\in\R$, such that
    \begin{equation}\label{supadieq}
 \I\epsi    \frac{\D}{\D t}
 U^\epslam_{\rm sa} (t,s)   =\Big( H^\epslam(t) + \epsi^{m+d+1}  \,R^\epslam(t)\Big)U^\epslam_{\rm sa} (t,s) \,.
   \end{equation}
In addition we have:
\begin{itemize}
\item[(a)]   
   If at some time $t'\in \R$ it holds that $\frac{\D^n}{\D t^n} H^\Lambda(t')=0$ for all $n=1,\ldots, m+d$, then 
$S^\epslam(t') = K^\epslam(t') =0$ and thus, in particular,  $P_{\rm sa}^\epslam(t') = P_*^\epslam(t')$.
\item[(b)]   The off-diagonal part   $K_1^{ \Lambda}(t)^{\rm OD}  := P_*^\Lambda(t) K_1^\Lambda(t) P_*^\Lambda(t)^\perp + P_*^\Lambda(t)^\perp  K_1^\Lambda(t) P_*^\Lambda(t)$ of $K_1$ equals Kato's generator $K_\parallel$, i.e.\
\[
[K_1^\epslam(t) ,P_*^\epslam(t) ] = [K_\parallel^\epslam(t) ,P_*^\epslam(t) ] \,.
\]
If $\delta<g$ in Assumption (A2), then  $ P_*^\epslam(t) K_1^\epslam(t) P_*^\epslam(t) \equiv 0$.\\
\item[(c)] 
If $\delta =0$, i.e.\  $\sigma^\epslam_*(t) = \{E^\epslam_*(t)\}$, then  the relevant blocks of   $A_1 =: \tilde A_1 + P_*^\perp A_1P_*^\perp$ and   $K_2 =: \tilde K_2 + P_*^\perp K_2P_*^\perp$ are 
\begin{equation}\label{A1}
\tilde A_1^\epslam(t) :=   P^\epslam_*(t)\dot P^\epslam_*(t)  R^\epslam_*(t) + R^\epslam_*(t)\dot P^\epslam_*(t) P^\epslam_*(t)   
\end{equation}
and 
\begin{equation}\label{K2}
 \tilde  K_2^\epslam(t) = P_*^\epslam(t) \dot P_*^\epslam(t) R_*^\epslam(t) \dot P_*^\epslam(t)  P_*^\epslam(t)\,,
\end{equation}
where
$R^\epslam_*(t) := (H^\epslam(t)-E_*^\epslam(t))^{-1} P_*^\epslam(t)^\perp$ denotes the reduced resolvent.
\end{itemize}
 \end{proposition}
\medskip

\begin{remark}\rm
An alternative formula for $\tilde A_1^\epslam(t)$ that emerges from the proof is
\[
\tilde A_1^\epslam(t) =      (R_*^{ \Lambda}(t))^2     \dot H^{ \Lambda}(t)^{\rm OD} +
 \dot H^{ \Lambda}(t)^{\rm OD} (R_*^{ \Lambda}(t))^2
 \,.
 \]
 \end{remark}

While the superadiabatic approximation is of conceptual interest in itself, for applications one needs explicit expansions of $ U^\epslam(s,t) \, B^\epslam \, U^\epslam(t,s)  \approx U^\epslam_{\rm sa}(s,t) \, B^\epslam \, U^\epslam_{\rm sa}(t,s) $. The following theorem is at the basis of computing such expansions.

\begin{theorem}\label{expThm}{\bf (Asymptotic expansions)}\\
 Let $H$ satisfy Assumptions (A1)$_{m,L_H}$ and (A2) for some   $m\in\N$, $L_H  \in\loc$, $T\geq 0$, and let
$  S  , K $ be as in Proposition~\ref{SuperLemma}.  
Let  $L\in\loc$ with $\ell\cdot \ell_H =0$  and define for  $1\leq  k< m+d$
\[
S^\epslam_{(k)}(t) := \sum_{\mu=1}^{k} \epsi^{\mu-1} A^\epslam_\mu(t)\quad\mbox{ and }\quad 
K^\epslam_{(k)}(t) :=  K^\epslam_\parallel +  \sum_{\mu=2}^{k+1} \epsi^{\mu-1} P^\epslam_*(t) K^\epslam_\mu(t)P^\epslam_*(t)  \,.
\]
Note that in the definition of $K^\epslam_{(k)}(t)$ there appears Kato's generator $K^\epslam_\parallel$ and not $K^\epslam_1$.
Moreover, let $U^\epslam_{\parallel\,(k)}(t,s)$ be the 
solution of
 \begin{equation}  \label{Uparallel}
\I  \frac{\D}{\D t} U^\epslam_{\parallel\,(k)} (t,s) =   K^\epslam_{(k)} (t)  \,U^\epslam_{\parallel\,(k)}(t,s)\,, \quad U^\epslam_{\parallel\,(k)} (s,s) = {\bf 1}\;\mbox{ for }\; 
 t,s\in \R\,.
\end{equation}

Then there exists a constant $C<\infty$ independent of $\epsi$ and $\Lambda$  such that  for any $B\in \mathcal{L}_{\zeta,k+1,L}$ the following statements hold:
\begin{enumerate}
\item[\rm(a)] {\rm Expansion of the superadiabatic transformation:}\\
 \begin{align}\label{Vexp}
 \sup_{|t|\leq T}\;  \left\| V^\epslam(t)^* B^\epslam\, V^\epslam(t)- \sum_{j=0}^k
    \frac{ (-\I \epsi)^j}{j!}  \, \add_{S^\epslam_{(k+1-j)}(t)}^j (B^\epslam)  \right\|
  \leq  \epsi^{k+1} C  \|\Phi_B\|_{\zeta,k+1,L} M^{d -|\ell|-|\ell_H|}\,,
\end{align}
where $\add^j_S(B) := [  S,[ S,\ldots,[ S,B] ]]$ 
with $S$ appearing $j$ times.
\item[\rm(b)]  {\rm Expansion of the adiabatic evolution:}\\
Assume, in addition, that  the upper bound $\delta$ on the width of the spectral patch $\sigma_*^\epslam(t)$ in Assumption~(A2) satisfies $\delta<g$. Then 
\begin{align}
 \sup_{s,t\in[-T,T]}\;  \Big\| P_*^\epslam(s) \Big( U^\epslam_{\rm a} (s,t) & B^\epslam\,  U^\epslam_{\rm a} (t,s)  -  
      U^\epslam_{\parallel\,(k)} (s,t)  B^\epslam\,  U^\epslam_{\parallel\,(k)} (t,s) \Big) P_*^\epslam(s) \Big\|\nonumber\\& \label{Uaexp}
  \leq C\left( \frac{\delta}{\epsi}M^{d -|\ell| } + \epsi^{k }    M^{d -|\ell|-|\ell_H|}\right) \|\Phi_B\|_{\zeta,k+1,L}\,.
\end{align}
 Note that if $\sigma_*^\epslam(t) =\{E_*^\epslam(t)\}$ is a single (possibly degenerate) eigenvalue, then $\delta =0$.
\end{enumerate}
\end{theorem}

\begin{remark}\rm
Note that the first term on the right hand side of \eqref{Uaexp} is only small if $\delta$ is small. For general spectral patches $\sigma_*^\epslam(t)$ this estimate is of no use. However, there are  situations where $\sigma_*^\epslam(t)$ is an almost degenerate ground state where   $\delta\sim M^{-\infty}$ goes to zero faster than any inverse power of the system size $M$. Then the estimate \eqref{Uaexp} remains useful in situations where one first takes the thermodynamic limit $M\to \infty$ and only afterwards the adiabatic limit $\epsi\to 0$. This is for example the case when computing  linear response formulas.
\end{remark}

In the following corollary we exemplify how to combine the expansions of Theorem~\ref{expThm} in order to approximate the superadiabatic Heisenberg evolution of observables.

\begin{corollary}\label{ExpCor} 
Under the assumptions of Theorem~\ref{expThm} and the additional condition that $\sigma_*^\epslam(t) =\{E_*^\epslam(t)\}$ is a single eigenvalue, the following holds.
Let
\[
B^\epslam_{ (0)}(t,s)  := U^\epslam_{\parallel }(s,t) B^\epslam U^\epslam_{\parallel }(t,s) 
\]
and
\[
B^\epslam_{ (1)}(t,s)  :=  U^\epslam_{\parallel\,{(1)}} (s,t) B^\epslam U^\epslam_{\parallel\,{(1)}} (t,s)  - \I\epsi  U^\epslam_{\parallel}(s,t ) [ \tilde A_1^\epslam(t),B^\epslam]  U^\epslam_{\parallel}(t,s )
\,,
\]
where $\tilde A_1^\epslam(t)$ is explicitly given in \eqref{A1} and the generator of $U^\epslam_{\parallel\,{(1)}} (t,s)$ is $K^\epslam_{ {(1)}} (t ) = K_\parallel^\epslam + \epsi \tilde K_2^\epslam(t)$ and is explicitly given in \eqref{Kato} and \eqref{K2}.
Then 
 \begin{equation} \label{spezialfall}
   \sup_{t,s\in [-T,T]}\;   \Big\|   P_{\rm sa}^\epslam(s)\Big(    U^\epslam(s,t) \, B^\epslam \, U^\epslam(t,s) - V^\epslam(s)  B^\epslam_{ (0)}(t,s)V^\epslam(s)^*\Big) P_{\rm sa}^\epslam(s)
 \Big\| \leq \epsi  \,C\,M^{d-|\ell|-|\ell_H|} \,\|\Phi_B\|_{\zeta,2,L}   
 \end{equation}
and 
\begin{equation} 
   \sup_{t,s\in [-T,T]}\;   \Big\|   P_{\rm sa}^\epslam(s)\Big(    U^\epslam(s,t) \, B^\epslam \, U^\epslam(t,s) -  V^\epslam(s)  B^\epslam_{ (1)}(t,s)V^\epslam(s)^*\Big) P_{\rm sa}^\epslam(s)
 \Big\| \leq \epsi^2  \,C\,M^{d-|\ell|-|\ell_H|} \,\|\Phi_B\|_{\zeta,2,L}   \,.
 \end{equation}
 
  \end{corollary}
\begin{proof}
This follows from a straightforward combination of the previous estimates,
\begin{eqnarray*}\lefteqn{
P_{\rm sa} (s)    U (s,t) \, B \, U(t,s) P_{\rm sa}(s) \stackrel{\eqref{ATstatement}}{=}
P_{\rm sa}(s)    U_{\rm sa}(s,t) \, B \, U_{\rm sa}(t,s) P_{\rm sa}(s)+ R_1}\\
&=& V(s)\, P_*(s)    U_{\rm a}(s,t)\, \big(V(t)^*  \, B \, 
V(t)\big) \, U_{\rm a}(t,s)\, P_*(s)V(s)^* +R_1\\
&\stackrel{\eqref{Vexp}}{=}& V(s)\,  P_*(s)  U_{\rm a}(s,t)\,  \big(B - \I\epsi [  A_1(t),B]\big) \,  U_{\rm a}(t,s)\,P_*(s) V(s)^* +  R_2\\
&\stackrel{\eqref{Uaexp}}{=}& V(s)\,  P_*(s) \left(  U_{\parallel\,{(1)}}(s,t)\,  B\, U_{\parallel\,{(1)}}(t,s ) -  \I\epsi  U_{\parallel }(s,t ) [  A_1(t), B]  U_{\parallel}(t,s ) \right)   P_*(s) V(s)^* +R_3\\
&=& V(s)\,  P_*(s) \left(  U_{\parallel\,{(1)}}(s,t)\,  B\, U_{\parallel\,{(1)}}(t,s ) -  \I\epsi  U_{\parallel}(s,t )P_*(t) [  A_1(t), B] P_*(t) U_{\parallel}(t,s ) \right)   P_*(s) V(s)^* +R_3
 \\&\stackrel{\eqref{A1}}{=}&  P_{\rm sa} (s) V(s) B _{ (1)}(t,s) V(s)^* P_{\rm sa} (s) +R_3\,,
\end{eqnarray*}
where all the remainder terms {$R_{1,2,3}$} are bounded in norm by $\epsi^2  \,C\,M^{d-|\ell|-|\ell_H|} \,\|\Phi_B\|_{\zeta,2,L}  $.
 \end{proof}

Note that Theorem~\ref{AdiThm0} follows from \eqref{spezialfall} using that  {the condition} $\frac{\D^n}{\D t^n} H^\Lambda(0)=0$ for all $n=1,\ldots,d+1$ implies $P^\epslam_{\rm sa}(0) = P^\epslam_*(0)$ and $V^\epslam(0)={\bf 1}$.
 
The other  proofs are organised as follows. In Section~\ref{PropProof} we construct  the adiabatic expansion and prove Proposition~\ref{SuperLemma}. In Section~\ref{sec:proofAdiabatic}, based on Proposition~\ref{SuperLemma}, we prove Theorem~\ref{AdiThm} and then Theorem~\ref{expThm} (which then implies Corollary~\ref{ExpCor}).
In Section~\ref{sec:AdiCurrent} we first discuss  applications of the above results in the context of adiabatic charge transport. 
\bigskip

\section{Adiabatic  currents and quantum Hall systems}\label{sec:AdiCurrent}

 In this section we apply Theorem~\ref{AdiThm} and its Corollary~\ref{ExpCor} in order to compute currents and  current densities induced by   adiabatic changes of a Hamiltonian when the system starts in its gapped ground state. Then we briefly discuss the application to conductivity and conductance in quantum Hall systems.
  
First note that in general the  total current operator  on a torus $\Lambda$ is only well defined for Hamiltonians with finite-range hoppings, since for a long-range hop on a torus the direction of the hop might be ambiguous. This is related to the fact that
there is no ``good'' position operator $Q$ on the torus   that yields the  current operator in the form $J = \I [H,Q]$ for general $H$. Thus  we     restrict ourselves to Hamiltonians $H$ that are uniformly finite range uniformly in time.
Recall that this means, in particular, that there is a uniform bound on the size of the sets $X\subset\Lambda$ where $\Phi^\Lambda_H(X)$ does not vanish, i.e.\  there exists a number $r\in\N$ such that 
\begin{equation} \label{eqn:r}
\sup_{t\in[0,\infty)} \sup_{\Lambda} \max \{ |X|\,|\, X \in \mathcal{F}_{H(t)}(\Lambda) \} \, \leq r  \,, \;\text{ where }\; \mathcal{F}_{H(t)}(\Lambda):= \{X \subset \Lambda\,|\, \Phi^\Lambda_H (t,X)\not=0\}.
\end{equation}
 Hence, if $\Lambda$ is sufficiently large, for each $X\in \mathcal{F}_{H(t)}(\Lambda)$ and any point $y\in X$ it holds that $  X \lminus y \subset \{-M+2,\ldots, M-1\}^d$, i.e.\ the shifted set $X$  does not ``cross the boundary'' of   $\Lambda$. With the help of the shifted position operator  
\[
Q^\Lambda_y   := \sum_{x\in \Lambda} (x\lminus y)  a_x^*a_x \in \mathcal{A}_\Lambda^\mathfrak{N}
\]
``centered'' at $y$
we can now define the   interaction of the microscopic current operator as 
\[
\Phi^\Lambda_J(t,X) := \begin{cases}
\I [ \Phi_H^\Lambda (t,X), Q^\Lambda_{y  }] & \mbox{for } X \in \mathcal{F}_{H(t)}(\Lambda) \mbox{ and any }y \in X\,,\\
0 & \mbox{for } X \in \mathcal{F}(\Lambda) \setminus \mathcal{F}_{H(t)}(\Lambda)\,.
\end{cases}
\]
Note that the definition is independent of the choice of $y \in X$ because for no $y\in X$ does the set $X$   overlap the set where the shifted position operator $Q^\Lambda_{ y }$ is discontinuous.
The  current operator on $\Lambda$ is defined accordingly as
\[
J^\Lambda(t) =  \sum_{X \subset \Lambda} \Phi_J^{\Lambda}(t,X)  \in \mathcal{A}_\Lambda^\mathfrak{N}\,.
\]
 Since  $H(t)$ is uniformly finite range, also  $J(t)$ has this property.
 For a Hamiltonian of the form \eqref{ExHamiltonian} the current operator is explicitly given by\footnote{Note that the following expression makes also sense if $T$ is not compactly supported but only exponentially decaying, and one could use it as a definition of the current operator in this specific case.}
\[
J_{TVW}^\Lambda (t) = -\I \sum_{(x,y)\in  \Lambda^2}   (x\lminus y) \, a^*_x \,T(t,x\lminus y) \,a_y\,.
\]
For the discussion of currents it is more transparent if we shift the time-evolution to states. In the following we write
\[
 \rho^\epslam(t) := U^\epslam(t,0) \rho_0^\Lambda U^\epslam(0,t)\,, \;\; \rho_\parallel^\epslam(t) := U_\parallel^\epslam(t,0) \rho_0^\Lambda U_\parallel^\epslam(0,t)\,,
 \;\; \mbox{and} \;\;  \rho_{\rm sa}^\epslam(t) := U_{\rm sa}^\epslam(t,0) \rho_0^\Lambda U_{\rm sa}^\epslam(0,t)
\]
for the full resp.\ adiabatic resp.\ super-adiabatic evolution of an initial state $\rho_0^\epslam$. The \emph{current density} (in the macroscopic time scale) at time $t$ is then, by definition,
\begin{equation} \label{eqn:curr_dens}
\mathcal{J}^\epslam(t) := \frac{1}{\epsi |\Lambda|}   \tr  \left(\rho^\epslam(t)  J^\Lambda(t) \right) \,.
\end{equation}
 To make contact to   certain formulas for $\mathcal{J}^\epslam(t)$ that are widespread in the literature {(see e.g.\ \cite{NT,AS,HM,GMP})}, we introduce the family of twisted Hamiltonians $\{H(\alpha)\}_{\alpha\in\R^d}$ defined by the twisted interactions
\begin{equation}\label{Halpha}
\Phi^\Lambda_H(\alpha,t,X) := \begin{cases}
 \E^{-\I \alpha\cdot Q^\Lambda_{y  }}\, \Phi_H^\Lambda (t,X) \,\E^{\I \alpha\cdot Q^\Lambda_{y  }} & \mbox{for } X \in \mathcal{F}_{H(t)}(\Lambda) \mbox{ and any }y \in X\,,\\
0 & \mbox{for } X \in \mathcal{F}(\Lambda) \setminus \mathcal{F}_{H(t)}(\Lambda)\,.
\end{cases}
\end{equation}
Then
\[
J^\Lambda(t)  =  \nabla_\alpha H^\Lambda(\alpha,t)\big|_{\alpha=0},
\]
and, by standard perturbation theory, the ground state projection $P^\Lambda_*(\alpha,t)$ is a differentiable function of $\alpha$ for $\alpha$ in a possibly $\Lambda$-dependent neighborhood of $ 0\in\R^d$.

 As a corollary of the adiabatic theorem, Theorem~\ref{AdiThm0} (b), we can now easily show that the  current density  is given by one of the standard formulas used in the physics and mathematics literature as a definition of the adiabatic current density in such systems. In this very general setting, however, we have to add one more assumption, namely the \emph{vanishing of persistent currents} in the system. More precisely, we assume that for any   ground state {projection}  $ G ^\Lambda(t)$, i.e.\  $P_*^\Lambda(t) G ^\Lambda(t) P_*^\Lambda(t) = G ^\Lambda(t)$, the stationary current vanishes, 
 \begin{equation}\label{persvanish}
  \tr  \left( G ^\Lambda(t) J^\Lambda(t) \right) = 0 \qquad \mbox{ for all $\Lambda$.} 
 \end{equation}
 That means that in such a system the only current flowing is the one induced by the change of the Hamiltonian.\footnote{Alternatively, we could take  a point of view that is often taken  in response theory and compute the relative quantity
\[ \frac{1}{\epsi |\Lambda|}   \tr  \left(\left(\rho^\epslam(t)  - \rho^\epslam_{{\rm sa},(0) }(t) \right) J^\Lambda(t) \right) \,. \] 
 That is, we are only interested in the current induced by the change of the Hamiltonian and not in the persistent current flowing  through the system even in the stationary state.}
 A sufficient condition for \eqref{persvanish} to hold in the case of a non-degenerate ground state, i.e.\ ${\dim} \operatorname{ran} P_*^\Lambda(t)\equiv1$,  is space-inversion symmetry.

\begin{corollary}\label{currentcor}
Let the Hamiltonian satisfy conditions  (A1)$_m$  and (A2), $\sigma_*^\Lambda(t) = \{E^\Lambda(t)\}$, and  assume that $H$ is uniformly finite range. Then for every $T>0$ there is a constant $C>0$ such that
\begin{equation}\label{corstat1}
\sup_{\Lambda(M):M\geq M_0} \,\sup_{t\in[-T,T]}\left|
\frac{1}{\epsi |\Lambda|}   \tr  \left(\rho^\epslam(t)  J^\Lambda(t) \right)  -  \frac{1}{ \epsi |\Lambda|}   \tr   \left(\rho^\epslam_{{\rm sa} }(t)  J^\Lambda(t) \right) \right| \leq \epsi^{ m-1 }\, C\,.
\end{equation}
Assume, in addition, $\rho_0^\epslam = P_*^\epslam(0) \rho_0^\epslam P_*^\epslam(0)$   and that the system  admits no persistent currents, i.e.\ that \eqref{persvanish} holds, then
\begin{eqnarray}\label{corstat1bis} 
\mathcal{J}^\epslam (t) &= &\frac{\I}{|\Lambda|} \tr  \left( \rho^\Lambda_\parallel(t)    \left( J^\epslam(t) R^\epslam_*(t)\dot P^\epslam_*(t)   - \dot P^\epslam_*(t)  R^\epslam_*(t) J^\epslam(t)\right)
\right)
 +\Or(\epsi)\\\label{corstat1bis2} 
&= &  \frac{\I}{|\Lambda|}  \tr \left(\rho^\Lambda_\parallel(t)  \left[  \dot P^\epslam_*(t) ,\nabla_\alpha P^\epslam_*(t) |_{\alpha=0} \right] \right) +\Or(\epsi)
\end{eqnarray}
uniformly in the system size $|\Lambda|$ and on any bounded time interval $[-T,T]$.

\end{corollary}
\begin{proof}
 Statement \eqref{corstat1} follows immediately from statement  \eqref{ATstatement} of Theorem~\ref{AdiThm}, cf.\ also Remark~1 below Theorem~\ref{AdiThm}. 
For the second statement observe that according to Theorem~\ref{AdiThm0}~(b) we have
 \begin{eqnarray*}
 \frac{1}{\epsi |\Lambda|} \tr  \left(\rho^\epslam(t)  J^\Lambda(t) \right) &=& \frac{1}{\epsi |\Lambda|} \tr  \left(U^\epslam(t,0) \rho_0^\Lambda U^\epslam(0,t) J^\Lambda(t) \right)\\
  &=& \frac{1}{\epsi |\Lambda|} \tr  \left( \rho_0^\Lambda P_*^\epslam(0) U^\epslam(0,t) J^\Lambda(t) U^\epslam(t,0)P_*^\epslam(0)\right)\\
  &=& \frac{1}{\epsi |\Lambda|} \tr  \left( \rho_0^\Lambda P_*^\epslam(0) J^\epslam_{\parallel\,(1)}(t) P_*^\epslam(0)\right) +\Or(\epsi \|\Phi_{J(t)}\|_{a,2,0})\,.
 \end{eqnarray*}
 The first summand in 
  \[
J^\epslam_{\parallel\,(1)}(t) =   U^\epslam_{\parallel\,(1)} (0,t)\, J^\epslam(t)  \,U^\epslam_{\parallel\,(1)}(t,0)+ \I\,\epsi\,  U^\epslam_{\parallel}(0,t ) \left( J^\epslam(t) R^\epslam_*(t)\dot P^\epslam_*(t)   - \dot P^\epslam_*(t)  R^\epslam_*(t) J^\epslam(t)\right)  U^\epslam_{\parallel}(t,0 )
 \]
 does not contribute because $U^\epslam_{\parallel\,(1)} (t,0)\rho_0^\Lambda U^\epslam_{\parallel\,(1)} (0,t)= P_*^\epslam(t)U^\epslam_{\parallel\,(1)} (t,0)\rho_0^\Lambda U^\epslam_{\parallel\,(1)} (0,t)P_*^\epslam(t)$  and we assume \eqref{persvanish}. 
 For the second summand we find by straightforward algebra that 
 \begin{eqnarray*}\lefteqn{
\frac{1}{\epsi |\Lambda|} \tr  \left( \rho_0^\Lambda P_*^\epslam(0) J^\epslam_{(1)\parallel}(t) P_*^\epslam(0)\right) }\\&=&
\frac{\I}{|\Lambda|} \tr  \left( \rho^\Lambda_\parallel(t)  P^\epslam_*(t) \left( J^\epslam(t) R^\epslam_*(t)\dot P^\epslam_*(t)   - \dot P^\epslam_*(t)  R^\epslam_*(t) J^\epslam(t) 
\right)
P^\epslam_*(t)
\right)\,,
\end{eqnarray*}
proving \eqref{corstat1bis}, since $\rho^\Lambda_\parallel(t) = P^\epslam_*(t)\rho^\Lambda_\parallel(t)  P^\epslam_*(t)$.
 To evaluate this expression further,  first observe that (omitting time-variables and superscripts for better readability)
 \begin{align*}
P_* J R &= P_*\,(\nabla_\alpha H)  R |_{\alpha=0}   = P_*\,(\nabla_\alpha (H -E_{ * } ) ) R |_{\alpha=0} \\
& = \nabla_\alpha
( P_* (H -E_{ * } )  R) |_{\alpha=0} - (\nabla_\alpha P_*)\, (H -E_{ * } )  R |_{\alpha=0} - P_*\, (H -E_{ * } )   \nabla_\alpha R |_{\alpha=0}\\
&=  - (\nabla_\alpha P_*)\,P_*^\perp |_{\alpha=0}\,.
\end{align*}
Hence,
\begin{eqnarray*}\lefteqn{
 P^\epslam_*(t) \left( J^\epslam(t) R^\epslam_*(t)\dot P^\epslam_*(t)   - \dot P^\epslam_*(t)  R^\epslam_*(t) J^\epslam(t) 
\right)
P^\epslam_*(t)=}\\
&=&  P^\epslam_*(t)\left(- (\nabla_\alpha P^\epslam_*(t) |_{\alpha=0}) P^{\epslam }_*(t)^\perp \dot P^\epslam_*(t) +\dot P^\epslam_*(t)P^{\epslam }_*(t)^\perp (\nabla_\alpha P_*^\epslam(t) |_{\alpha=0})
\right)
P^\epslam_*(t)\\
&=&  P^\epslam_*(t)  \left[  \dot P^\epslam_*(t) ,\nabla_\alpha P^\epslam_*(t) |_{\alpha=0} \right]  
P^\epslam_*(t)\,,
\end{eqnarray*}
proving also \eqref{corstat1bis2}.
\end{proof}

To obtain even more explicit formulas, let $(\varphi_n^\Lambda(t))_{n=0,\ldots, \dim (\Hi_{\Lambda,N})-1 }$ be an orthonormal basis 
 of  eigenvectors of $H^\Lambda(t)$,
\[
H^\Lambda(t) \, \varphi_n^\Lambda(t)  = E^\Lambda_n(t)\, \varphi^\Lambda_n(t) \,,
\]
such that $\operatorname{span} \{ \varphi^\Lambda_0(0) ,\ldots,  \varphi^\Lambda_{\dege-1}(0)\} = {\rm ran} P_*(0)$ and $\varphi_j^\Lambda(t) = U^\epslam_{\parallel}(t,0 ) \varphi_j^\Lambda({0})$ for $j=0,\ldots, \dege-1$.
Insering this into  \eqref{corstat1bis} and \eqref{corstat1bis2}
 we find by a straightforward computation%
\footnote{Notice that
$
\langle \varphi_n^\Lambda(t) , \dot \varphi_0^\Lambda(t) \rangle = 0$  for all $n=0,\ldots \kappa-1$.
} 
 two formulas for the leading order approximation to the macroscopic current density: dropping the dependence on time, this reads
\begin{align}\label{currentformula}
\mathcal{J}^\Lambda(t)    &=  -\frac{2}{|\Lambda|} {\rm Im}\left( \sum_{n\geq \dege} \frac{\langle\varphi_n^{\Lambda},\partial_t \varphi_0^{\Lambda}\rangle \,\langle \varphi_0^{\Lambda}, J^{\Lambda} \varphi_n^{\Lambda}\rangle}{E_n^{\Lambda}-E_0^{\Lambda}}
\right)\;+\;\Or(\epsi)\\
\label{currentformula2}
&= -\,  \frac{ 2}{|\Lambda|}\,{\rm Im} \left\langle  \partial_t  \varphi_0^\Lambda,  \nabla_\alpha \varphi_0^\Lambda  \big|_{\alpha=0}  \right\rangle
\;+\;\Or(\epsi)\,.
\end{align}
The right-hand side of \eqref{currentformula}, to be compared with \eqref{NiuThouless}, matches exactly the integrand of Formula  (2.13) in {\cite{NT}} (see also Formula (2.5) in \cite{ThPump}) for lattice systems: contrary to \cite{ThPump,NT}, however, in our case the formula holds even for a possibly degenerate ground state. 
Formula~\eqref{currentformula2} has the form of a curvature of the line bundle of ground states and was derived  e.g.\ in  \cite{AS}.  Let us stress once again that the error terms in both formulas above are bounded uniformly in the system size $|\Lambda|$. 

\subsection{Conductivity in quantum Hall systems}

Since the quantum Hall effect is the most prominent application of adiabatic currents,  let us briefly recall how \eqref{currentformula2} relates to the quantum Hall current. In a quantum Hall system an electromotive force in the form of a linear electric potential is applied across a two-dimensional sample and the Hall current is measured perpendicular to the electromotive force.  The general idea from \cite{NT,AS} is to implement the electromotive force in the case of a torus-geometry of the sample by a time-dependent ``gauge'' transformation. Let $H_0^\Lambda$ be the time-independent Hamiltonian of the unperturbed system and 
$H_0^\Lambda(\alpha_1,\alpha_2)$ the corresponding family of twisted Hamiltonians as in \eqref{Halpha}.
Then,  if the field is applied in the $2$-direction, the time-dependent Hamiltonian of the system is\footnote{We ignore the initial smooth switching   of the electric field, which could be modeled by putting $H^\Lambda(t) := H_0^\Lambda ( 0, f(\mathcal{E} t))$ for some smooth function $f:\R\to\R$ supported in $[0,\infty)$ with $f(s) =s$ for $s>s_0$.}
\[
H^\Lambda(t) := H_0^\Lambda ( 0, \mathcal{E} t)\,.
\]
Transforming to the   time variable $ \alpha_2 = \mathcal{E}t$ and assuming that the gap remains open for all $ \alpha_2 \in[0,2\pi)$, we obtain exactly an adiabatic problem to which Theorem~\ref{AdiThm} and Corollary~\ref{currentcor} apply  with $\epsi$ replaced by $\mathcal{E}$.  Note that now $t$ is the relevant time-variable, and the current density in \eqref{eqn:curr_dens}  does not have a prefactor $1/\mathcal{E}$.

According to \eqref{currentformula2}, the induced current density in the $1$-direction at time $t$ is, uniformly in the system size,
\begin{equation} \label{conductivity}
\begin{aligned}
\frac{1}{  |\Lambda|}   \tr  \left(\rho^{\mathcal{E},\Lambda}(t)  J^\Lambda_{1}(t) \right)   &= -\,  \frac{ 2}{|\Lambda|}\,{\rm Im} \left\langle  \partial_t  \varphi_0^\Lambda(0,\mathcal{E}t),  \partial_{\alpha_1} \varphi_0^\Lambda(0,\mathcal{E}t)  \right\rangle    +\Or(\mathcal{E}^2)\\ 
&=
\mathcal{E}\, \frac{ 2}{|\Lambda|}\,{\rm Im} \left\langle  \partial_{\alpha_1} \varphi_0^\Lambda(0,\mathcal{E}t) ,  \partial_{\alpha_2} \varphi_0^\Lambda(0,\mathcal{E}t)\right\rangle +\Or(\mathcal{E} ^2)\,.
\end{aligned}
\end{equation}
Hence, the Hall conductivity at finite system size and finite field $\mathcal{E}$, that is  the ratio between the current density and the applied field, is  
\begin{equation}\label{conductivityformula}
\sigma^{\mathcal{E},\Lambda}_{12}(t)  =  \frac{ 2}{|\Lambda|}\,{\rm Im} \left\langle  \partial_{\alpha_1} \varphi_0^\Lambda(0,\mathcal{E}t) ,  \partial_{\alpha_2} \varphi_0^\Lambda(0,\mathcal{E}t)\right\rangle +\Or(\mathcal{E} )\,.
\end{equation}
A quantity of physical interest would be the zero-field Hall conductivity of the infinite system, i.e.\ the limit
\[
\sigma_{1 2} := \lim_{\mathcal{E}\to 0} \lim_{M \to \infty} \sigma^{\mathcal{E},\Lambda}_{12}(t)\,. 
\]
This quantity is expected to be independent of $t$ and quantized, i.e.\  to take  values in $\frac{1}{2\pi}\frac{1}{\kappa}\Z$ in our units
\footnote{ Notice that $1/2\pi = e^2/h$ in units where $e=1$ and $\hbar=1$.}  
\cite{Th}. Recall that $\kappa$ is the degeneracy of the ground state, which we now assume to become constant for $M$ large enough.
The existence of this limit clearly depends on the details of the Hamiltonian $H_0$. However, our result shows that it suffices to analyze the leading order term in \eqref{conductivityformula}, since the error term is of order $\mathcal{E}$ uniformly in the system size.

For non-interacting systems and $\kappa=1$, quantization of $\sigma_{12}$ is well known    (e.g.\ \cite{BES,ASS}). Recently also integer quantization of Hall conductivity in  interacting  Haldane-type  models with small interaction was shown by Giuliani, Mastropietro, and Porta \cite{GMP}.
Although they do not take \eqref{conductivity} as a definition of conductivity, they also assume validity of a linear response approximation.
On the other hand, as they start from perturbing a gapped non-interacting system with a non-degenerate ground state, they do not need to assume a uniform gap for the interacting system. 

In general, however, a proof of quantization of Hall conductivity for interacting systems is still an open problem, even when starting from formula \eqref{conductivityformula}, which is now established rigorously by our result. Also an averaging procedure (c.f.\ \cite{AS,NT} and the next subsection for averaging in the case of Hall conductance) does not prove quantization in a simple way: Assume that the gap of $H_0(\alpha_1,\alpha_2)$ remains open for all $(\alpha_1,\alpha_2)\in [0,2\pi)^2$ and that $\kappa=1$.  Introduce
\[
\sigma^{\mathcal{E},\Lambda}_{12}(\alpha_1, t)  =  \frac{ 2}{|\Lambda|}\,{\rm Im} \left\langle  \partial_{\alpha_1} \varphi_0^\Lambda(\alpha_1,\mathcal{E}t) ,  \partial_{\alpha_2} \varphi_0^\Lambda(\alpha_1,\mathcal{E}t)\right\rangle\,.
\]
 Then the average of  $\sigma^{\mathcal{E},\Lambda}_{12}(\alpha_1, t)$  is
\begin{align*} 
\langle \sigma_{12}\rangle &:=
\frac{\mathcal{E}}{4\pi^2 }\int_0^{2\pi/\mathcal{E}}\D t\int_0^{2\pi}\D\alpha_1  \frac{ 2}{|\Lambda|}\,{\rm Im} \left\langle  \partial_{\alpha_1} \varphi_0^\Lambda(\alpha_1,\mathcal{E}t) ,  \partial_{\alpha_2} \varphi_0^\Lambda(\alpha_1,\mathcal{E}t)\right\rangle\\
&=  \frac{1}{|\Lambda|}  \int_{[0,2\pi)^2} \frac{\D^2\alpha}{4\pi^2 } \;2\,{\rm Im} \left\langle  \partial_{\alpha_1} \varphi_0^\Lambda(\alpha ) ,  \partial_{\alpha_2} \varphi_0^\Lambda(\alpha )\right\rangle \;\in\; \frac{1}{|\Lambda|} \cdot \,\tfrac{1}{2\pi}\Z\,,
\end{align*}
since the integral is the Chern number of  a line bundle over the torus. Without any further assumption, it is not obvious (and not even clear if it should be expected, see \cite{ThPump}) that this Chern number is a multiple of $|\Lambda|$. 

This statement can be proved by assuming that the system be translational invariant, see e.g.\ \cite{ASS}, by relating it to the Hall conductance (see the next subsection). 
A similar result can be obtained even in the presence of disorder, which breaks translation invariance pointwise, but under an homogeneity assumption that models a disordered crystalline system, in the sense of \cite{BES}. The latter approach allows also for the presence of a mobility gap rather than a spectral gap, but is however limited to non-interacting fermion systems.

\subsection{Conductance in quantum Hall systems}

The Hall conductance is somewhat  easier to handle. The latter is usually defined as the ratio  of the current $I$   through a fiducial line in the two-dimensional sample, say the line $\{x_1=0\}$, and a voltage drop $\Delta V$ across a fiducial line, say $\{x_2=0\}$, in the perpendicular direction. 
To model these quantities in our setting on the torus we follow essentially \cite{HM} and define  yet another 2-parameter family of Hamiltonians. As before, let $H_0$ be a uniformly locally-finite gapped Hamiltonian and define
\[
\mathfrak{N}_j    := \sum_{x\in \Lambda_j}  a_x^*a_x    \in \mathcal{A}_\Lambda^\mathfrak{N}\,,
\] 
that is, the number operator counting particles in the left, respectively lower, half  $\Lambda_j:=\{x\in\Lambda\,|\, x_j\leq 0\}$, $j=1,2$,   of the square $\Lambda$. Then the interaction of the Hamiltonian $H_0(\beta_1,\beta_2)$ is defined in two steps as
\[
\Phi^\Lambda_{H_0(\beta_1,0)}( X) := \begin{cases}
 \E^{-\I \beta_1 \mathfrak{N}_1 }\, \Phi_{H_0}^\Lambda ( X) \,\E^{\I \beta_1 \mathfrak{N}_1} &   \mbox{if } X\cap \Lambda_1\not=\emptyset\,,\; X\cap \Lambda\setminus \Lambda_1\not=\emptyset\,,  {\rm dist}(X,\{x_1=0\})\leq r\\
 \Phi_{H_0}^\Lambda ( X) & \mbox{otherwise, }
\end{cases}
\]
and then
\[
\Phi^\Lambda_{H_0(\beta_1,\beta_2)}( X) := \begin{cases}
 \E^{-\I \beta_2 \mathfrak{N}_2 }\, \Phi_{H_0(\beta_1,0)}^\Lambda ( X) \,\E^{\I \beta_2 \mathfrak{N}_2} &   \mbox{if } X\cap \Lambda_2\not=\emptyset\,,\; X\cap \Lambda\setminus \Lambda_2\not=\emptyset\,,{\rm dist}(X,\{x_2=0\})\leq r\\
\Phi^\Lambda_{H_0(\beta_1,0)}( X) & \mbox{otherwise. }
\end{cases}
\]
{In the above, $r$ is as in \eqref{eqn:r}.}
As in the case of conductivity, we consider the time-dependent Hamiltonian $H^\Lambda(t) := H^\Lambda_0(0,\Delta V t)$,
and the current-through-the-line operator is  
\[
I^\Lambda (t) := \partial_{\beta_1} H^\Lambda_0(0,\Delta V t)\,.
\]
Note that $\dot H^\Lambda(t)$ is now localized in a strip $-r< x_2<r$ and $I^\Lambda$ is localized in a strip $-r<x_1<r$.
Hence, assuming as before a gap for all $t\in [0,2\pi/\Delta V)$,   we can now apply Theorem~\ref{AdiThm} for $L$-localized driving and observable and   obtain   
\[
   \tr  \left(\rho^{\Delta V,\Lambda}(t)  I^\Lambda (t) \right)    \;=\; 
\Delta V\cdot 2 \,{\rm Im} \left\langle  \partial_{\beta_1} \varphi_0^\Lambda(0,\Delta V\,t) ,  \partial_{\beta_2} \varphi_0^\Lambda(0,\Delta V\,t)  \right\rangle+\Or(\Delta V ^2)\,.
\]
We thus proved that  the Hall conductance for the finite system at finite voltage $\Delta V$ is given by
\begin{equation}\label{conductance}
\tilde \sigma^{\Delta V,\Lambda}_{12}(t) \;=\; 2 \,{\rm Im} \left\langle  \partial_{\beta_1} \varphi_0^\Lambda(0,\Delta V\,t) ,  \partial_{\beta_2} \varphi_0^\Lambda(0,\Delta V\,t)  \right\rangle +\Or(\Delta V  )\,.
\end{equation}
As first observed in \cite{AS} and \cite{NT}, in this case 
  the averaging argument from the previous section readily shows quantization of the average Hall conductance, 
\[
\langle \tilde \sigma_{12}\rangle \in \tfrac{1}{2\pi}\,\Z\,.
\]
However, it follows from a recent result of Hastings and Michalakis   \cite{HM} (see \cite{BBDF} for a streamlined version of the proof under potentially stronger assumptions)\footnote{As remarked before, strictly speaking \cite{HM} and \cite{BBDF} apply to spin systems only. However, it is believed that these can be transferred to the present setting of interacting fermions with the appropriate modifications. See also \cite{KWKH} for some numerical indications in this sense.} that the leading term in \eqref{conductance} is indeed quantized up to terms that are almost-exponentially small in  the linear size $M$ of $\Lambda$. In particular,  there is a sequence  $k_M\in \Z$  and a function $f \colon 2\N \to \R$ with $\lim_{M\to \infty} M^n f(M) = 0$ for any $n\in\N$  such that
\begin{equation}\label{quantcond}
 \left| 2 \,{\rm Im} \left\langle  \partial_{\beta_1} \varphi_0^\Lambda(0,\Delta V\,t) ,  \partial_{\beta_2} \varphi_0^\Lambda(0,\Delta V\,t)  \right\rangle\;-\; \tfrac{k_M}{2\pi}\right| \leq f(M)\,.
\end{equation}
 In \cite{BBDF}, Theorem 1.4, it is also shown that if ground state expectations of all local observables have a thermodynamic limit,  then also $k_M$ converges and thus becomes constant for $M$ large enough.

In summary it thus follows that in such systems the infinite volume Hall conductance at zero field is quantized,
\begin{equation}\label{infinitecond}
\tilde \sigma_{12} := \lim_{\Delta V\to 0} \lim_{M\to\infty}\frac{  \tr  \left(\rho^{\Delta V,\Lambda}(t)  I^\Lambda (t) \right)  }{\Delta V} = \frac{k}{2\pi} \quad\mbox{ for some }\quad k\in\Z\,.
\end{equation}
As a more technical remark, we note that in \cite{HM} the authors actually prove quantization of the quantity
$2 \,{\rm Im} \left\langle  \partial_{\beta_1} \varphi_0^\Lambda(0) ,  \partial_{\beta_2} \varphi_0^\Lambda(0 )  \right\rangle$ in the sense of \eqref{quantcond} without assuming a gap for $\beta\not=0$. But it follows from their proof that when the gap persists for all $\beta=(0,\beta_2)$, then \eqref{quantcond} holds for all $t$ with the same $k_M$.  However, in order to derive the formula \eqref{conductance} for the conductance from microscopic first principles via the adiabatic theorem, we cannot dispose of the gap conditions for all times $t$. While the derivation of \eqref{infinitecond} through the combination of our adiabatic theorem and the results of \cite{HM} and \cite{H} constitutes the first rigorous proof of quantization of Hall conductance for interacting fermion systems starting from microscopic first principles, it is not yet fully satisfactory because of the gap assumption for all $t$, instead  of only for the fixed initial Hamiltonian $H_0$.
Although it is argued in \cite{BBDF} that in the specific example discussed in the present section the gap assumption for $H_0(0,0)$ implies 
a gap for $H_0(0,\beta_2)$ for all $\beta_2$, we expect that in general  
 one has to leave the realm of standard adiabatic theory and consider almost stationary states for systems where the driving closes the gap. Such states are constructed in \cite{T2}.
 
A different approach to derive the quantization of the Hall conductivity in interacting fermionic systems has been developed by Fr\"ohlich in the early nineties, see \cite{Fr} and references therein for a recent account. This approach is based on the coupling of matter (in the form of fermionic fields) to an electromagnetic gauge field $A$. An effective action for $A$ is derived by ``integrating out'' the fermionic degrees of freedom, and response coefficients like the Hall conductance can be computed from the derivative of this effective action with respect to $A$.

Another open problem is to show that 
    \eqref{conductivity} and \eqref{conductance} hold, at least in the thermodynamic limit,  with errors that are asymptotically smaller than any power of $\mathcal{E}$, resp.\ of $\Delta V$. For non-interacting systems this can be indeed shown (e.g.\ \cite{PST, ST}) and it is expected to hold for interacting systems as well. Indeed, in \cite{KS} the authors show under a gap assumption for all $\beta$ that the averaged Hall conductance satisfies \eqref{conductance}  with error terms of order $(\Delta V)^\infty$. However, their error estimates are not uniform in the size of the system and could  deteriorate in the thermodynamic limit. 

\section{The adiabatic expansion: Proof of Proposition~\ref{SuperLemma} } \label{PropProof}

\begin{proof}[Proof of Proposition~\ref{SuperLemma}.]
To simplify the notation and to improve readability, we often drop the dependence on the box $\Lambda$, on $\epsi$,  and on time $t$.
The strategy of the proof is to determine inductively the coefficients $A_\mu$ and $K_\mu$. 

We start by computing 
\begin{eqnarray*}\lefteqn{ 
\I \epsi   \frac{\D}{\D t}
 U_{\rm sa} (t,s)  \stackrel{\eqref{superadi}}{=} \I \epsi \frac{\D}{\D t}\Big(  V(t) U_{\rm a} (t,s)  V (s)^* \Big) }\\ &  \stackrel{\eqref{adiabaticevolution}}{=}&  \I\epsi  {  \dot V(t)} U_{\rm a}(t,s)  V (s)^* +     V(t)   {  H_{\rm a} (t)} U_{\rm a} (t,s) V (s)^*    \\[2mm] 
&  \stackrel{\eqref{HaDef}}{=}&   V(t)\Big(   \I\epsi    V(t)^* \dot V(t)  +    H (t)  +\epsi K(t)
\Big)V(t)^*U_{\rm sa} (t,s)     \\  & = &   \Big( H(t)  + \underbrace{V(t)\big(   \I\epsi    V(t)^* \dot V(t)  +     H (t)- V(t)^*H(t)V(t)   
  +\epsi K(t)
\big)V(t)^* }_{=:\tilde R(t)}    \Big) \,U_{\rm sa} (t,s) \,.
\end{eqnarray*}
We now choose the coefficients $A_\mu$ entering through \eqref{KSDef} in the definition \eqref{VDef} of $V $ and the coefficients~$K_\mu$ entering through \eqref{KSDef} in the definition \eqref{HaDef} of $H_{\rm a} $  in such a way that $U_{\rm a}$ satisfies \eqref{adievo2} and
  the remainder term $\tilde R$ satisfies 
\begin{equation}\label{Rexp}
\tilde R (t)   = V(t)\Big(  \I\epsi    V(t)^* \dot V(t)  +     H (t)- V(t)^*H(t)V(t)   
  +\epsi K(t)
\Big)  V(t)^* \stackrel{!}{=}\Or (\epsi^{n+1}) \,
\end{equation}
{where we set $n:=m+d$.}
To this end we expand $\tilde R(t)$ in powers of $\epsi$ and choose the coefficients $A_\mu$ and $K_\mu$ inductively.
 Expanding $V^*HV$ yields
\begin{align*}
V^*HV &=  \E^{-\I  \epsi  S  } H \E^{\I  \epsi  S  } =  \sum_{k=0}^n \frac{\epsi^k}{k!}\, \ad_{S }^k(H)
+  \frac{ \epsi^{n+1}}{(n+1)!} \, \E^{-\I \tilde\gamma S  }   \ad_{S }^{n+1} (H) \E^{\I  \tilde \gamma S  }\\
&=: \sum_{\mu=0}^n  \epsi^\mu H_\mu +\epsi^{n+1}h_n(\epsi)\,,
\end{align*}
where  $\tilde \gamma\in [0,\epsi]$ and each $H_\mu$ is defined as  the sum of those terms in the series that carry a factor $\epsi^\mu$.  Above  we denoted by $\ad_{S }^k(H) \equiv \ad_{S (t)}^{k} (H(t)) $ the nested commutator $[-\I S (t),[\cdots, [-\I S (t),[-\I S (t),H (t)]]\cdots]]$, where $-\I S (t)$ appears $k$ times. Including the factor $-\I$ into the definition of  $\ad_S := -\I\, \add_S$ will make computations in the following more transparent. While one could write down an explicit expression for $H_\mu$ (cf.\ \cite{BDF}), this is not necessary for the following. It is only important that 
\[
H_\mu =- \ad_{H} (A_\mu) + L_\mu\,,
\]
where $L_\mu$ contains a finite number of iterated commutators of the operators $A_\nu$, $\nu<\mu$, with~$H$.
Explicitly, the first orders are 
\[
H_0 = H\,, \qquad H_1 = -\ad_{H} (A_1)\,,\qquad H_2 =-\ad_{H} (A_2) - \tfrac{1}{2} [A_1,[A_1,H]]\,.
\]
In order to expand $V^*\dot V$, one uses Duhamel's formula
\[
\I \epsi V^*\dot V = -\epsi^2 \int_0^1 \E^{-\I \lambda \epsi S } \,\dot S   \,\E^{\I \lambda \epsi S }\,\D \lambda\,,
\]
expands the integrand as a series of nested commutators, and integrates term by term, to obtain
\begin{align*}
\I \epsi V^*\dot V &= -\epsi^2    \sum_{k=0}^{n-2} \frac{\epsi^k}{(k+1)!} \,\ad_{S }^k(\dot S )
- \frac{ \epsi^{n+1}}{(n-1)!} \int_0^1 \E^{-\I \lambda \tilde\gamma S  }   \ad_{S  }^{n-1} (\dot S ) \E^{\I \lambda \tilde \gamma S  }\,\D \lambda
\\
& = \sum_{\mu=1}^{n } \epsi^\mu Q_\mu + \epsi^{n+1} q_n(\epsi)\,,
\end{align*}
where again $Q_\mu$ collects all terms in the sum proportional to $\epsi^\mu$. Note that $Q_\mu$ is a finite sum of iterated commutators of the operators $A_\nu$ and $\dot A_\nu$ for $\nu<\mu$. One finds for the first terms
\[
Q_1 = 0\,,\qquad  Q_2 = -\dot A_1 \,,\qquad Q_3 = - \dot A_2 + \tfrac{\I}{2} [ A_1,\dot A_1]\,.
\]
Writing also $\epsi K = \sum_{\mu =1}^n \epsi^\mu K_\mu$, 
inserting the expansions into \eqref{Rexp} yields
\[
 \I\epsi   V ^* \dot  V  +H-     V ^* H  V + \epsi K    = \sum_{\mu=1}^n \epsi^\mu (Q_\mu- H_\mu+K_\mu) + \epsi^{n+1} ( q_n(\epsi)- h_{n}(\epsi) )\,,
\]
and it remains to determine $A_1,\ldots, A_n$ and $K_2,\ldots, K_n$ inductively such that
\[
0\stackrel{!}{=} (Q_\mu- H_\mu+K_\mu )  =  (Q_\mu+      \ad_{H} (A_\mu)- L_\mu  +K_\mu    ) 
\]
i.e.\ 
\begin{equation}\label{solve}
 \ad_{H} (A_\mu)   =  L_\mu   -Q_\mu  -K_\mu
\end{equation}
for all $\mu = 1,\ldots, n$. 

With   $L_1=Q_1=0$, for $\mu=1$ we thus must choose $A_1$ and $K_1$ such that 
\begin{equation}\label{solve1}
\ad_H(A_1)   = -     K_1\,.
\end{equation}
Recall that in standard adiabatic theory one   chooses $K_1=K_\parallel$, ensuring   \eqref{adievo2}. 
 Since  the map $B\mapsto \ad_H(B) =   -\I \,[H,B]$ defines an automorphism of the space 
\[
\mathcal{A}^{\rm OD}_\Lambda := \left\{ B \in \mathcal{A}^{\mathfrak{N}}_\Lambda\,|\, B= P_*BP_*^\perp + P_*^\perp B P_*\right\}
\]
of off-diagonal operators, cf.\ Appendix~\ref{AppendixI},  and since $K_\parallel\in \mathcal{A}^{\rm OD}_\Lambda$,  the equation $\ad_{H} (A_1)   =     -K_\parallel $
  has a unique off-diagonal solution $A_1$.
 However, since $K_\parallel$ is not a local Hamiltonian in general, the corresponding $A_1$ would  not be a local Hamiltonian as well and we cannot set  $K_1=K_\parallel$.
On the other hand, we need that $K_1^{\rm OD} =   K_\parallel$ in order to have the crucial intertwining property \eqref{adievo2} for the adiabatic evolution. The way out of this apparent dilemma  is to add a diagonal part $K_1^{\rm D}$ such that 
$K_1 = K_\parallel + K_1^{\rm D}$ is a local Hamiltonian.

This can be achieved by employing a linear map $\mathcal{I} ^\Lambda_H: \mathcal{A}^{\mathfrak{N}}_\Lambda\to \mathcal{A}^{\mathfrak{N}}_\Lambda$ constructed in the context of the so-called quasi-adiabatic or  spectral flow which has the following  properties, cf.\  Appendix~\ref{AppendixI}: 
\begin{enumerate}
\item[(I1)] $\mathcal{I}_H  $  maps local Hamiltonians to local Hamiltonians,
$\mathcal{I}_H : \mathcal{L}_{\mathcal{S}, k+1,L_H} \to \mathcal{L}_{\mathcal{S}, k,L_H }$.
\item[(I2)] $\mathcal{I}_H  $ commutes with $H$ and $P_*$,
\[
\mathcal{I}_H  (HA) = H \mathcal{I}_H  ( A) \,, \quad \mathcal{I}_H  ( AH) =   \mathcal{I}_H  ( A)H\,, \quad \mathcal{I}_H  (P_*A) = P_* \mathcal{I}_H  ( A) \,, \quad \mathcal{I}_H  ( AP_*) =   \mathcal{I}_H  ( A)P_*\,.
\]
\item[(I3)] The restriction of $\mathcal{I}_H  $  to $\mathcal{A}^{\rm OD}$ inverts the map $ \ad_H(\cdot)$, i.e.\ for $B\in  \mathcal{A}^{\rm OD}$ it holds that
\[
  \mathcal{I}_H ( \ad_H(B) )  = \ad_H(   \mathcal{I}_H (B) ) =  B\,.
\]
\end{enumerate}
A slight  modification of the standard definition of $\mathcal{I}_H $ {(see e.g.\ \cite{HW,BMNS})} explained in Appendix~\ref{AppendixI} allows for a fourth  property:
\begin{enumerate}
\item[(I4)] If the width $\delta$ of the spectral patch $\sigma_*$ is smaller than the gap $g$, then $\mathcal{I}_H$ can be constructed in such a way  that all operators in the range of $\mathcal{I}_H  $ have a vanishing $P_*(\cdots)P_*$ block, i.e.\ $P_* \,\mathcal{I}_H (A) P_* = 0$    for all  $A\in \mathcal{A}$.
 \end{enumerate}
\begin{lemma}\label{AKlem}
Let  $A_1 := - \mathcal{I}_H ( \mathcal{I}_H (\dot H ) )$ and $K_1  := -\ad_H(A_1)$. Then  
\[
  K^{\rm OD}_1 = K_\parallel\,.
\]
If $\delta<g$, then $P_* K_1P_*=P_* A_1P_* =0$. 
\end{lemma}
\begin{proof}
We have 
\[
 K^{\rm OD}_1  =-\ad_H(A_1^{\rm OD}) \stackrel{\rm (I2)}{=} 
 \ad_H( \mathcal{I}_H(\mathcal{I}_H(\dot H^{\rm OD})))
   \stackrel{\rm (I3)}{=}   \mathcal{I}_H(\dot H^{\rm OD})\,.
\]
Using that $\dot P_* \in \mathcal{A}^{\rm OD}$, we find that
\begin{eqnarray*}
 [ \mathcal{I}_H (\dot H^{\rm OD} ), P_*]  &\stackrel{\rm (I2)}{=} &\mathcal{I}_H ([ \dot H^{\rm OD} ,P_*])= \mathcal{I}_H ([ \dot H  ,P_*]) =  - \mathcal{I}_H ([   H ,\dot P_*]) \stackrel{\rm (I3)}{=}    \I \dot P_* =   \I [[\dot P_*,P_*],P_*] \\&= &[K_\parallel,P_*]\,,
\end{eqnarray*}
which implies that $
\mathcal{I}_H (\dot H )^{\rm OD}  = K_\parallel^{\rm OD}=K_\parallel
$.\end{proof}

By (I1), $A_1$ is a local Hamiltonian and by Lemma~\ref{manyadlemma} also $K_1$ is a local Hamiltonian. However, to make  the following induction work, we need to be a bit more explicit:
 According to (I1) there is a sequence $(\xi_{0,k})_{k\in\N_0}$ in $\mathcal{S}$ depending only on $H$ and its time derivatives $H^{(r)}$, $r=1,\ldots, n$,  through their $\|\cdot\|_{a,l}$-norms such that   $ \mathcal{I}_H (\dot H)^{(r)} \in \mathcal{L}_{\xi_{0,k},k,L_H}$ for $r=0,\ldots, n-1$ uniformly in time, i.e.\  with
\[
    \|\Phi_{  \mathcal{I}_H (\dot H) }^{(r)}\|_{\xi_{0,k},k,L_H,T}    <\infty
\]
for all $k\in\N_0$, $r=0,\ldots, n-1$, and $T\geq 0$. Applying (I1)  once more, we conclude that there is a sequence 
$(\xi_{1,k})_{k\in\N_0}$ in $\mathcal{S}$ such that $A_1^{(r)} \in \mathcal{L}_{\xi_{1,k},k,L_H}$ for all $k\in\N_0$ and $r=0,\ldots, n-1$ uniformly in time. Finally, by Lemma~\ref{manyadlemma}, also $K_1^{(r)} \in \mathcal{L}_{\xi_{1,k},k,L_H}$ for all $k\in\N_0$ and $r=0,\ldots, n-1$ uniformly in time.

We now proceed inductively. Assume that for $\mu>1$ we constructed $A_\nu$ and $K_\nu$ for all $\nu <\mu$. Thus     $Q_\mu$ and $L_\mu$ are determined and we need to solve \eqref{solve}.
Assuming that $  K_\mu^{\rm OD}=0$,   the off-diagonal part of \eqref{solve} is solved by setting
\[
A_\mu = \mathcal{I}_H  ( L_\mu-Q_\mu   )\,.
\]
Then we pick $K_\mu$ to make the diagonal part of the right-hand side vanish as well:
\[
K_\mu = L_\mu- Q_\mu -\ad_H( A_\mu)    = H_\mu - Q_\mu  \,.
\]
Note that $K_\mu$ is indeed diagonal by (I3) for $\mu>1$.

Assuming   that $A_\nu^{(r)} \in \mathcal{L}_{\xi_{\mu,k},k,L_H}$ for all $k\in\N_0$ and $r=0,\ldots, n-\nu$ uniformly in time, we find by Lemma~\ref{manyadlemma}  that 
 $L_\mu^{(r)}$ and $Q_\mu^{(r)}$ are all in $\mathcal{L}_{\xi_{\mu-1,k},k,L_H}$ uniformly in time for $r=0,\ldots,n-\mu $ and $k\in\N_0$. Thus by Lemma~\ref{Ilemma2} there is a sequence 
$(\xi_{\mu,k})_{k\in\N_0}$ in $\mathcal{S}$ such that $A_\mu^{(r)} \in \mathcal{L}_{\xi_{\mu,k},k,L_H}$ for all $k\in\N_0$ and $r=0,\ldots, n-\mu$ uniformly in time.
In summary, using also Lemma~\ref{Slemma}, we conclude that $S^{\epsi}= \sum_{\mu=1}^n \epsi^{\mu-1} A_\mu$
has an interaction $\Phi_{S^{\epsi}}$ such that for some sequence $\xi_k$ in $\mathcal{S}$ it holds that $\|\Phi_{S^{\epsi}}\|_{\xi_k,k,L_H,T}<\infty$ for all $k\in\N_0$ and $T\geq 0$.

 To see that the remainder term
 $ R   = V \left( q_n - h_n \right) V^*$ is in $\mathcal{L}_{\mathcal{S},\infty,L_H}$ uniformly in time, first note that $q_n$ and $h_n$ each contain a number of terms that are just multi-commutators and can be estimated by Lemma~\ref{manyadlemma}, as well as a remainder term from the Taylor expansion, that can be estimated by combining 
Lemma~\ref{manyadlemma} and Lemma~\ref{transformlemma}. Finally the conjugation with $V$ that leads to $R $ is again 
 estimated by  Lemma~\ref{transformlemma}. Thus we proved \eqref{supadieq} and are left to check the additional claims (a)--(c).

For (a) note that 
if at some time $t'\in \R$ it holds that $\frac{\D^n}{\D t^n} H (t')=0$ for all $n=1,\ldots, m+d$, then, by the above induction,   also
$S^\epslam(t') = K^\epslam(t') =0$.  

Claim (b) was shown in Lemma~\ref{AKlem}. 

For (c) we first recall that for $\sigma_*^\epslam(t) = \{E_*^\epslam(t)\}$ we have $P_*A_1P_*=0$ 
and thus that $\tilde A_1 := A_1 - P_*^\perp A_1P_*^\perp = \tilde A_1^{\rm OD}$ is off-diagonal. Hence
$\tilde A_1$ is the unique off-diagonal solution of $ -\ad_H(\tilde A_1)   =     K_1^{\rm OD} = K_\parallel$, which, according to Appendix~\ref{AppendixI}, is given by
\[
\tilde A_1 = \I [ K_\parallel, R_*]  = -[ [\dot P_*,P_*],R_*] = P_*\dot P_* R_* + R_* \dot P_* P_* 
\]
with $R_* := (H-E_*)^{-1} P_*^\perp$.
For $\tilde K_2$ we finally obtain
\begin{eqnarray*}
\tilde K_2= P_* K_2 P_* &=&  P_*( \dot A_1 + \I\,[H,A_2] - \tfrac{1}{2} [A_1,[A_1,H]]   )P_*= - \tfrac{1}{2}  P_*    [A_1,[A_1,H-E_*]]    P_*
\\&=& - \tfrac{1}{2}  P_*    [\tilde A_1,[\tilde A_1,H-E_*]]    P_* = P_*\dot P_* R_* \dot P_* P_*\,.
\end{eqnarray*}
Here we used that $P_*\dot A_1 P_* = -P_*\mathcal{I}_H (\mathcal{I}_H (\ddot H)) P_* =0$.
This concludes the proof of Proposition~\ref{SuperLemma}.
\end{proof}

\section{Proof of the adiabatic theorem} \label{sec:proofAdiabatic}

 We start with the proof of  the superadiabatic theorem, Theorem~\ref{AdiThm}. 

\begin{proof}[Proof of Theorem~\ref{AdiThm}.]
We freely use the notation from Proposition~\ref{SuperLemma} and its proof provided in the last section.
Also note that for  any self-adjoint operator $B \in \mathcal{A}^\mathfrak{N}_\Lambda$ it holds that
$\|B\| = \sup_P |\tr (PB)|$, where the supremum is taken over all  rank-1 orthogonal projections $P$.
Let thus $P$ be any such projection  and consider first the evolution of a local observable $O\in \mathcal{A}_X^\mathfrak{N}$, $X\subset \Lambda$. A simple Duhamel argument gives
\begin{eqnarray}\label{proof1} 
D(O)& :=& \tr  \left( P  \left(     U^\epslam_{\rm sa}(s,t) \,O \, U^\epslam_{\rm sa}(t,s) \right) - U^\epslam(s,t) \, O  \, U^\epslam(t,s)  \right)\nonumber\\
&=& \tr \left( \left(   U^\epslam_{\rm sa}(t,s)P   U^\epslam_{\rm sa}(s,t) - U^\epslam(t,s)P U^\epslam(s,t) \right)\,O\right)\nonumber\\
&=&    \int_s^t\D \tau \, \tr \Big(  \frac{\D}{\D \tau} \Big( U^\epslam(s,\tau)  U^\epslam_{\rm sa}(\tau,s)P   U^\epslam_{\rm sa}(s,\tau)U^\epslam(\tau,s)\Big)
 U^\epslam(s,t)\,O\, U^\epslam(t,s)
\Big)\nonumber\\
&=& 
\I\epsi^n  \int_s^t\D \tau \, \tr \left(   \left[  R^\epslam(\tau),  U^\epslam_{\rm sa}(\tau,s) P   U^\epslam_{\rm sa}(s,\tau)\right]   U^\epslam(\tau,t)\,O\, U^\epslam(t,\tau)
\right)\nonumber\\
&=& 
- \I \,\epsi^n  \int_s^t\D \tau \, \tr \left(   U^\epslam_{\rm sa}(\tau,s) P   U^\epslam_{\rm sa}(s,\tau) \left[ R^\epslam(\tau), U^\epslam(\tau,t)\,O\, U^\epslam(t,\tau)\right] 
\right)\,.
\end{eqnarray}
Since $P$ and thus also $U^\epslam_{\rm sa}(\tau,s) P   U^\epslam_{\rm sa}(s,\tau)$ has trace one, we have that 
\begin{equation}\label{localEst}
\begin{aligned}
\left|D(O)\right| 
&\leq |t-s| \,\epsi^n \sup_{\tau\in[0,T]} \left\|\left[ R^\epslam(\tau), U^\epslam(\tau,t)\,O\, U^\epslam(t,\tau)\right] \right\|\\
&\leq C \, \epsi^n \, \|\Phi_R\|_{\zeta_0,0,L_H,T} \,\|O\|\, |X|^2  \,\zeta({\rm dist}(X,L_H)) \,|t-s|  (1+\epsi^{-d}|t-s| ^{d})\\
&\leq C\,\epsi^m\,  \|\Phi_R\|_{\zeta_0,0,L_H,T}\,\|O\| \,|X|^2  \,\zeta({\rm dist}(X,L_H))\,|t-s|  (1+|t-s| ^{d})
\,,
\end{aligned}
\end{equation}
where the second inequality follows from Lemma~\ref{CommuLemma} (note that due to the adiabatic time scale we pick up the factor $\epsi^{-d}$) and we set ${\rm dist}(X,L_H) := \min_{x \in X} {\rm dist}(x,L_H)$ (compare \eqref{distxL}).
Recall  that $B^\epslam = \sum_{X\subset \Lambda} \Phi_B^\epslam(X)$. Hence,  substituting $\Phi_B^\epslam (X)$ for~$O$ in \eqref{localEst}, we obtain 
\begin{eqnarray}  \label{comparerhoPi}\lefteqn{\hspace{-20pt}
\Big|    \tr  \left( P  \left(     U^\epslam_{\rm sa}(s,t) \,\Phi_B^\epslam (X) \, U^\epslam_{\rm sa}(t,s)  - U^\epslam(s,t) \, \Phi_B^\epslam (X)  \, U^\epslam(t,s) \right) \right)\Big| \leq }\nonumber\\
& \leq& 
 \epsi^m C \sum_{X \subset \Lambda}   |X|^2\,\zeta({\rm dist}(X,L_H)) \, \|\Phi_B^\Lambda(X)\| 
 \nonumber\\
&\leq &
  \epsi^mC   \sum_{x\in \Lambda} \sum_{X \subset \Lambda:\:x\in X}   |X|^2\,\zeta({\rm dist}(x,L_H)) \, \|\Phi_B^\Lambda(X)\|\nonumber\\
&\leq &
 \epsi^mC  \sum_{x\in\Lambda}   \zeta({\rm dist}(x,L_H)) \,  \sum_{y\in \Lambda} F_\zeta(d_L^\Lambda(x,y))\sum_{X \subset \Lambda:\:x,y\in X}  |X|^2 \, \frac{\|\Phi_B^\Lambda(X)\|}{F_\zeta(d_L^\Lambda(x,y))}\nonumber\\
&\leq&  \epsi^mC  \|\Phi_B\|_{\zeta,2,L} \,\|F\|_\Gamma \sum_{x\in\Lambda}   \zeta({\rm dist}(x,L_H))  \zeta({\rm dist}(x,L)) 
\nonumber\\
&\leq &\epsi^mC  \|\Phi_B\|_{\zeta,2,L} \,\|F\|_\Gamma M^{d-|\ell|-|\ell_H|}
\,.
\end{eqnarray}
In the third inequality we used that summing over all sets $X$ for which $x$ minimizes the distance to $L_H$ and then over all $x\in\Lambda$ would also include each term in the sum on the previous line at least once.
In the second-to-last inequality we used Lemma~\ref{distLemma}.
\end{proof}

Note that a key step in the previous proof was the application of Lemma~\ref{CommuLemma} to control the evolution of local observables on the {\em long} adiabatic time scale. Since this lemma is only available  for evolutions generated by exponentially localised Hamiltonians, the problem of approximating the superadiabatic time-evolution by dropping higher order terms in the generator $H_{\rm a}$ is non-trivial and will only have a satisfactory solution for the case $\delta=0$.

\begin{proof}[Proof of Theorem~\ref{expThm}.] 
 We start with (a). By evaluating the Taylor formula for the analytic function $\gamma \mapsto \E^{-\I \gamma S^\epslam }  B^\epslam \E^{\I  \gamma S^\epslam }$ at $\gamma=\epsi$, we find
\begin{eqnarray}\label{PiExpand}
 \nonumber
        \E^{-\I \epsi S^\epslam }  B^{\Lambda} \E^{\I  \epsi S^\epslam }  
&=  &  \sum_{j=0}^k \frac{ \epsi^j}{j!} \,  \ad_{S^\epslam }^j (B^\epslam) 
+ \frac{ \epsi^{k+1}}{(k+1)!}\,  \E^{-\I \tilde\gamma S^\epslam }   \ad_{S^\epslam }^{k+1} (B^\epslam) \E^{\I  \tilde \gamma S^\epslam } 
 \nonumber\\
&=:&   \sum_{j=0}^k \frac{\epsi^j}{j!} \,  \ad_{S^\epslam_{(k+j-1)} }^j (B^\epslam) 
  + R_1
\end{eqnarray}
for some $\tilde\gamma\in [0,\epsi]$. 
The norm of the remainder term $R_1$ can now be estimated using Lemma~\ref{lemma:comm1} and the argument that took us from \eqref{localEst} to \eqref{comparerhoPi}. Note, however,  that one uses in this argument  the existence of a function $\tilde \zeta\in\mathcal{S}$ with $\mathcal{B}_{\zeta,k+1,L_H} \subset \mathcal{B}_{\tilde \zeta,k+1,L_H}$ such that $S^\epslam  \in \mathcal{L}_{\tilde \zeta, k+1,L_H}$ (compare Lemma~\ref{Slemma} (a)). 

For (b) first note that $U^\epslam_{\parallel\,(k)}(t,s) $ generated by $K^\epslam_{(k)}(t)$   agrees with $\widehat U^\epslam_{\parallel\,(k)}(t,s) $ generated by 
$
\widehat K^\epslam_{(k)}(t)  := \sum_{\mu=1}^{k+1} \epsi^{\mu-1} K^\epslam_\mu(t)
$
on the range on $P_*^\epslam(s)$ since both evolutions have the intertwining property \eqref{adievo2} and the difference of the generators  
\[
K^\epslam_{(k)}(t) - \widehat K^\epslam_{(k)}(t) = K^\epslam_\parallel - K^\epslam_1  -  \sum_{\mu=2}^{k+1} \epsi^{\mu-1} P^\epslam_*(t)^\perp K^\epslam_\mu(t)P^\epslam_*(t)^\perp
\]
 is non-zero only in its $P_*^\perp(\cdots)P_*^\perp$-block.

Now we follow in principle the same strategy as in the proof of Theorem~\ref{AdiThm} to compare the time-evolutions
$U^\epslam_{\rm a} $  and  
$\widehat U^\epslam_{\parallel\,(k)} $. Hence we only point out the differences. Replacing $U^\epslam_{\rm sa} $ by $U^\epslam_{\rm a} $ and $U^\epslam $ by $\widehat U^\epslam_{\parallel\,(k)} $ in  \eqref{proof1}
and restricting to projections $P$ with $P= P P_*(s) $, 
the difference of the generators $R^\epslam(\tau)$ is replaced by 
\[
T^\epslam(\tau) := H_{\rm a}^\epslam(\tau) - \widehat K^\epslam_{(k)}(\tau) = H^\epslam(\tau)  +  \sum_{\mu=k+2}^{m+d} \epsi^{\mu-1} K^\epslam_\mu(\tau)
\]
and we need to control  the norm of 
\begin{eqnarray}\lefteqn{
\epsi^{-1} P_*^\epslam(\tau) \left[  T^\epslam(\tau) , U^\epslam_{\parallel\,(k)} (\tau,t)\,O\, U^\epslam_{\parallel\,(k)} (t,\tau)\right] P_*^\epslam(\tau) }\nonumber\\
&=&
\epsi^{-1} P_*^\epslam(\tau) \left[ H^\epslam(\tau) - \inf(\sigma^\epslam_*(t))-\delta/2  , U^\epslam_{\parallel\,(k)} (\tau,t)\,O\, U^\epslam_{\parallel\,(k)} (t,\tau)\right] P_*^\epslam(\tau) \label{term1} \\
 &&+\; 
 \epsi^{-1}P_*^\epslam(\tau) \left[  \sum_{\mu=k+2}^{m+d} \epsi^{\mu-1} K^\epslam_\mu(\tau) , U^\epslam_{\parallel\,(k)} (\tau,t)\,O\, U^\epslam_{\parallel\,(k)} (t,\tau)\right] P_*^\epslam(\tau) \,.\label{term3}
 \end{eqnarray}
In \eqref{term1} we subtracted the number $\inf(\sigma^\epslam_*(\tau))+\delta/2$, which has vanishing commutator with any operator.  Since $\|(H^\epslam(\tau) - \inf(\sigma^\epslam_*(\tau)) )P_*^\epslam(\tau)\|=\delta$ uniformly in $\Lambda$, we have 
\[
\|\eqref{term1}\|\leq  \frac{\delta}{\epsi} \|O\|\,.
\]
For \eqref{term3} we proceed as in the proof of Theorem~\eqref{AdiThm} with one difference:
  This time we cannot apply Lemma~\ref{CommuLemma} as before, since the Hamiltonian $\tilde K^\epslam_{(k)}(\tau)$ is not in $\mathcal{L}_{\mathcal{E},0}$ but only in $\mathcal{L}_{\mathcal{S},\infty}$. However, since \eqref{Uparallel} has no adiabatic time scaling, we don't need to control the growth of the error
in time and we can use the second estimate of Lemma~\ref{CommuLemma}. \end{proof}

\appendix
\section*{Appendices}
In the following appendices we collect the various technical details that are at the basis of the adiabatic theorem and the underlying formalism. Throughout these appendices we will make use of the notation established in Section~\ref{sec:framework}, but for the sake of readability we will often drop the superscript $\Lambda$  when no confusion arises.

\section{Lemma on functions in $\mathcal{S}$}\label{Sapp}

In this appendix we prove the following lemma on functions in $\mathcal{S}$.

\begin{lemma}\begin{itemize}
\item[{\rm (a)}]
For $\zeta,\xi \in\mathcal{S}$ it holds that  either $\mathcal{B}_{\zeta,n,L} \subset \mathcal{B}_{\xi,n,L}$ or
$\mathcal{B}_{\xi,n,L}\subset \mathcal{B}_{\zeta,n,L} $ (or both) for all $n\in\N_0$ and $L \in \loc$. Hence the spaces $\mathcal{B}_{\mathcal{S},n,L}$, and thus also $\mathcal{B}_{\mathcal{S},\infty,L}$, are indeed vector spaces.
\item[{\rm (b)}] Let $f:[0,\infty)\to (0,\infty)$ be a function with 
$ \sup\{ r^nf(r)\,|\, r\in [0,\infty)\} <\infty$   for all  $ n\in\N_0$. Then there exists a function $\zeta\in \mathcal{S}$ and $c>0$ such that $cf\leq \zeta$.
\end{itemize}\label{Slemma}
\end{lemma}
\begin{proof}
First note that for $\xi,\zeta\in\mathcal{S}$ with $\xi\leq \zeta$ it holds that $\|\Phi\|_{\zeta,n,L} \leq \|\Phi\|_{\xi,n,L}$ and hence
$\mathcal{B}_{\xi,n,L} \subset \mathcal{B}_{\zeta,n,L}$. We will now show that for any pair of functions $\xi,\zeta\in\mathcal{S}$
it holds that either $\mathcal{B}_{\xi,n,L} \subset \mathcal{B}_{\zeta,n,L}$ or $\mathcal{B}_{\zeta,n,L} \subset \mathcal{B}_{\xi,n,L}$ or $\mathcal{B}_{\xi,n,L} = \mathcal{B}_{\zeta,n,L}$ for all $n\in\N_0$ and $L \in \loc$.

To this end, let $\tilde \xi := \ln(\xi)$
  for $\xi\in \mathcal{S}$, i.e.\ $\xi=\E^{\tilde \xi}$. Then $\tilde \xi$ is  non-increasing and   super-additive,  that is, 
\[
\tilde \xi(x+y) \geq \tilde \xi(x) + \tilde \xi(y) \quad  \forall x,y \in [0,\infty) \,.
\]
For the following considerations we can restrict functions in $\mathcal{S}$ to $\N_0\subset[0,\infty)$, as $d^\Lambda$ and ${\rm dist}(\cdot,L)$ take values only in $\N_0$ and thus the norms $\|\cdot\|_{\xi,n,L}$ depend only on the values of  $\xi$ on~$\N_0$.

Fekete's super-additivity lemma  \cite{F} 
 says that for any super-additive function $f:\N_{0}\to \R$ the limit $\lim_{x\to \infty}  f(x)/x$ exists and equals
\[
c_f := \sup_{x\in \N} \frac{  f(x)}{x}\, .
\]
In general, the limit could be $+\infty$. However, for $\xi\in\mathcal{S}$  we know that $c_{ \tilde\xi }\leq 0$ since $\lim_{x\to\infty}\xi(x) = 0$ and thus $\lim_{x\to\infty}\tilde \xi(x) = -\infty$.

Now assume that we have two functions $\xi,\zeta\in \mathcal{S}$. 
If $c_{ \tilde\xi }<c_{ \tilde\zeta }$, then there exists $x_0 \in\N $ such that
\[
\tilde \xi(x)\leq \tilde \zeta(x)\quad\forall x \ge x_0\,. 
\]
 If $\tilde \xi(x)\leq \tilde \zeta(x)$ also for $x< x_0$, then correspondingly $\xi =\E^{\tilde\xi}\le\E^{\tilde\zeta}= \zeta$, and we argued at the beginning of the proof that $\mathcal{B}_{\xi,n,L} \subset \mathcal{B}_{\zeta,n,L}$ for all $n \in \N_0$ and $L \in \loc$.  Assume  on the contrary  that 
\[
\tilde a:=  \min_{0\leq x\leq x_0}  \left\{ \tilde \zeta(x)- \tilde \xi(x) \right\} <0\,.
\]
 Then 
\[
\tilde \xi_a(x) :=  \tilde \xi(x) +\tilde a  \leq \tilde \zeta(x) \quad\forall x\in [0,\infty) \,.
\]
Since $\tilde a<0$, also $\tilde \xi_a$ is super-additive and thus 
\[
\xi_a := \E^{\tilde a} \xi= \E^{\tilde \xi_a}    \leq \E^{\tilde \zeta} = \zeta
\]
is in $\mathcal{S}$. Notice now  that for $\xi\in\mathcal{S}$ and $a \in (0,1] $  then  $a\xi\in\mathcal{S}$ and, as sets, 
$\mathcal{B}_{\xi,n,L} = \mathcal{B}_{a\xi,n,L}$  for all $n\in\N_0$ and $L \in \loc$ (since trivially $\|\Phi\|_{a\xi,n,L} = a^{-1} \|\Phi\|_{\xi,n,L}$).
Hence, if $\Phi \in \mathcal{B}_{\xi,n,L}=\mathcal{B}_{\xi_a,n,L}$, then
 $\Phi\in\mathcal{B}_{\zeta,n,L}$ as well.

In the case that $c_{ \tilde\xi} =c_{ \tilde\zeta} $ but $\xi\not=\zeta$ we have that $\lim_{x\to\infty} (\tilde \zeta(x)- \tilde \xi(x))  /x =0$, and thus either $\tilde a :=  \inf_{x\in\N_0}  \left\{ \tilde \zeta(x)- \tilde \xi(x) \right\} <0$ but finite or $\tilde b:=  \inf_{x\in\N_0}  \left\{\tilde \xi(x)- \tilde \zeta(x)\right\} <0$ but finite
(or both). Assume without loss of generality that $\tilde a<0$ (otherwise revert the roles of $\zeta$ and $\xi$). Then by the same argument given before we find that $\E^{\tilde a} \xi\leq \zeta$ and can conclude analogously that $\mathcal{B}_{\xi,n,L}\subset \mathcal{B}_{\zeta,n,L}$.

In summary we found that for any $\xi,\zeta\in\mathcal{S}$ either $\mathcal{B}_{\xi,n,L} \subset \mathcal{B}_{\zeta,n,L}$ or $\mathcal{B}_{\zeta,n,L} \subset \mathcal{B}_{\xi,n,L}$ or $\mathcal{B}_{\xi,n,L} = \mathcal{B}_{\zeta,n,L}$  for all $n \in \N_0$ and $L \in \loc$, and thus we proved part (a).

For part (b),  set $c^{-1}:=\sup_{x \in [0,\infty)} f(x)$ and $\tilde f(x) := \ln(c f(x))$.  Note that the assumptions on $f$ imply  that 
\[
\lim_{x\to \infty} (\tilde f(x) + k\ln(x+1)) = -\infty \qquad\mbox{for all $k\in\N_0$}\,.
\]
Hence there exists a strictly   increasing  sequence $(x_k)_{k\in \N_0}$ such that $\tilde f(x) \leq - k \ln(x+1)$ for all $x\geq x_k$.  Notice that, since $c f(x) \le 1$ we have $\tilde f(x) \le 0$, so that we can take $x_0 = 0$.  Then 
\[
\tilde g (x) = -\;\sum_{k=0}^\infty {\bf 1}_{x_k \leq x < x_{k+1}}(x) \, k \,\ln(x+1)
\]
defines a super-additive function, since each function $-k\ln(x+1)$ is convex and thus super-additive and $-k\ln(x+1) < -\tilde k\ln(x+1)$
for $k>\tilde k$.
Using that
$\lim_{x\to \infty} (\tilde g(x) + k\ln(x+1)) = -\infty$ for all $k\in\N_0$ and $\tilde g(x)\geq \tilde f(x)$ for all $x\geq x_0 =0 $,
we find that
\[
 cf = \E^{\tilde f}  \leq \E^{\tilde g} =: g
\]
with $g\in\mathcal{S}$.
\end{proof}

\section{Lieb--Robinson bound on the torus}\label{AppendixLR}

One key technical ingredient in all of the following constructions is the so-called Lieb--Robinson bound \cite{LR} for the speed of propagation of local changes in interacting systems on lattices.
We will state a recent version of the Lieb--Robinson bound for fermionic systems by Nachtergaele, Sims, and Young \cite{NSY} in Theorem~\ref{LRB}, but adapted to our present setting of a torus. Given Lemma~\ref{LRlemma} below, the proof of Theorem~\ref{LRB} works line by line as the proof in \cite{NSY}.

First we need to introduce some more notation.
It is well known (see e.g.\ \cite{NSY}) and straightforward to check that the functions $F$ and $F_\zeta$ have the following crucial properties.
\[
\|F \|_\Gamma:=  \sup_{x\in\Gamma}\sum_{y\in\Gamma} F  (d (x,y))     <\infty
\]
and 
\[
   \sup_{x,y\in\Gamma} \sum_{z\in \Gamma}\frac{F_\zeta (d (x,z)) F_\zeta (d (z,y))}{F_\zeta (d (x,y)))} <\infty\,.
\]
However, we will mainly need the following local versions on the ``torus'' $\Lambda$.
\begin{lemma}\label{LRlemma}
It holds for all $\Lambda$, $\zeta\in\mathcal{S}$, and $L\in\loc$, that  
\[
  \sup_{x\in\Lambda}\sum_{y\in\Lambda} F  (d^\Lambda(x,y)) \leq \|F  \|_\Gamma 
\]
and
\[
\sup_{x,y\in\Lambda}\sum_{z\in\Lambda}\frac{F_\zeta(d_L^\Lambda(x,z)) F_\zeta(d_L^\Lambda(z,y))}{F_\zeta(d_L^\Lambda(x,y))}  \leq 2^{d+1} \|F\|_\Gamma
\]
and
\[
\sup_{x,y\in\Lambda}\sum_{z\in\Lambda}\frac{F_\zeta(d_L^\Lambda(x,z)) F_\zeta(d^\Lambda(z,y))}{F_\zeta(d_L^\Lambda(x,y))}  \leq 2^{2d+2} \|F\|_\Gamma\,.
\]

\end{lemma}
\begin{proof} By translation invariance of $d^\Lambda$ we have  
\[
 \sup_{x\in\Lambda}\sum_{y\in\Lambda} F  (d^\Lambda(x,y))  = \sum_{y\in\Lambda} F  (d^\Lambda (0,y)) \leq \sum_{y\in\Gamma} F  (d (0,y)) = \|F  \|_\Gamma \,.
\]
For the other estimate, first recall that $\zeta\in\mathcal{S}$ satisfies $ \zeta(r+s) \geq \zeta(r)\zeta(s)$ and is monotonically decreasing. Using this and 
 the triangle inequality for $d^\Lambda_L$, one easily sees that it
suffices to show the second estimate for $\zeta\equiv1$. 

Consider any  function $\delta:\Lambda\times\Lambda\to [0,\infty)$ satisfying the triangle inequality $\delta(x,y)\leq \delta(x,z)+\delta(z,y)$ and $\delta(x,y)\geq d^\Lambda(x,y)$ for all $x,y,z\in\Lambda$.
Then, using that $F$ is decreasing, we find for $x,y,z\in\Lambda$ that
\begin{align*}
\frac{F(\delta(x,z)) F(\delta(z,y))}{F(\delta(x,y))}&\leq \frac{F(\delta(x,z)) F(\delta(z,y))}{F(\delta(x,z)+\delta( z,y ) )}
= \frac{ (1+ \delta(x,z) + \delta(z,y))^{d+1}}{(1+ \delta(x,z)  )^{d+1}(1  + \delta(z,y))^{d+1}}\\
&\leq
2^d  \frac{ (1+ \delta(x,z)  )^{d+1} + \delta(z,y)^{d+1}}{(1+ \delta(x,z)  )^{d+1}(1  + \delta(z,y))^{d+1}}\\
&\leq
2^d\left(\frac{1}{(1+ \delta(x,z))^{d+1}} + \frac{1}{(1+ \delta(z,y))^{d+1}}
\right)\\
&\leq
2^d\left(\frac{1}{(1+ d^\Lambda(x,z))^{d+1}} + \frac{1}{(1+ d^\Lambda(z,y))^{d+1}}
\right)\\
&=
2^d\big(F(d^\Lambda(x,z)) + F(d^\Lambda(z,y)) \big)\,.
\end{align*}
Together with the first estimate, the second one follows.  The third inequality follows along the same lines using in the first step that
\[
\delta(x,y) = d^\Lambda_L(x,y) \leq d(x,z) + 2 d^\Lambda(z,y) + {\rm dist}(x,L) +{\rm dist}(z,L) =  d^\Lambda_L(x,z) + 2 d^\Lambda(z,y)\,. \qedhere
\]
\end{proof}

Two more definitions are required for the formulation of the Lieb--Robinson bound.
For $X\subset \Lambda\subset \Gamma$, the set of \emph{boundary sets} of $X$ in $\Lambda$ is
\[
S_\Lambda(X) := \{Z\subset \Lambda\,|\,Z\cap X\not=\emptyset\mbox{ and } Z\cap(\Lambda\setminus X)\not=\emptyset\}\,.
\]
For a (possibly time-dependent) interaction $\Phi$, the \emph{$\Phi$-boundary} of a set $X\in\mathcal{F}(\Gamma)$ is defined as 
\[
\partial_\Phi X = \{ x\in X\,|\, \exists Z \in S_\Gamma(X)\,,\;t\in [0,\infty) \mbox{ with } x\in Z \mbox{ and } \Phi(t,Z)\not=0\}\,.
\]

\begin{theorem}[Lieb--Robinson bound]\label{LRB}
Let $H\in \mathcal{L}_{\zeta,0}$ with interaction $\Phi$ depending continuously on $t\in [0,\infty)$. For $t,s \in [0,\infty)$ denote by $u^{ \Lambda}_{t,s}$ its dynamics on $\mathcal{A}_\Lambda$, that is,
\[ u^{ \Lambda}_{t,s}(A)^\Lambda := U^{ \Lambda}(t,s) \,A^\Lambda\,  U^{ \Lambda}(s,t) \]
where $U^{ \Lambda}(t,s)$ is defined as in \eqref{Hdynamics} with $\epsi=1$.
Let $X,Y \subset \Lambda$ with $X\cap Y = \emptyset$ and let $A\in \mathcal{A}_X^+$ be even and $B\in \mathcal{A}_Y$.
Then
\begin{multline*}
\| [ u^{ \Lambda}_{t,s}(A), B] \| \leq \frac{1}{2^{2d } \|F\|_\Gamma} \|A\|\,\|B\|\, \left( \exp\left(  2^{2d+2} \|F\|_\Gamma  \|\Phi \|_{\zeta,0} \cdot|t-s|\right) -1\right)\times \\
\hspace{7cm} \times \sum_{x\in \partial_{\Phi }X}\sum_{y\in Y} F_\zeta(d^\Lambda(x,y))
\end{multline*}
for all $t,s\in [0,\infty)$, $s\leq t$.
\noindent Moreover, 
\[
\sum_{x\in \partial_{\Phi }X}\sum_{y\in Y} F_\zeta(d^\Lambda(x,y)) \leq  \|F \|_\Gamma \min\{|X|,|Y|\} \;\zeta( d^\Lambda(X,Y))\,.
\]
In the case of   $\zeta(r) = \E^{-ar}$ this motivates  the definition of the \emph{Lieb--Robinson velocity}
\begin{equation}\label{vDef}
v:= \frac{ 2^{2d+2} \|F\|_\Gamma \|\Phi \|_{a,0} }{a}\,.
\end{equation}
\end{theorem}

\section{Technicalities on  local Hamiltonians}\label{AppendixTech}

This appendix is devoted to the proof of several results concerning local operators and  local Hamiltonians that were used repeatedly in the proof of the adiabatic theorem, Theorem~\ref{AdiThm}.

We start with a simple lemma that is at the basis of most arguments concerning localization near $L$.
\begin{lemma}\label{distLemma} 
It holds that 
\[
\sum_{y\in\Lambda} F_\zeta (d^\Lambda_L(x,y)) \leq \zeta\left( {\rm dist} (x,L)\right)\,\|F\|_\Gamma \le \|F\|_\Gamma\,.
\]
\end{lemma}
\begin{proof}
One has
\begin{align*}
\sum_{y\in\Lambda} F_\zeta (d^\Lambda_L(x,y)) &= \sum_{y\in\Lambda}\frac{\zeta(d^\Lambda_L(x,y))}{(1 + d^\Lambda_L(x,y))^{d+1}}
\;\leq\; \sup_{y\in\Lambda} \zeta(d^\Lambda_L(x,y))\,\|F\|_\Gamma \\
&\leq \zeta\left({\rm dist}(x,L)\right)\,\|F\|_\Gamma\,.
\end{align*}
The second inequality in the statement follows from the fact that for $\zeta \in \mathcal{S}$ we have $\zeta \le 1$. Indeed, $\zeta(0) = \zeta(0+0) \ge \zeta(0)^2$ implies $1 \ge \zeta(0) \ge \zeta(r)$ for $r \in [0,\infty)$ due to monotonicity of~$\zeta$.
\end{proof}

The next lemma shows that the norm of a  local Hamiltonian localized near $L$ grows at most like the volume of $L$.
\begin{lemma}\label{BoundLemma}
Let $H\in \mathcal{L}_{\zeta,0,L}$, then there is a constant $C_\zeta $ depending only on $\zeta$ such that
\[
  \|H^\Lambda\|\leq M^{d-|\ell|}\,C_\zeta \, \|\Phi_H\|_{\zeta,0,L} 
  \,.
\]
\end{lemma}
\begin{proof} We have
\begin{align*}
\|H^\Lambda\|&\leq
\sum_{Z\subset\Lambda} \|\Phi_H(Z)\| \leq \sum_{x,y\in\Lambda}\sum_{Z \subset \Lambda:\:\{x,y\}\subset Z}\frac{\|\Phi_H^\Lambda(Z)\|}{F_\zeta(d^\Lambda_L(x,y))} F_\zeta(d^\Lambda_L(x,y))\\
&\leq  \|\Phi_H\|_{\zeta,0,L} \sum_{x \in\Lambda} \sum_{ y\in\Lambda} F_\zeta(d^\Lambda_L(x,y)) \leq    \|\Phi_H\|_{\zeta,0,L} \sum_{x \in\Lambda} \zeta\left({\rm dist}(x,L) \right)\,\|F\|_\Gamma \\
&\leq  \|\Phi_H\|_{\zeta,0,L} \sum_{x \in\Gamma} \zeta\left({\rm dist}(x,L) \right)\,\|F\|_\Gamma \leq C_\zeta  \, \|\Phi_H\|_{\zeta,0,L}\, M^{d-|\ell|}
\,,
\end{align*}
since the series in $x$ is summable in $|\ell|$ directions. 
\end{proof} 

We continue with a norm estimate on iterated commutators with  local Hamiltonians all localized near the same $L$.  
\begin{lemma}\label{lemma:comm1}
There is a constant $C_k$ depending only on $k\in\N$ such that
for any $A_1  \in \mathcal{L}_{\zeta,k,L}$, $A_2,\ldots,A_k \in \mathcal{L}_{\zeta,k}$, $X \subset \Lambda$, and   $O\in\mathcal{A}^+_X$ it holds that
\[
\| \add_{A_k^\Lambda} \circ \cdots \circ \add_{A_1^\Lambda} (O)  \|\leq  C_k\,\|O \| \,|X|^k\,   \zeta\left({\rm dist}(X,L) \right) \,\|\Phi_{A_1}\|_{\zeta,k-j,L}\,
\prod_{j=2}^k \|\Phi_{A_j}\|_{\zeta,k-j}\  \,.
\]
\end{lemma}
\begin{proof} 
For better readability we give the proof only for the double commutator. The general statement is then obvious. We estimate 
\begin{align*}
\| &[A_2, [A_1,O]] \| \leq \sum_{\substack{Z_1\subset\Lambda\\ Z_1\cap X\not=\emptyset} } \sum_{\substack{Z_2\subset\Lambda\\ Z_2\cap (X\cup Z_1)\not=\emptyset} }\|[\Phi_{A_2}^{ \Lambda }(Z_2) ,[\Phi_{A_1}^{ \Lambda }(Z_1),O ]\| \\
&\leq  2^2
 \|O \| \sum_{\substack{x_1\in X,\\ y_1\in \Lambda}} F_\zeta(d_L^\Lambda(x_1,y_1)) \sum_{\substack{Z_1 \subset \Lambda: \\ \{x_1,y_1\}\subset Z_1}}
\frac{\|\Phi^{ \Lambda }_{A_1}(Z_1)\|}{F_\zeta(d_L^\Lambda(x_1,y_1))} \times\\
& \qquad\qquad \times
\sum_{\substack{x_2\in X\cup Z_1, \\ y_2\in \Lambda}} F_\zeta(d^\Lambda(x_2,y_2)) \sum_{\substack{Z_2 \subset \Lambda: \\ \{x_2,y_2\}\subset Z_2}}
\frac{\|\Phi^{ \Lambda }_{A_2}(Z_2)\|}{F_\zeta(d^\Lambda(x_2,y_2))}\\
&\leq 
 2^2 
 \|O \| \|\Phi^{ \Lambda }_{A_2}\|_{\zeta,0} \sum_{\substack{x_1\in X, \\ y_1\in \Lambda}} F_\zeta(d_L^\Lambda(x_1,y_1)) \sum_{\substack{Z_1 \subset \Lambda: \\ \{x_1,y_1\}\subset Z_1}}
\frac{\|\Phi^{ \Lambda }_{A_1}(Z_1)\|}{F_\zeta(d_L^\Lambda(x_1,y_1))} 
\sum_{\substack{x_2\in X\cup Z_1, \\ y_2\in \Lambda}} F_\zeta(d^\Lambda(x_2,y_2))\\
&\leq
 2^2 
 \|O \| \|\Phi^{ \Lambda }_{A_2}\|_{\zeta,0} \,\|F\|_\Gamma \sum_{\substack{x_1\in X, \\ y_1\in \Lambda}} F_\zeta(d_L^\Lambda(x_1,y_1))  
 \sum_{\substack{Z_1 \subset \Lambda: \\ \{x_1,y_1\}\subset Z_1}}
\frac{\|\Phi^{ \Lambda }_{A_1}(Z_1)\|}{F_\zeta(d_L^\Lambda(x_1,y_1))} (|X|+|Z_1|).
\end{align*}
Using that $|X|+|Z_1| \le 2 |X|\,|Z_1|$ as $|X|,|Z_1|\ge 1$, we can further bound
\begin{align*}
\| [A_2, [A_1,O]] \| &\leq
 2^3 
 \|O \| \|\Phi^{ \Lambda }_{A_2}\|_{\zeta,0}\,\|\Phi^{ \Lambda }_{A_1}\|_{\zeta,1,L} \,\|F\|_\Gamma \,|X|\,\sum_{\substack{x_1\in X \\ y_1\in \Lambda}} F_\zeta(d_L^\Lambda(x_1,y_1)) \\
 &\leq 
 2^3 
 \|O \| \|\Phi^{ \Lambda }_{A_2}\|_{\zeta,0}\,\|\Phi^{ \Lambda }_{A_1}\|_{\zeta,1,L} \,\|F\|_\Gamma^{2} \,|X|\,\sum_{x\in X}  \zeta\left({\rm dist}(x,L) \right)\\
&\leq  2^3 
 \|O \| \|\Phi^{ \Lambda }_{A_2}\|_{\zeta,0}\,\|\Phi^{ \Lambda }_{A_1}\|_{\zeta,1,L} \,\|F\|_\Gamma^{2} \,|X|^2\,   \zeta\left({\rm dist}(X,L) \right)
\end{align*}
and conclude the proof.
\end{proof}

The next lemma shows that such an iterated commutator of  local $L$-localized Hamiltonians is itself a  local
$L$-localized Hamiltonian. It is an adaption of Lemma~4.6 (ii) in \cite{BDF}.
\begin{lemma}\label{manyadlemma}
Let $n\in\N_0$, $k\in \N$, $A_0 \in \mathcal{L}_{\zeta,n,L}$, and $A_1,\ldots, A_k\in \mathcal{L}_{\zeta,n+k}$. Then $\add_{A_k}\cdots\add_{A_1}(A_0) \in \mathcal{L}_{\zeta,n,L}$
and
\[
\| \Phi_{\add_{A_{k}}\cdots\add_{A_{1}}(A_0)}\|_{\zeta,n,L} \leq C_{k,n}\, \|\Phi_{A_0}\|_{\zeta,n+k,L} 
 \prod_{j=1}^k \|\Phi_{A_j}\|_{\zeta,n+k} 
\]
with a constant $C_{k,n}$ depending only on $k$ and $n$.
In particular, for $A_0 \in  \mathcal{L}_{\zeta,\infty,L}$ and $A_1,\ldots, A_k\in \mathcal{L}_{\zeta,\infty} $ also $\add_{A_k}\cdots\add_{A_1}(A_0) \in \mathcal{L}_{\zeta,\infty,L}$.
\end{lemma}
\begin{proof}
One defines the interaction of a commutator $  [A_1, A_0]$ as
\begin{equation} \label{interaction_comm}
\Phi^\Lambda_{ [A_1, A_0]}(Z) := \sum_{\substack{X_1,X_0\subset \Lambda:\\ X_1\cup X_0=Z, \,X_1\cap X_0\not=\emptyset}} [ \Phi^\Lambda_{A_1}(X_1), \Phi^\Lambda_{A_0}(X_0)]\,.
\end{equation}
We need to estimate the sum
\begin{align*}
\sum_{  Z \subset \Lambda:\: \{x,y\}\subset Z} &|Z|^n \frac{\|\Phi^\Lambda_{ [A_1, A_0]}(Z)\|}{F_\zeta(d_L^\Lambda(x,y))}\\
&\leq 2 \sum_{k=0}^n \binom{n}{k} \hspace{-2mm}\sum_{  Z \subset \Lambda:\: \{x,y\}\subset Z} \hspace{-2mm}\sum_{\substack{X_1,X_0\subset \Lambda:\\ X_1\cup X_0=Z, \,X_1\cap X_0\not=\emptyset}} \hspace{-8mm}|X_0|^k \,|X_1|^{n-k} \frac{\| \Phi^\Lambda_{A_1}(X_1)\|\,\| \Phi^\Lambda_{A_0}(X_0)\|}{F_\zeta(d_L^\Lambda(x,y))}
\end{align*}
uniformly in $x,y$ and $\Lambda$. One now splits the sum into four parts which are estimated separately: $X_0\cap \{x,y\}$ is either $\emptyset$, $\{x\}$, $\{y\}$, or $\{x,y\}$.
The part of the sum where  $x,y\in X_0$  can be estimated by
\begin{gather*}
  2 \sum_{k=0}^n \binom{n}{k}\sum_{  X_0\ni  x,y  }
  |X_0|^k\frac{\| \Phi^\Lambda_{A_0}(X_0)\|}{F_\zeta(d_L^\Lambda(x,y))}
  \sum_{z_0\in X_0} \sum_{z_1\in\Lambda} \sum_{X_1\ni z_0,z_1}
  |X_1|^{n-k} \frac{\| \Phi^\Lambda_{A_1}(X_1)\|\,}{F_\zeta(d^\Lambda(z_0,z_1))}F_\zeta(d^\Lambda(z_0,z_1))\\
  \leq 2 \sum_{k=0}^n \binom{n}{k} \|F \|_\Gamma   \|\Phi_{A_0}\|_{\zeta,k+1,L} \|\Phi_{A_1}\|_{\zeta,n-k}\\
  \leq 2^{n+1} \|F \|_\Gamma   \|\Phi_{A_0}\|_{\zeta,n+1,L} \|\Phi_{A_1}\|_{\zeta,n} \,,
\end{gather*}
 where we used  $\|\Phi\|_{\zeta,n,L}\leq \|\Phi\|_{\zeta,m,L}$ whenever $n\leq m$.
The part of the sum where  $x\in X_0$ but $y\in X_1\setminus X_0$  can be estimated by
\begin{gather*}
  2 \sum_{k=0}^n \binom{n}{k}\sum_{  z\in\Lambda  } \frac{F_\zeta(d^\Lambda(z,y))F_\zeta(d_L^\Lambda(x,z))}{F_\zeta(d_L^\Lambda(x,y))}\sum_{  X_0\ni x,z}
  |X_0|^k\frac{\| \Phi^\Lambda_{A_0}(X_0)\|}{F_\zeta(d_L^\Lambda(x,z))}
  \sum_{X_1\ni z,y}
  |X_1|^{n-k} \frac{\| \Phi^\Lambda_{A_1}(X_1)\|\,}{F_\zeta(d^\Lambda(z,y))} \\
  \leq 2^{n+1}  2^d 2^{ 2d+2}\|F\|_\Gamma  \|\Phi_{A_0}\|_{\zeta,n,L} \|\Phi_{A_1}\|_{\zeta,n}\,.
\end{gather*}
For the remaining cases just interchange the role of $A_1$ and $A_0$. We can finally collect the four estimates and find that
\[
\| \Phi_{[A_1,A_0]}\|_{\zeta,n,L} \leq 2^{n+2} (\|F \|_\Gamma +2^{3d+2} \|F \|_\Gamma)  \|\Phi_{A_0}\|_{\zeta,n+1,L} \|\Phi_{A_1}\|_{\zeta,n+1,L}\,.
\]
The rest follows by induction.
\end{proof}

We also need to control the norm of commutators with time-evolved local observables. This is the content of the next lemma, which is adapted from Lemma~4.7 in~\cite{BDF}.
\begin{lemma}\label{CommuLemma}
Let $H\in \mathcal{L}_{a,0}$ generate the dynamics $u^{ \Lambda}_{t,s}$ with Lieb--Robinson velocity $v$ as in \eqref{vDef}. Then there exists a constant $C>0$ such that for any $O\in \mathcal{A}_X^+$ with $X\subset \Lambda$, for any $A\in \mathcal{L}_{\zeta,0,L}$  and for any $t,s\in [0,\infty)$   it holds that 
\[
\| [A,u^{ \Lambda}_{t,s}(O)]\| \leq C \|O\| \|\Phi_A\|_{\zeta,0,L} |X|^2 \,\zeta({\rm dist}(X,L  )) \,(1+|t-s|^d)\,.
\]

 If $H\in \mathcal{L}_{\tilde\zeta,0}$ for some $\tilde \zeta\in\mathcal{S}$, one still has that for any $T>0$ there exists a constant $C$ such that
 \[
 \sup_{t,s\in [0,T]}\| [A,u^{ \Lambda}_{t,s}(O)]\| \leq C \|O\| \|\Phi_A\|_{\zeta,0,L} |X|^2 \,\zeta({\rm dist}(X,L  ))  \,.
 \]
\end{lemma}

\begin{proof}
We consider first the case $H \in \mathcal{L}_{a,0}$. One uses the following property of partial traces, proved in  
\cite{NSW}.
\begin{lemma}[Lemma~2.1 of \cite{NSW}]\label{NSWLemma}
Let $\Hi_1$ and $\Hi_2$ be Hilbert spaces. Then the partial trace $\mathbb{E}: \mathcal{B}(\Hi_1\otimes\Hi_2)\to \mathcal{B}(\Hi_1)$ is a completely positive linear map with the following property:
Whenever $A\in  \mathcal{B}(\Hi_1\otimes\Hi_2)$ satisfies the commutator bound 
\[
\| [A,{\bf 1}\otimes B]\|\leq \eta\|A\|\|B\|\quad\mbox{ for all }\;B\in \mathcal{B}(\Hi_2)
\]
for some $\eta>0$, then 
\[
\| A - \mathbb{E}(A)\otimes {\bf 1}\|\leq \eta \|A\|\,.
\]
\end{lemma}

We decompose $\Lambda$ into regions $X _{v|t-s| +k}\equiv X^\Lambda_{v|t-s| +k}$, where
for any $Y\subset \Lambda$ and $\delta \geq 0$ we let
\[
Y^\Lambda_\delta := \{ z\in\Lambda\,|\, d^\Lambda(z,Y)\leq \delta\}
\]
be the ``fattening'' of the set $Y$ by $\delta$ in $\Lambda$.
Moreover, denote by $\mathbb{E}_Y: \mathcal{A}_\Lambda = \mathcal{A}_Y \otimes \mathcal{A}_{\Lambda\setminus Y} \to \mathcal{A}_Y$ the corresponding partial trace.
Defining
\[
O^{(0)} := \mathbb{E}_{X_{v|t-s|}} (u^{ \Lambda}_{t,s}(O))
\]
and for $k\geq 1$
\begin{align*}
O^{(k)} &:= \mathbb{E}_{X_{v|t-s|+k}} (u^{ \Lambda}_{t,s}(O)) - \mathbb{E}_{X_{v|t-s|+k-1}} (u^{ \Lambda}_{t,s}(O))\\
&= \mathbb{E}_{X_{v|t-s|+k}} \left( ({\bf 1} - \mathbb{E}_{X_{v|t-s|+k-1}} )u^{ \Lambda}_{t,s}(O)\right)
\end{align*}
we can write $u^{ \Lambda}_{t,s}(O) = \sum_{k=0}^\infty O^{(k)}$, where the sum is always finite, since eventually $X_{v|t-s|+k}=\Lambda$. According to Lemma~\ref{lemma:comm1} we have
\[
\| [A,O^{(k)}]\| \leq 2 \|O^{(k)}\|  \, \|F\|_\Gamma\,\|\Phi_A\|_{\zeta,0 ,L }\,\zeta({\rm dist}(X,L))\times \left\{ \begin{array}{cl} |X|\,(2(1+k))^d  &   \mbox{if } v|t-s|\leq 1\\
|X|\, (4kv|t-s|)^d& \mbox{if } v|t-s|> 1 \end{array}\right.\,,
\]
since
\[
|X_{v|t-s|+k}| \leq |X|\,(2(v|t-s|+k))^d \leq \left\{ \begin{array}{cl} |X|\,(2(1+k))^d  &   \mbox{if } v|t-s|\leq 1\\
|X|\, (4kv|t-s|)^d& \mbox{if } v|t-s|> 1\,.\end{array}\right.
\]
Since $\|\mathbb{E}_Y\|=1$, the norm of $O^{(k)}$ is estimated as
\[
\|O^{(k)}\|\leq \| u^{ \Lambda}_{t,s}(O) - \mathbb{E}_{X_{v|t-s|+k-1}} (u^{ \Lambda}_{t,s}(O))\otimes {\bf 1}\|\,.
\]
Using the Lieb--Robinson bound we find for any $B\in \mathcal{A}_{\Lambda\setminus X_{v|t-s|+k-1}}$ that
\begin{align*}
\| [ &u^{ \Lambda}_{t,s}(O), B] \| \\
&\leq \frac{1}{2^{2d }\|F\|_\Gamma} \|O\|\,\|B\|\, \left( \E^{a v \cdot|t-s|} -1\right) 
\|F \|_\Gamma \min\{|X|,|\Lambda\setminus X_{v|t-s|+k-1}|\} \;\E^{-a\cdot (v|t-s|+k)}\\
&\leq  2^{-2d} \|O\|\,\|B\|   \,|X|  \,\E^{-ak}
\end{align*}
and thus by Lemma~\ref{NSWLemma}
\[
\|O^{(k)}\|\leq  2^{-2d}  \|O\|\,  |X|  \,\E^{-ak} \,.
\]
Summing up, we conclude that for $v|t-s|> 1$ we have 
\[
\| [A,u^{ \Lambda}_{t,s}(O)]\| \leq    2^{-2d} v^d  \|O\|\,\|F \|_\Gamma  \,\zeta({\rm dist}(X,L))\,|X|^2 \|\Phi_A\|_{\zeta,0,L}\,    |t-s|^d \sum_{k=0}^\infty k^d\E^{-ak} 
\]
and a similar estimate in the case $v|t-s|\leq 1$. 

The bound for Hamiltonians in $\mathcal{L}_{\tilde\zeta,0}$, with $\tilde\zeta \in \mathcal{S}$, follows analogously by just using $X_k$ instead of $X_{v|t-s|+k}$, and noting that the exponential $\exp\left( 2^{2d+2}\|F\|_\Gamma  \|\Phi_H \|_{\tilde\zeta,0} \cdot|t-s|\right)$ from the Lieb--Robinson bound is bounded for $t,s$ in bounded sets.
\end{proof}

The final lemma in this appendix shows that adjoining a  local $L$-localized Hamiltonian with a unitary that is itself the exponential of a  local Hamiltonian yields again a  local and $L$-localized Hamiltonian. Here we adapted Lemma~4.8 from \cite{BDF}.
\begin{lemma}\label{transformlemma}
Let $S\in \mathcal{L}_{\zeta,0}$ be self-adjoint and let $D \in \mathcal{L}_{\mathcal{S},\infty,L}$, i.e.\ there is a sequence $(\tilde\zeta_n)_{n\in\N_0}$ in $\mathcal{S}$ such that $ \|\Phi_D\|_{\tilde \zeta_n,n+1,L}<\infty$. Then the family of operators
\[
\left\{A^\Lambda := \E^{-\I S^\Lambda} \,D^\Lambda\, \E^{\I S^\Lambda} \right\}_\Lambda
\]
defines a Hamiltonian $A\in \mathcal{L}_{\mathcal{S},\infty,L}$. More precisely, there is a constant $C_{\| \Phi_S\|_{\zeta,0 } }$ depending  on $\| \Phi_S\|_{\zeta,0 }$,   $\zeta$, $(\tilde\zeta_n)_{n\in\N_0}$, and $d$, and a sequence  $(\xi_n)_{n\in\N_0}$ in $ \mathcal{S}$,   such that
\[
\| \Phi_A \|_{\xi_n,n,L} \leq C_{\| \Phi_S\|_{\zeta,0 } } \, \|\Phi_D\|_{\tilde \zeta_n,n+1,L}
\]
for all $n\in\N_0$.
\end{lemma}
\begin{proof}
We use the strategy and the notation from Lemma~\ref{CommuLemma}. For $O\in\mathcal{A}_X$, $X\subset \Lambda$,  
  define
\[
\Delta_0^\Lambda(O)  := \mathbb{E}_{X } ( \E^{-\I S^\Lambda} \,O\, \E^{\I S^\Lambda}  )
\]
and for $k\geq 1$
\[
\Delta_k^\Lambda(O) :=  \mathbb{E}_{X_{ k}} \left( ({\bf 1} - \mathbb{E}_{X_{ k-1}} )\, \E^{-\I S^\Lambda} \,O\, \E^{\I S^\Lambda}\right)\,.
\]
Again we have  $\E^{-\I S^\Lambda} \,O\, \E^{\I S^\Lambda} = \sum_{k=0}^\infty \Delta_k^\Lambda(O)$, where the sum is always finite, since eventually $X_{ k}=\Lambda$. As $\|\mathbb{E}_Y\|=1$, the norm of $\Delta_k^\Lambda(O)$ is estimated as
\[
\|\Delta_k^\Lambda(O)\|\leq \| \E^{-\I S^\Lambda} \,O\, \E^{\I S^\Lambda}  - \mathbb{E}_{X_{ k-1}} (\E^{-\I S^\Lambda} \,O\, \E^{\I S^\Lambda} )\otimes {\bf 1}\|\,.
\]
Since $S\in \mathcal{L}_{\zeta,0 }$,  using the Lieb--Robinson bound we find for any $B\in \mathcal{A}_{\Lambda\setminus X_{ k-1}}$ that
\[
\| [ \E^{-\I S^\Lambda} \,O\, \E^{\I S^\Lambda} , B] \|  \leq  
2^{-2d} \|O\|\,\|B\|\, \left(\E^{2^{2d+2}\|F\|_\Gamma \|\Phi_S \|_{\zeta,0 }} -1\right) 
   |X|  \;\zeta(k)
\]
and thus by Lemma~\ref{NSWLemma}
\begin{equation}\label{lastlemma1}
\|\Delta_k^\Lambda(O)\|\leq  2^{-2d} \|O\| \, \left( \E^{2^{2d+2}\|F\|_\Gamma \|\Phi_S \|_{\zeta,0 }} -1\right) 
   |X|  \;\zeta(k)\,.
\end{equation}
An interaction for $A$ can now be defined by
\begin{equation}\label{lastlemma2}
\Phi_A^\Lambda(Z) := \sum_{k=0}^\infty \sum_{Y\subset \Lambda\,:\, Y_k = Z} \Delta_k^\Lambda(\Phi^\Lambda_D(Y))\,.
\end{equation}
We can rewrite
\begin{align*}
\sum_{Z\subset\Lambda\,:\, \{x,y\}\subset Z} &|Z|^n \|\Phi^\Lambda_A (Z)\| \leq 
 \sum_{ Z\subset\Lambda\,:\, \{x,y\}\subset Z } \sum_{k=0}^\infty \sum_{Y\subset \Lambda\,:\, Y_k = Z} |Z|^n\, \|\Delta_k^\Lambda(\Phi^\Lambda_D(Y))\|\\
 &\leq 
 \sum_{Y\subset\Lambda} \sum_{k=0}^\infty {\bf 1}(x,y\in Y_k) \,|Y_k|^n \, \|\Delta_k^\Lambda(\Phi^\Lambda_D(Y))\|\\
 &\leq
 \sum_{Y\subset\Lambda\,:\, \{x,y\}\subset Y =Y_0 } \sum_{k=0}^\infty  (2k)^{dn} \,|Y |^n \, \|\Delta_k^\Lambda(\Phi^\Lambda_D(Y))\| \\
 &\quad + \; \sum_{m=1}^\infty  \sum_{Y\subset\Lambda\,:\, \{x,y\}\subset Y_m}  {\bf 1}(\{x,y\}\cap Y_{m-1}^c \not= \emptyset)
 \sum_{k\geq m} (2k)^{dn}\,|Y |^n \,\|\Delta_k^\Lambda(\Phi^\Lambda_D(Y))\|\\
 &=: S_1 + S_2\,.
\end{align*}
With \eqref{lastlemma1} one finds
\begin{align*}
S_1 &\leq   C \left( \sum_{Y\subset\Lambda\,:\, \{x,y\}\subset Y} |Y|^{n+1}\,\|\Phi^\Lambda_D(Y)\|\right) \left( \sum_{k=0}^\infty  (2k)^{dn}\zeta(k)\right)\\
 &\leq  
 C \, 
\|\Phi_D\|_{\tilde \zeta_n,n +1 ,L}\,F_{\tilde \zeta_n} (d^\Lambda_L(x,y))\,
 \,.
\end{align*}
For $S_2$ first note that
\[
\sum_{Y\subset\Lambda\,:\, \{x,y\}\subset Y_m}  {\bf 1}(\{x,y\}\cap Y_{m-1}^c \not= \emptyset) \leq \sum_{z_1\in B_m(x)} \sum_{z_2\in B_m(y)} 
\sum_{Y\subset \Lambda\,:\, \{z_1,z_2\}\subset Y} \,1\,
\]
together with  \eqref{lastlemma1}  shows that
\begin{align*}
S_2&\leq \sum_{m=1}^\infty \sum_{z_1\in B_m(x)} \sum_{z_2\in B_m(y)} 
\sum_{Y\subset \Lambda\,:\, \{z_1,z_2\}\subset Y} \,|Y|^{n+1} \,\,\|\Phi^\Lambda_D(Y)\|\, \sum_{k\geq m}  (2k)^{dn}\zeta(k) \\
&\leq 
 C \, 
\|\Phi_D\|_{\tilde \zeta_n,n+1,L}\,
\sum_{m=1}^\infty  \sum_{k\geq m}   (2k)^{dn}\zeta(k)\, \sum_{z_1\in B_m(x)} \sum_{z_2\in B_m(y)} 
F_{\tilde \zeta_n} (d^\Lambda_L(z_1,z_2))\,
 \,.
\end{align*}
The triangle inequality implies that
\[
 d^\Lambda_L(x,y)\leq  d^\Lambda_L(x,z_1)+ d^\Lambda_L(z_1,z_2)+ d^\Lambda_L(z_2,y) \leq 2m +  d^\Lambda_L(z_1,z_2)
\]
and with $m_0 =  \lfloor d^\Lambda_L(x,y)/4\rfloor$ we have for $m\leq m_0$ that $d^\Lambda_L(z_1,z_2)\geq d^\Lambda_L(x,y)/2$.
Thus
\begin{align*}
S_{2,1}&:= \sum_{m=1}^{m_0}  \sum_{k\geq m}  (2k)^{dn}\zeta(k)\, \sum_{z_1\in B_m(x)} \sum_{z_2\in B_m(y)} 
F_{\tilde \zeta_n} (d^\Lambda_L(z_1,z_2))\,
\\
 &\leq F_{\tilde \zeta_n} (d_L^\Lambda(x,y)/2)   \left( \sum_{k=0}^\infty  (2k)^{dn}\zeta(k)\right) \sum_{m=1}^{m_0} (2m)^{ 2d}
\,.
\end{align*}
Now let $m_0(r) := \lfloor r/4\rfloor$. Then the function
\[
f_1(r):= F_{\tilde \zeta_n} (r/2) \,(r +1)^{d+1} \, \sum_{m=1}^{m_0(r)}  m ^{ 2d} = \frac{\tilde\zeta_n(r/2)}{(r/2+1)^d} \,(r +1)^{d+1} \, \sum_{m=1}^{m_0(r)}  m ^{ 2d}
\]
satisfies the assumption of  part (b) of  Lemma~\ref{Slemma} and is thus, up to a constant factor,  bounded by some function $\xi_{n,1}\in\mathcal{S}$. We conclude that
$S_{2,1}\leq C \,F_{\tilde \xi_{n,1}} (d^\Lambda_L(x,y) )$.
The rest of $S_2$ is 
\begin{align*}
S_{2,2}&:=\sum_{m>m_0}  \sum_{k\geq m}   (2k)^{dn}\zeta(k)\, \sum_{z_1\in B_m(x)} \sum_{z_2\in B_m(y)} 
F_{\tilde \zeta_n} (d^\Lambda_L(z_1,z_2))\\ &\leq \|F  \|_\Gamma \sum_{m>m_0}  (2m)^d\sum_{k\geq m}   (2k)^{dn}\zeta(k)\,,
\end{align*}
and, as before, the function
\[
f_2(r):= (r +1)^{d+1} \sum_{m>m_0(r)}  (2m)^d\sum_{k\geq m}   (2k)^{dn}\zeta(k)
\]
satisfies the assumption of Lemma~\ref{Slemma} and is thus, up to a constant factor,  bounded by some function $\xi_{n,2}\in\mathcal{S}$. We conclude that
$S_{2,2}\leq C \,F_{\tilde \xi_{n,2}} (d^\Lambda_L(x,y) )$.

In summary we proved that 
\[
\sum_{Z\subset\Lambda\,:\, \{x,y\}\subset Z} |Z|^n \|\Phi^\Lambda_A (Z)\|  \leq C \,\|\Phi_D\|_{\tilde \zeta_n,n+1,L}\,\left(F_{\tilde \zeta_n} (d^\Lambda_L(x,y))
+F_{\tilde \xi_{n,1}} (d^\Lambda_L(x,y) )+F_{\tilde \xi_{n,2}} (d^\Lambda_L(x,y) )
\right)
\]
and Lemma~\ref{Slemma}  (a)  implies that $\Phi_A \in \mathcal{B}_{\xi_n,n+1,L}$ for some $\xi_n\in\mathcal{S}$.
 \end{proof}

\section{Local inverse of the Liouvillian}\label{AppendixI}
In this appendix we  discuss  the map $\mathcal{I}_H  = \mathcal{I}_{H,g}$ and  prove its properties used in the proof of the adiabatic theorem.

We start with some abstract considerations: Let $H\in \mathcal{B}(\Hi)$ be a self-adjoint operator on some finite-dimensional Hilbert space $\Hi$ and $P\in \mathcal{B}(\Hi)$ a spectral projection of $H$. The inner product $\langle A,B\rangle := \tr A^* B$ turns the algebra  $\mathcal{A} := \mathcal{B}(\Hi)$ into a Hilbert space that splits into the orthogonal sum $\mathcal{A} = \mathcal{A}^{\rm D}_P \oplus \mathcal{A}^{\rm OD}_P$  of the subspaces of diagonal and off-diagonal operators  with respect to $P$, 
\[
A = (PAP + P^\perp AP^\perp) + (P^\perp AP + P AP^\perp)=: A^{\rm D}_P + A^{\rm OD}_P\,.
\] 
The linear map (called Liouvillian)
\[
\ad_H:\mathcal{A}\to \mathcal{A}\,,\quad A \mapsto \ad_H(A) := -\I\, \add_H(A) = -\I [H,A]
\]
is skew-adjoint,
\[
\langle \ad_H(A), B\rangle = \I \,\tr ( [H,A]^*B) =     \I\, \tr (  A^* H  B) - \I\, \tr (   A^*  B H)  = -\langle A , \ad_H(B)\rangle\,,
\]
and thus
\[
{\rm ran}(\ad_H) = ({\rm ker}(\ad_H))^\perp\,.
\]
Let $A\in {\rm ker}(\ad_H)$, i.e.\   $[H,A]=0$. Then also $[A,P]=0$ and thus $ {\rm ker}(\ad_H)\subset  \mathcal{A}^{\rm D}_P$; consequently $\mathcal{A}^{\rm OD}_P \subset  {\rm ker}(\ad_H)^\perp= {\rm ran}(\ad_H) $. Hence for $B\in \mathcal{A}^{\rm OD}_P$ the equation
\begin{equation}\label{abstractCommu}
\ad_H(A) =B 
\end{equation}
has a unique solution   $A\in {\rm ker}(\ad_H)^\perp$. Since with $B^{\rm D}_P=0$ also $[H,A]^{\rm D}_P = [H, A^{\rm D}_P]=0$, 
for any solution $A$ of \eqref{abstractCommu}  $A^{\rm D}_P\in{\rm ker}(\ad_H) $. Hence, for $B\in \mathcal{A}^{\rm OD}_P$ the unique solution to \eqref{abstractCommu} in $  {\rm ker}(\ad_H)^\perp$ is actually off-diagonal, i.e.\ lies in $\mathcal{A}^{\rm OD}_P$. In summary we conclude that the map 
\[
\ad_{H,P}:\mathcal{A}^{\rm OD}_P \to \mathcal{A}^{\rm OD}_P\,,\quad A \mapsto \ad_H(A) =-\I\, [H,A]
\]
is an isomorphism and we denote its inverse by $\ad_{H,P}^{-1}$. 

Note that if   $P$ is the spectral projection of $H$ corresponding to a single eigenvalue $E$, i.e.\ $HP=EP$, then there exists an explicit formula for  $\ad_{H,P}^{-1} $,
\[
\ad_{H,P}^{-1} (B) =\I\,[(H-E)^{-1} P^\perp,B]=:\I\, [ R , B]\,,
\]
in terms of the reduced resolvent $R := (H-E)^{-1} P^\perp$. This can be checked by a simple computation:
\begin{align*}
[H, [ R , B]] &= [H - E, [ (H-E)^{-1} P^\perp , B]]\\
&= P^\perp   B - (H-E) B  (H-E)^{-1} P^\perp  -  (H-E)^{-1} P^\perp B (H-E) + BP^\perp\\
&= P^\perp   B - (H-E) P B P^\perp (H-E)^{-1}    -  (H-E)^{-1} P^\perp B P (H-E) + BP^\perp\\
&= P^\perp   B+ BP^\perp = B\,,
\end{align*}
where we used repeatedly that $B$ is off-diagonal and that $(H-E) P=0$.

One key ingredient to the proof of the adiabatic theorem is the following extension of the inverse Liouvillian $\ad_{H,P}^{-1}$ to a map on the full space $\mathcal{A}$. 
To construct it, first note that for any $g>0$ one can find a function
  $W_g\in L^1(\R)$ satisfying $\sup\{|s|^n |W_g(s)|\,|\, |s|>1\}<\infty$ for all $n\in\N$ and with a Fourier transform $\widehat W_g \in C^\infty(\R)$ satisfying
  \[
\widehat W_g(\omega) = \frac{-\I}{\sqrt{2\pi}\omega} \quad\mbox{ for } |\omega|\geq g\qquad\mbox{ and }\qquad   \widehat W_g(0)=0\,.
\]
 An example of a function $W_g$ having all these properties      is given in \cite{BMNS}. 
 
  We need a slightly modified version $\mathcal{W}_{g,\delta}$ of this function: Let $g>\delta \geq 0$ and $\chi_{g,\delta}\in C^\infty(\R)$ an even function with $\chi_{g,\delta}(\omega) = 0$ for $\omega \in [-\delta,\delta]$ and $\chi_{g,\delta}(\omega) = 1$ for $|\omega|\geq g$. Then $\mathcal{W}_{g,\delta}$ defined through its Fourier transform $\widehat{\mathcal{W}}_{g,\delta}:=  \chi_{g,\delta}\,\widehat{\mathcal{W}}_{g }$ 
 satisfies, in addition to the properties mentioned above for $\mathcal{W}_{g }$, also  $ \widehat{\mathcal{W}}_{g,\delta}(\omega)=0$ for all $\omega\in  [-\delta,\delta]$.
 
 Now let $H$ be a self-adjoint operator and assume that $\sigma_*(H)\subset \sigma(H)$ is a set of neighbouring eigenvalues, i.e.\ there is an interval $I\subset \R$ such that $I\cap \sigma(H) = \sigma_*(H)$.
Let   $g:= {\rm dist}( \sigma_*(H), \sigma(H)\setminus\sigma_*(H)) > 0$ be the size of the spectral gap, $\delta := {\rm diam}(\sigma_*(H))$ the width of the spectral patch $\sigma_*(H)$, and $P_*$ the corresponding spectral projection. 
 
 \begin{lemma}\label{Ilemma}
Let $H$, $\sigma_*$, and $P_*$ as above. Then for any $0\leq \tilde\delta <g$ the map 
\[
\mathcal{I}_{H,g,\tilde\delta} : \mathcal{A}\to \mathcal{A}\,, \quad \mathcal{I}_{H,g,\tilde\delta}(A) := \int_\R W_{g,\tilde\delta}(s) \,\E^{\I Hs}\,A\,\E^{-\I Hs} \,\D s
\]
satisfies 
\begin{equation}\label{ILinverse}
\mathcal{I}_{H,g,\tilde\delta}|_{\mathcal{A}^{\rm OD}_{P_*}} =  \ad_{H,P_*}^{-1}\,.
\end{equation}
If $\delta\leq \tilde \delta$, then moreover
\[
P_*\,\mathcal{I}_{H,g,\tilde\delta}(A)\,P_* = 0 \quad\mbox{ for all } A\in \mathcal{A}\,.
\]
\end{lemma}
\begin{proof}
Both claims follow immediately by inserting the spectral decomposition of $H=\sum_n E_n P_n$, where we enumerate the eigenvalues such that $P_* = \sum_{n=1}^{n_*} P_n$,  into the definition of~$\mathcal{I}$:
\begin{align*}
\frac{1}{\sqrt{2\pi}}\,\mathcal{I}_{H,g,\tilde\delta}(A) &= \sum_{n,m} \widehat W_{g,\tilde\delta}(E_m-E_n)  P_n\,A\,P_m\\
&=  \sum_{n,m\leq n_*} \widehat W_{g,\tilde\delta}(E_m-E_n)  P_n\,A\,P_m+ \sum_{n,m>n_* }  \widehat W_{g,\tilde\delta}(E_m-E_n)  P_n\,A\,P_m\\
& \quad +\; \sum_{\substack{n\leq n_* \\ m>n_*}}  \widehat W_{g,\tilde\delta}(E_m-E_n)  P_n\,A\,P_m+ \sum_{\substack{n> n_* \\ m\leq n_*}}   \widehat W_{g,\tilde\delta}(E_m-E_n)  P_n\,A\,P_m\\
&= \sum_{n,m\leq n_*} \widehat W_{g,\tilde\delta}(E_m-E_n)  P_n\,A\,P_m+ \sum_{n,m>n_* }  \widehat W_{g,\tilde\delta}(E_m-E_n)  P_n\,A\,P_m\\
& \quad + \frac{ \I}{\sqrt{2\pi} } \left(  \sum_{\substack{n\leq n_* \\ m>n_*}}  \frac{P_n\,A\,P_m}{E_n-E_m}+ \sum_{\substack{n> n_* \\ m\leq n_*}}   \frac{P_n\,A\,P_m}{E_n-E_m}  \right)\,.
\end{align*}
For $A=A^{\rm OD}$ the first two terms in the final expression vanish and \eqref{ILinverse} is evident. In the first term we have $|E_m-E_n|<\delta$, and for $\delta \leq \tilde\delta$ this term thus vanishes identically.
\end{proof}
 
The usefulness  of  the map $\mathcal{I}_H$ in the context of  local Hamiltonians lies in the fact that it maps  local Hamiltonians to  local Hamiltonians, an observation originating from \cite{HW}.  Based on \cite{BMNS} and Lemma~4.8 in \cite{BDF} one can now show that it also preserves $L$-localization, i.e.\ that
\[
\mathcal{I}_{H,g}( \mathcal{L}_{\mathcal{S},\infty,L}) \subset \mathcal{L}_{\mathcal{S},\infty,L}\,,
\]
whenever $H\in \mathcal{L}_{\mathcal{E},\infty}$.
More precisely, one has the following lemma, which is an adaption of Lemma~4.8 in \cite{BDF}.

\begin{lemma}\label{Ilemma2}
Assume (A1)$_m$ and (A2) and let $D \in \mathcal{L}_{\mathcal{S},\infty,L}$, i.e.\ there is a sequence $(\tilde\zeta_n)_{n\in\N_0}$ in $\mathcal{S}$ such that $ \|\Phi_D\|_{\tilde \zeta_n,n+1,L}<\infty$. Then the family of operators
\[
 \big\{  \mathcal{I}_{H,g, \tilde \delta}(D)^{ \Lambda }  \big\}_{\Lambda} 
\]
defines a Hamiltonian $\mathcal{I}_{H,g,\tilde\delta}(D)\in \mathcal{L}_{\mathcal{S},\infty,L}$. More precisely, there is a constant $C_{\| \Phi_H\|_{a,0} }$ depending  on  $\|W_g\|_1$,  $\| \Phi_H\|_{a,0}$, $(\tilde\zeta_n)_{n\in\N_0}$,  $a$, and $d$, and a sequence  $(\xi_n)_{n\in\N_0}$ in $ \mathcal{S}$,   such that
\[
\| \Phi_{\mathcal{I}_{H,g,\tilde\delta}(D)}(t) \|_{\xi_n,n,L} \leq C_{\| \Phi_{H(t)}\|_{a,0} } \, \|\Phi_D(t)\|_{\tilde \zeta_n,n+1,L}
\]
for all $n\in\N_0$ and $t\in[0,\infty)$.
\end{lemma}
We omit the proof as one can combine an approach similar to the proof of Lemma~\ref{transformlemma} and the superpolynomial decay of $W_{g,\tilde \delta}$ exactly as in \cite{BMNS} to arrive at the conclusion.


\begin{thebibliography}{00}

\bibitem{AS}
J.\ Avron and R.\ Seiler: Quantization of the Hall conductance for general, multiparticle Schr\"odinger Hamiltonians.
Physical Review Letters 54:259  (1985).

\bibitem{ASS}
J.\ Avron, R.\ Seiler,   and B.\ Simon: Charge deficiency, charge transport and comparison of dimensions. Communications in Mathematical Physics 159:399--422 (1994).

\bibitem{ASY}
J.\ Avron, R.\ Seiler, and L.\ Yaffe:  Adiabatic theorems and applications to the quantum Hall effect.  Communications in Mathematical Physics 110:33--49 (1987).

\bibitem{BBDF}
S.\ Bachmann, A.\ Bols, W.\ De Roeck, and M.\ Fraas:
Quantization of conductance in gapped interacting systems.
Annales Henri Poincar{\'e} 19:695--708 (2018). 

\bibitem{BDFletter}
S.\ Bachmann, W.\ De Roeck, and M.\ Fraas: 
Adiabatic theorem for quantum spin systems.
Physical Review Letters 119:060201 (2017).


\bibitem{BDF}
S.\ Bachmann, W.\ De Roeck, and M.\ Fraas:
The adiabatic theorem and linear response theory for extended quantum systems.
Communications in Mathematical Physics  361:997--1027 (2018).



\bibitem{BMNS}
S.\ Bachmann, S.\ Michalakis, B.\ Nachtergaele, and R.\ Sims:  Automorphic Equivalence within Gapped Phases of Quantum Lattice Systems. Communications in Mathematical Physics 309:835--871 (2012).

\bibitem{BES}
J.\ Bellissard, A.\ van Elst, and H.\ Schulz-Baldes:  The noncommutative geometry of the quantum Hall effect.  Journal of Mathematical Physics 35:5373--5451 (1994).

\bibitem{BGKS}
J.\ Bouclet,  F.\ Germinet, A.\ Klein, and J.\ Schenker: Linear response theory for magnetic Schr\"odinger operators in disordered media. Journal of Functional Analysis 226:301--372 (2005).

\bibitem{BD}
J.-B.\ Bru and W.\ de Siqueira Pedra:
Lieb--Robinson Bounds for Multi-Commutators and Applications to Response Theory. 
Springer Briefs in Mathematical Physics Vol.\ 13, Springer  (2016).

\bibitem{BDH}
J.-B.\ Bru, W.\ de Siqueira Pedra, and C.\ Hertling:  Microscopic Conductivity of Lattice Fermions at Equilibrium. Part II: Interacting Particles.  Letters in Mathematical Physics 106:81--107  (2016).

\bibitem{DNL}
G.\ De Nittis and M.\ Lein: 
Linear Response Theory -- An Analyitic-Algebraic Approach.
SpringerBriefs in Mathematical Physics 21, Springer (2017).

\bibitem{DS}
W.\ de Roeck and M.\  Salmhofer: Persistence of exponential decay and spectral gaps for interacting fermions. 
To appear in Communications in Mathematical Physics, \href{https://doi.org/10.1007/s00220-018-3211-z}{\texttt{DOI 10.1007/s00220-018-3211-z}} (2018).


\bibitem{F}
M.~Fekete:
\"{U}ber die Verteilung der Wurzeln bei gewissen algebraischen Gleichungen mit ganzzahligen Koeffizienten. 
Mathematische Zeitschrift 17:228 (1923). 

\bibitem{Fr}
J.\ Fr\"ohlich:
Chiral Anomaly, Topological Field Theory, and Novel States of Matter.
Reviews in Mathematical Physics 30:1840007 (2018).

\bibitem{GMP}
A.\ Giuliani, V.\ Mastropietro, and M.\ Porta: Universality of the Hall Conductivity in Interacting Electron Systems. Communications in Mathematical Physics 349:1107--1161 (2017).

\bibitem{H}
M.\ Hastings: The Stability of Free Fermi Hamiltonians. Preprint available at 
\href{http://arxiv.org/abs/1706.02270}{\texttt{arXiv:1706.02270}} (2017).

\bibitem{HM}
M.\ Hastings  and S.\ Michalakis:
Quantization of Hall Conductance for Interacting Electrons on a Torus.
Communications in Mathematical Physics 334:433--471 (2015).

\bibitem{HW}
M.\ Hastings  and  X.-G.\ Wen:
Quasiadiabatic continuation of quantum states: The stability of topological ground-state degeneracy and emergent gauge invariance.
 Physical Review B 72:045141 (2005).
 
\bibitem{JOP}
V.\ Jak\v{s}i\'{c}, Y.\ Ogata, and Cl.-A.\ Pillet: 
The Green–Kubo Formula for Locally Interacting Fermionic Open Systems.
Annales Henri Poincar\'{e} 8:1013--1036 (2007).


\bibitem{Ka}  T.\ Kato:  On the adiabatic theorem of
quantum mechanics. Phys.\ Soc.\ Jap.\ 5:435--439
(1950).

\bibitem{KSV}
R.D.\ King-Smith  and D.\ Vanderbilt:
Theory of polarization of crystalline solids. 
Physical Review B 47:1651 (1993).

\bibitem{KS}
M.\ Klein and R.\ Seiler: Power-law corrections to the Kubo formula vanish in quantum Hall systems. Communications in Mathematical Physics 128:141--160 (1990).

\bibitem{KLRVZ}
Y.E.\ Kraus, Y.\ Lahini, Z.\ Ringel, M.\ Verbin, and O.\ Zilberberg:
Topological states and adiabatic pumping in quasicrystals.
Physical Review Letters 109:106402 (2012).

\bibitem{KWKH}
K.\ Kudo, H.\ Watanabe, T.\ Kariyado, and Y.\ Hatsugai:
Many-body Chern number without integration.
Preprint available at 
\href{http://arxiv.org/abs/1808.10248}{\texttt{arXiv:1808.10248}} (2018).


\bibitem{L}
R.\ Laughlin:  Quantized Hall conductivity in two dimensions. Physical Review B 23:5632  (1981).

  \bibitem{LR}
E.\  Lieb and D.\ Robinson:
 The finite group velocity of quantum spin systems.
  Communications in Mathematical Physics 28:251--257 (1972).
  
\bibitem{LSZAB}
M.\ Lohse, C.\ Schweizer, O.\ Zilberberg, M.\ Aidelsburger, and I.\ Bloch:
A Thouless quantum pump with ultracold bosonic atoms in an optical superlattice.
Nature Physics 12:350--354 (2016).


\bibitem{MMPT}
G.\ Marcelli, D.\ Monaco, G.\ Panati, and S.\ Teufel: Quantum (spin) Hall conductivity: Kubo-like formula (and beyond). In preparation.

\bibitem{NSW}
B.\ Nachtergaele, V.\ Scholz, and R.\ Werner:
Local approximation of observables and commutator bounds.
Operator Methods in Mathematical Physics. Springer Basel, 2013. 143--149.

\bibitem{NSY} 
B.\ Nachtergaele,  R.\ Sims, and A.\ Young:
Lieb--Robinson bounds, the spectral flow, and stability of the spectral gap for lattice fermion systems. 
To appear in: F.\ Bonetto, D.\ Borthwick, E.\ Harrell, and M.\ Loss (eds.), Mathematical Problems in Quantum Physics. Proceedings of the conference QMATH13, Atlanta, October 8-11, 2016. Contemporary Mathematics 717, AMS (2018).
Preprint available at 
\href{http://arxiv.org/abs/1705.08553}{\texttt{arXiv:1705.08553}} (2017).

\bibitem{NTTIOWTT}
Sh.\ Nakajima, T.\ Tomita, Sh.\ Taie, T.\ Ichinose, H.\ Ozawa, L.\ Wang, M.\ Troyer, and Y.\ Takahashi:
Topological Thouless pumping of ultracold fermions.
Nature Physics 12:296--300 (2016).
  
\bibitem{Nenciu}
G.\ Nenciu:
Linear adiabatic theory. Exponential estimates. 
Communications in Mathematical Physics 152:479--496 (1993).
  
\bibitem{NT}
Q.\ Niu and D.J.\ Thouless:
 Quantised adiabatic charge transport in the presence of substrate disorder and many-body interaction.
 Journal of Physics A: Mathematical and General 17:2453 (1984).
 
 \bibitem{NTW}
Q.\ Niu, D.J.\ Thouless, and Y.-Sh.\ Wu:
 Quantized Hall conductance as a topological invariant.
 Physical Review B 31:3372 (1985).
 
\bibitem{PST}
G.\ Panati, C.\  Sparber, and S.\ Teufel:
Geometric currents in piezoelectricity.
 Archive for Rational Mechanics and Analysis 191:387  (2009).

\bibitem{PST2}
G.\ Panati, H.\ Spohn, and S.\ Teufel: Space-adiabatic perturbation theory. Advances in Theoretical and Mathematical Physics 7: 145--204 (2003).

\bibitem{PST3} 
G.\ Panati, H.\ Spohn, and S.\ Teufel: 
Effective dynamics for Bloch electrons: Peierls substitution and beyond. Communications in Mathematical Physics 242:547--578 (2003).

\bibitem{ST}   
H.\ Schulz-Baldes  and S.\ Teufel:
Orbital Polarization and Magnetization for Independent Particles in Disordered Media.
 Communications in Mathematical Physics 319:649--681 (2013).
   
\bibitem{T}
S.\ Teufel:
Adiabatic Perturbation Theory in Quantum Dynamics.
Lecture Notes in Mathematics 1821,   Springer    (2003).

\bibitem{T2}
S.\ Teufel:
Non-equilibrium almost-stationary states for interacting electrons on a lattice.
Preprint available at
\href{http://arxiv.org/abs/1708.03581}{\texttt{arXiv:1708.03581}} (2017).

\bibitem{ThPump}
D.J.\ Thouless:
 Quantization of particle transport. 
 Physical Review B 27:6083 (1983).
 
\bibitem{Th}
D.J.\ Thouless:
 Topological quantum numbers in nonrelativistic physics. 
 World Scientific (1998).
 
\end{thebibliography}
\end{document}